\PassOptionsToPackage{dvipsnames,svgnames}{xcolor}

\documentclass[acmsmall,screen,authorversion=true, nonacm=true]{acmart}

\setcopyright{acmlicensed}
\copyrightyear{2025}
\acmYear{2025}
\acmDOI{XXXXXXX.XXXXXXX}

\acmJournal{PACMCGIT}
\acmVolume{1}
\acmNumber{1}
\acmArticle{1}
\acmMonth{1}

\usepackage{mathtools}
\usepackage{bm}
\usepackage{cases}
\usepackage{comment}
\usepackage{threeparttable}
\usepackage{subcaption}
\usepackage{wrapfig}
\usepackage{enumitem}
\usepackage{overpic}
\usepackage{siunitx}
\usepackage{makecell}
\usepackage{multirow}
\usepackage{colortbl}
\usepackage{float}

\usepackage{tikz}
\usetikzlibrary{decorations.fractals, spy}
\usepackage{tikz-3dplot}
\usepackage{tikz-dimline}
\usepackage{pgfplots}
\pgfplotsset{compat=1.18}

\newcommand{\del}{\bm{\nabla}}

\newlength{\imagewidth}

\begin{document}

\title{Representing Flow Fields with Divergence-Free Kernels for Reconstruction}

\author{Xingyu Ni}
\email{nixy@pku.edu.cn}
\orcid{0000-0003-1127-2848}
\affiliation{
\institution{School of Computer Science, Peking University}
\city{Beijing}
\country{China}
}

\author{Jingrui Xing}
\email{xjr01@hotmail.com}
\orcid{0000-0001-7219-9969}
\affiliation{
\institution{School of Intelligence Science and Technology, Peking University}
\city{Beijing}
\country{China}
}

\author{Xingqiao Li}
\email{lixingqiao@pku.edu.cn}
\orcid{0000-0002-8131-6140}
\affiliation{
\institution{School of Intelligence Science and Technology, Peking University}
\city{Beijing}
\country{China}
}

\author{Bin Wang}
\email{binwangbuaa@gmail.com}
\orcid{0000-0001-9496-772X}
\affiliation{
\institution{Independent}
\city{Beijing}
\country{China}
}
\authornote{corresponding authors}

\author{Baoquan Chen}
\email{baoquan@pku.edu.cn}
\orcid{0000-0003-4702-036X}
\affiliation{
\institution{State Key Laboratory of General Artificial Intelligence, Peking University}
\city{Beijing}
\country{China}
}
\authornotemark[1]

\begin{abstract}
Accurately reconstructing continuous flow fields from sparse or indirect measurements remains an open challenge, as existing techniques often suffer from oversmoothing artifacts, reliance on heterogeneous architectures, and the computational burden of enforcing physics-informed losses in implicit neural representations (INRs).
In this paper, we introduce a novel flow field reconstruction framework based on divergence-free kernels (DFKs), which inherently enforce incompressibility while capturing fine structures without relying on hierarchical or heterogeneous representations.
Through qualitative analysis and quantitative ablation studies, we identify the matrix-valued radial basis functions derived from Wendland's $\mathcal{C}^4$ polynomial (DFKs-Wen4) as the optimal form of analytically divergence-free approximation for velocity fields, owing to their favorable numerical properties,
including compact support, positive definiteness, and second-order differentiablility.
Experiments across various reconstruction tasks, spanning data compression, inpainting, super-resolution, and time-continuous flow inference, has demonstrated that DFKs-Wen4 outperform INRs and other divergence-free representations in both reconstruction accuracy and computational efficiency while requiring the fewest trainable parameters.
\end{abstract}

%
%



\begin{teaserfigure}
  \centering
  \newcommand{\formattedgraphics}[2]{\begin{overpic}[height=.14\linewidth,angle=-90]{#1}\put(1.6,95){\sffamily\scriptsize #2}\end{overpic}}
  \formattedgraphics{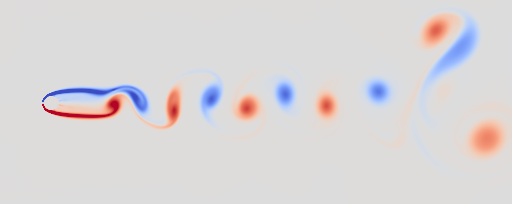}{Ground Truth}%
  \hfill
  \formattedgraphics{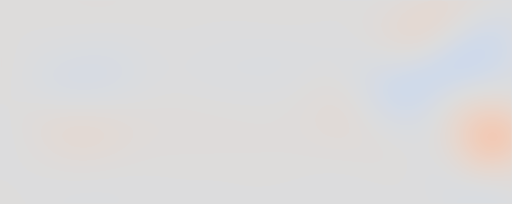}{SIREN}%
  \hfill
  \formattedgraphics{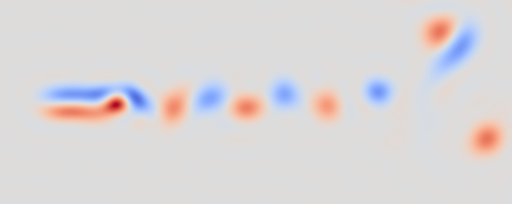}{Curl SIREN}%
  \hfill
  \formattedgraphics{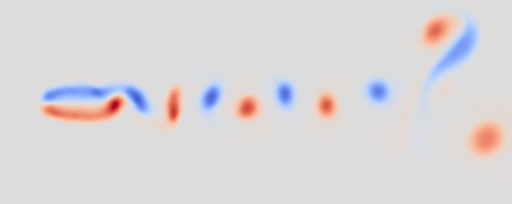}{Regular RBF}%
  \hfill
  \formattedgraphics{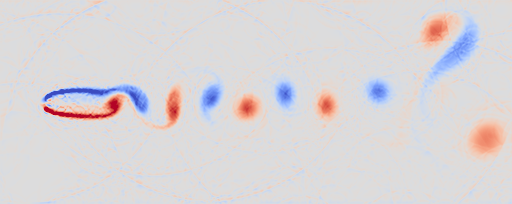}{DFK-Poly6}%
  \hfill
  \formattedgraphics{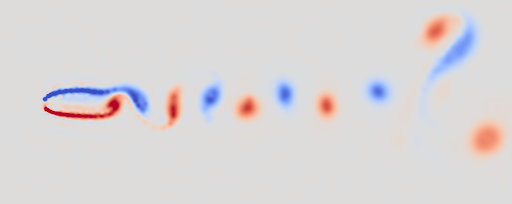}{Curl Kernel}%
  \hfill
  \formattedgraphics{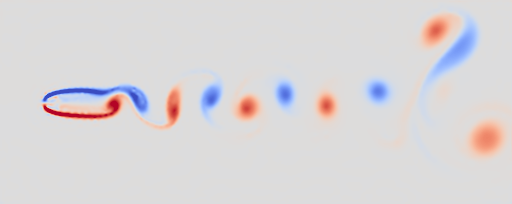}{\bfseries DFK-Wen4}%
  \\
  \vspace{-1em}
  \caption{Fitting experiments of the K\'{a}rm\'{a}n vortex street, with the resulting vorticity fields illustrated.
  We can optimize divergence-free kernels to compressly store the incompressible flow field data.
  The PSNR/SSIM values for each method are as follows:
  SIREN: 22.46/0.897; Curl SIREN: 28.17/0.936; Regular RBF: 30.90/0.967; DFK-Poly6: 32.80/0.824; Curl Kernel: 36.56/0.979; \textbf{DFK-Wen4:} \textbf{38.78}/\textbf{0.990}.}
  \label{fig:karman}
  \Description{Karman experiments}
\end{teaserfigure}

\maketitle

\section{Introduction}

Reconstructing high-fidelity continuous flow fields from sparse, incomplete, or indirect data is crucial across various scientific and engineering domains, including meteorology, biomedicine, hydraulic engineering, automotive manufacturing, and visual effects.
While traditional grid- and particle-based representations have been highly successful in forward simulations of fluid dynamics \cite{Bridson2015}, 
their resolution-dependent nature and inherently discrete signals pose significant challenges for optimization, making them less effective for inverse reconstruction tasks that require optimizing a separate model for each flow field instance.
This limitation has driven recent research toward implicit neural representations (INRs), which offer a resolution-agnostic framework for modeling continuous and differentiable physical fields.

However, neural network-based representations also face two major challenges. First, despite advancements such as positional encoding \cite{Mildenhall2021} and periodic activation functions \cite{Sitzmann2020}, INRs tend to oversmooth physical fields, leading to the loss of fine-scale details. In the context of flow field inference from multi-view RGB videos, the state-of-the-art approach \cite{Yu2023} leverages multi-resolution hash encoding \cite{Muller2022} to enhance network expressiveness and introduce vortex particles to recover missing fluid structures. This highly heterogeneous representation compromises robustness and increases the complexity of both implementation and further improvements. Second, physics-informed neural networks (PINNs) \cite{Raissi2019,Raissi2020}, which incorporate physical constraints as loss terms, suffer from optimization difficulties due to their heavy reliance on balancing different partial differential equation (PDE) losses. Striking an appropriate trade-off between fitting observed data and enforcing physical laws is nontrivial \cite{Wang2024,Chu2022,Gao2021,Wang2020}, often resulting in suboptimal or even infeasible solutions. Moreover, the need to compute high-order derivatives within the neural network significantly increases computational costs, as it leads to prohibitively large computational graphs,  limiting the applicability of these methods in time-sensitive scenarios.

In this paper, we propose a novel scheme for reconstructing flow fields based on divergence-free kernels (DFKs), which obviates the need for heterogeneous or hierarchical representations and eliminates incompressibility-related penalty terms in the optimization process. Although divergence-free kernels have been studied in numerical analysis and interpolation, their application to large-scale flow field reconstruction---especially in a optimization-based setting---has not been thoroughly explored.
At the core of this framework lies a set of matrix-valued radial basis functions (RBFs) \cite{Narcowich1994}, derived by applying differential operators to Wendland’s $\mathcal{C}^4$ polynomial \cite{Wendland1995}. These kernels, referred to as DFKs-Wen4 for brevity, serve as the numerical basis for flow field representation. The function space spanned by DFKs-Wen4 inherently satisfies the continuity equation for incompressible flows, rigorously enforcing the divergence-free property of velocity fields by construction. Additionally, DFKs-Wen4 possess critical properties such as compact support, positive definiteness, and second-order differentiability, while exhibiting strong alignment with fundamental solutions for fluid flows around obstacles. 
Moreover, unlike INRs, our kernel-based representataion offers greater modeling flexibility, as the position, radius, and weight of each kernel can be explicitly optimized to better resolve multiscale and complex fluid structures.
Comparative evaluations against alternative kernel-based methods and implicit neural representations (INRs) reveal that DFKs-Wen4 achieve superior performance across diverse objective functions and varying levels of data sparsity, all while maintaining a minimal number of trainable parameters. These findings underscore the potential of DFKs-Wen4 as a compelling alternative to neural network-based approaches for fluid reconstruction tasks, offering both theoretical rigor and computational efficiency.

Our experiments focus on the following flow field reconstruction tasks: (1) fitting dense velocity data to reduce memory consumption for flow field storage; (2) fitting dense but non-divergence-free data to perform pressure projection via the Helmholtz decomposition inherent in the representation; (3) fitting dense velocity data with missing regions, leveraging the representation’s generalization ability for inpainting; (4) fitting sparse velocity data, utilizing the representation’s generalization capability for super-resolution; and (5) inferring a time-continuous velocity field from dynamic, dense passive scalar data based on the advection equation.
It is important to note that our goal is not to achieve state-of-the-art performance in any specific application but rather to systematically compare different representations in terms of their expressiveness and optimization efficiency under controlled conditions. 
To isolate the impact of representation quality on reconstruction results, we avoid introducing confounding factors from auxiliary techniques in each case for both our method and the comparison methods. Instead, we optimize the parameters of the representations to fit the observed data with the continuity equation of fluid dynamics (i.e., the incompressibility condition) as the primary physical prior, which is either enforced as a soft constraint or directly embedded in the search space.

The technical contributions are summarized as follows:
\begin{itemize}
 \item A systematic practical framework for applying DFKs-Wen4 to flow field reconstruction,
 \item Qualitative analysis and quantitative ablation studies to identify the optimal form of the divergence-free kernels for fluid reconstruction tasks, and
 \item Comprehensive benchmarking against state-of-the-art INRs across diverse application scenarios.
\end{itemize}

\section{Related Work}

\paragraph{Kernel-based modeling}
The use of kernel functions to represent flow fields in forward fluid simulation has a long-standing history. Compared to traditional approaches such as finite element methods (FEMs) and finite difference methods (FDMs), kernel-based modeling offers simpler formulations and easier implementation, making it particularly well-suited for simulating fluids---materials characterized by highly dynamic topologies and intricate local details. Specifically, kernel functions have been employed to smooth the properties of neighboring particles (e.g., velocity, density, and pressure) \cite{Muller2003,Yu2013,Bender2015}, interpolate fluid quantities (such as velocity and density) onto grids \cite{Foster2001,Canabal2016,Chang2022}, and facilitate smooth interactions between particles and grid-based solvers \cite{Zhu2005,Jiang2015,Hu2018}. In these studies, kernels are typically designed to satisfy desirable mathematical properties, such as smoothness, compact support, and positive definiteness.

In recent years, kernel-based representations have also gained significant attention in the context of physical field reconstruction.
Particularly, for the reconstruction of 3D radiance fields, kernel-based methods---such as those used in Point-NeRF \cite{Xu2022} and 3D Gaussian Splatting (3DGS) \cite{Kerbl2023}---have demonstrated advantages over vanilla Neural Radiance Fields (NeRFs) \cite{Mildenhall2021}, offering improved optimization efficiency and faster inference speeds.
Inspired by these advancements, we explore the integration of kernel-based modeling into the reconstruction of flow fields, with the aim of achieving similar superior flexibility and performance compared to neural network-based methods (i.e., INRs).


\paragraph{Divergence-free representations}
As the continuity condition for incompressible flows, the divergence-free constraint on velocity fields plays a crucial role in their modeling.
In the field of computer graphics, recent works have pioneered fluid simulations that directly utilize divergence-free representations rather than relying on pressure projection or other post-processing techniques.
Some methods achieve this by maintaining a vector potential on grids and computing its curl to obtain a divergence-free flow field \cite{Chang2022,Lyu2024}, while others enforce divergence-free interpolation schemes directly on grids \cite{Nabizadeh2024}.
In the context of data-driven methods, \citet{Kim2019} encoded fluid simulation results---expressed using aforementioned velocity potentials---into convolutional neural networks (CNNs), while \citet{Richter2024} proposed neural network parameterizations that inherently satisfy the divergence-free condition in arbitrary dimensions, using flow field modeling as a key application scenario.

The divergence-free kernels discussed in this paper originate from research on matrix-valued RBF interpolation \cite{Narcowich1994}, with \citet{Lowitzsch2005} presenting a formulation closest to our adopted DFKs-Wen4.
Over the years, various DFKs have been employed to represent flow fields \cite{Wendland2009, Li2020, Fuselier2016, Skrinjar2009}, electromagnetic fields \cite{McNally2011}, and elastic potential fields \cite{HiOInterp22}. The idea of using DFKs for flow field reconstruction is also evident in the works of \citet{Macedo2010} and \citet{Zhou2019}, where support vector regression and RBF interpolation were respectively used for flow field fitting.
However, the integration of DFK representations with advanced optimizers from the deep learning field (e.g., Adam \cite{Kingma2017}, see \S\ref{sec:implementation}) and  cutting-edge visual reconstruction techniques (e.g., NeRF \cite{Mildenhall2021}, see \S\ref{sec:inference-videos}) for large-scale, even passive field-based, flow field reconstruction remains largely unexplored.



\paragraph{Flow field reconstruction}
Here we briefly review previous optimization-based methods for flow field reconstruction tasks discussed in this paper.

For divergence-free projection, \citet{Chen2023} introduced a neural solution to transient differential equations using implicit neural spatial representation (INSR) \cite{Xie2022}, though it faces efficiency challenges due to high-order derivatives.

For super-resolution (SR), \citet{Fukami2019} applied data-driven techniques to reconstruct low-resolution flow images for a 2D cylinder wake, and similar methods using CNNs and GANs have enhanced turbulence, plumes, and channel flows \cite{Deng2019, Xie2018, Werhahn2019, Liu2020}. 
Since 2020, PINNs \cite{Raissi2020} have been integrated to incorporate underlying physical laws, improving reliability and reducing reliance on high-resolution (HR) data. For example, \citet{Wang2020} utilized a physics-informed SR technique to reconstruct HR images in an advection-diffusion model of atmospheric pollution plumes, while \citet{Gao2021} developed a physics-constrained CNN for SR of vascular flow without labeled data.

For dynamic flow field inference, traditional approaches rely on specialized hardware \cite{Atcheson2008, Ji2013} or particle imaging velocimetry (PIV) \cite{Xiong2017}, which tracks passive markers in the flow. Tomographic methods have also been explored \cite{Gregson2014}. More recently, RGB-video-based techniques have reduced dependence on specialized setups. ScalarFlow \cite{Eckert2019} introduced long-term temporal physics constraints by optimizing residuals between reconstructed and simulated density and velocity, while \citet{Franz2021} used differentiable rendering for end-to-end optimization.
\citet{Deng2023} proposed vortex particles to predict 2D fluid motion in videos. Emerging methods combining PINNs and NeRFs, such as PINF \cite{Chu2022}, HyFluid \cite{Yu2023}, and PICT \cite{Wang2024}, enforce physics constraints via soft regularization.

Furthermore, 
optimization-based methods for incompressible flow editing and inpainting formulate interpolation as an energy minimization problem, enforcing incompressibility constraints \cite{Bhatacharya2012, Nielsen2011, Sato2018, Ozdemir2024}, 
in which \citet{Schweri2021} utilized a physics-awared neural network to inpaint missing flow data from satellite observations.

\begin{figure}[t]
  \centering%
  \newcommand{\formattedgraphics}[2]{\begin{overpic}[width=.138\linewidth,trim=6cm 6cm 6cm 6cm,clip]{#1}\put(3,89){\sffamily\scriptsize\color{white} #2}\end{overpic}}%
  \formattedgraphics{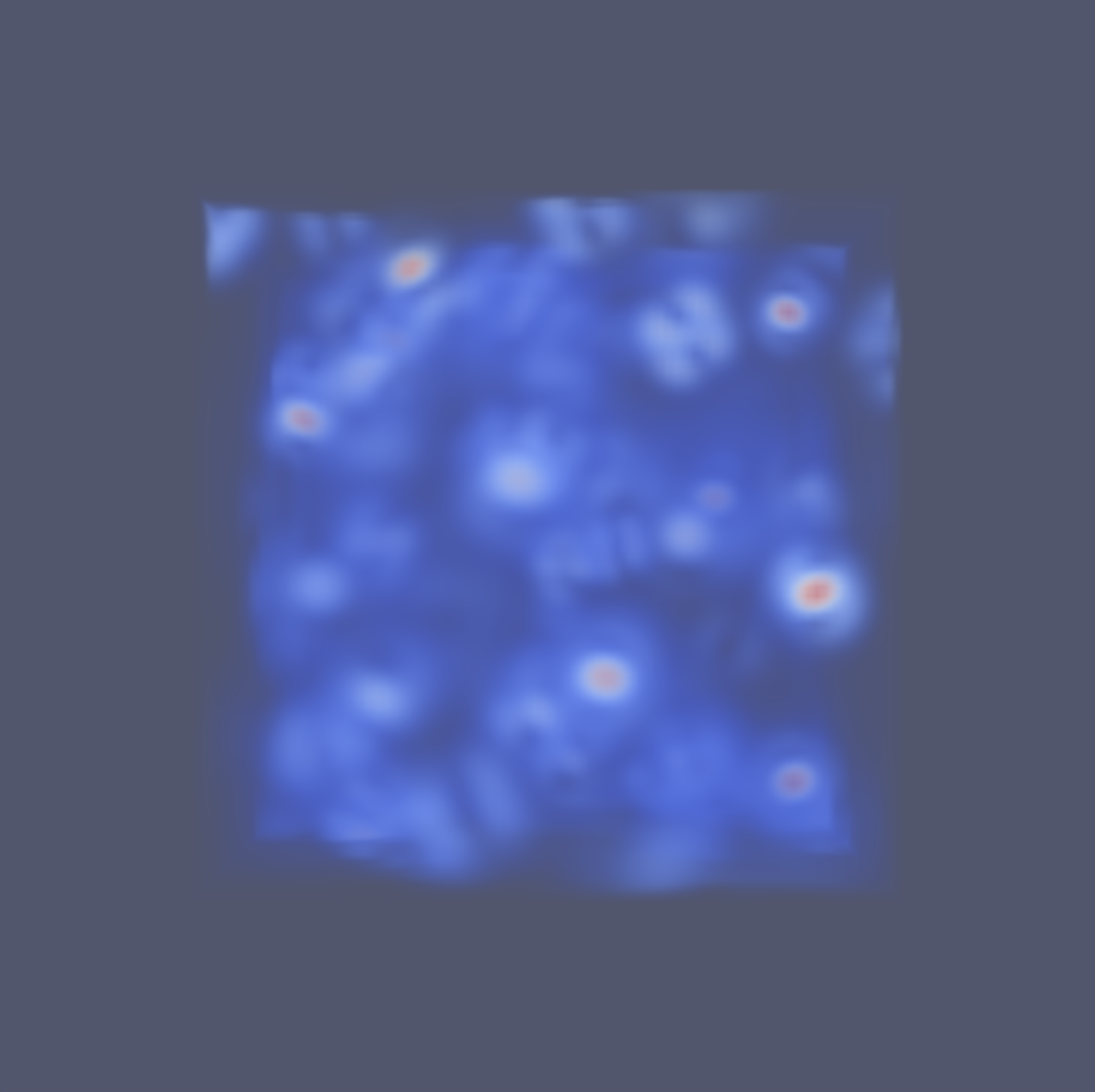}{Ground Truth}%
  \hfill%
  \formattedgraphics{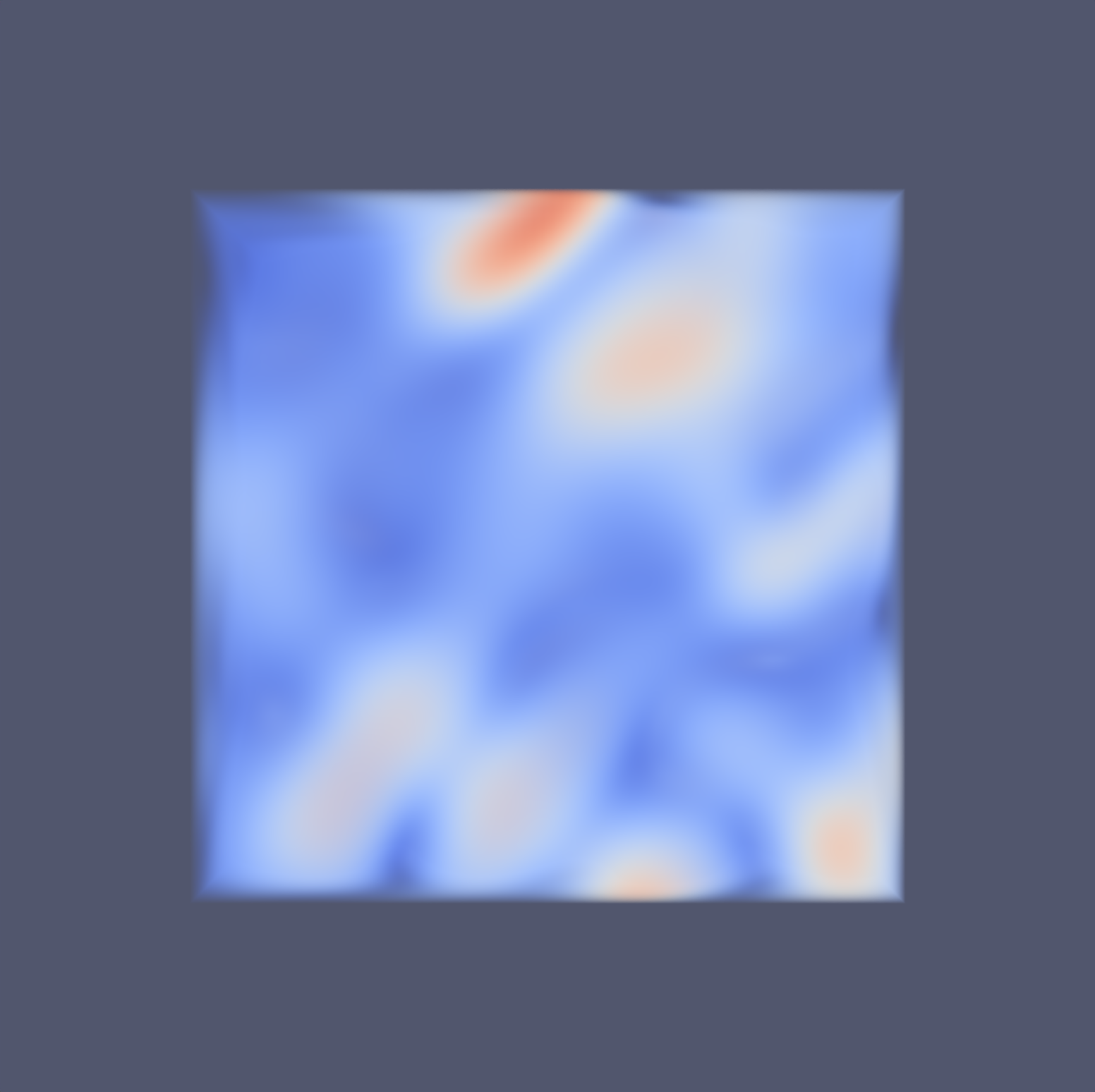}{SIREN}%
  \hfill%
  \formattedgraphics{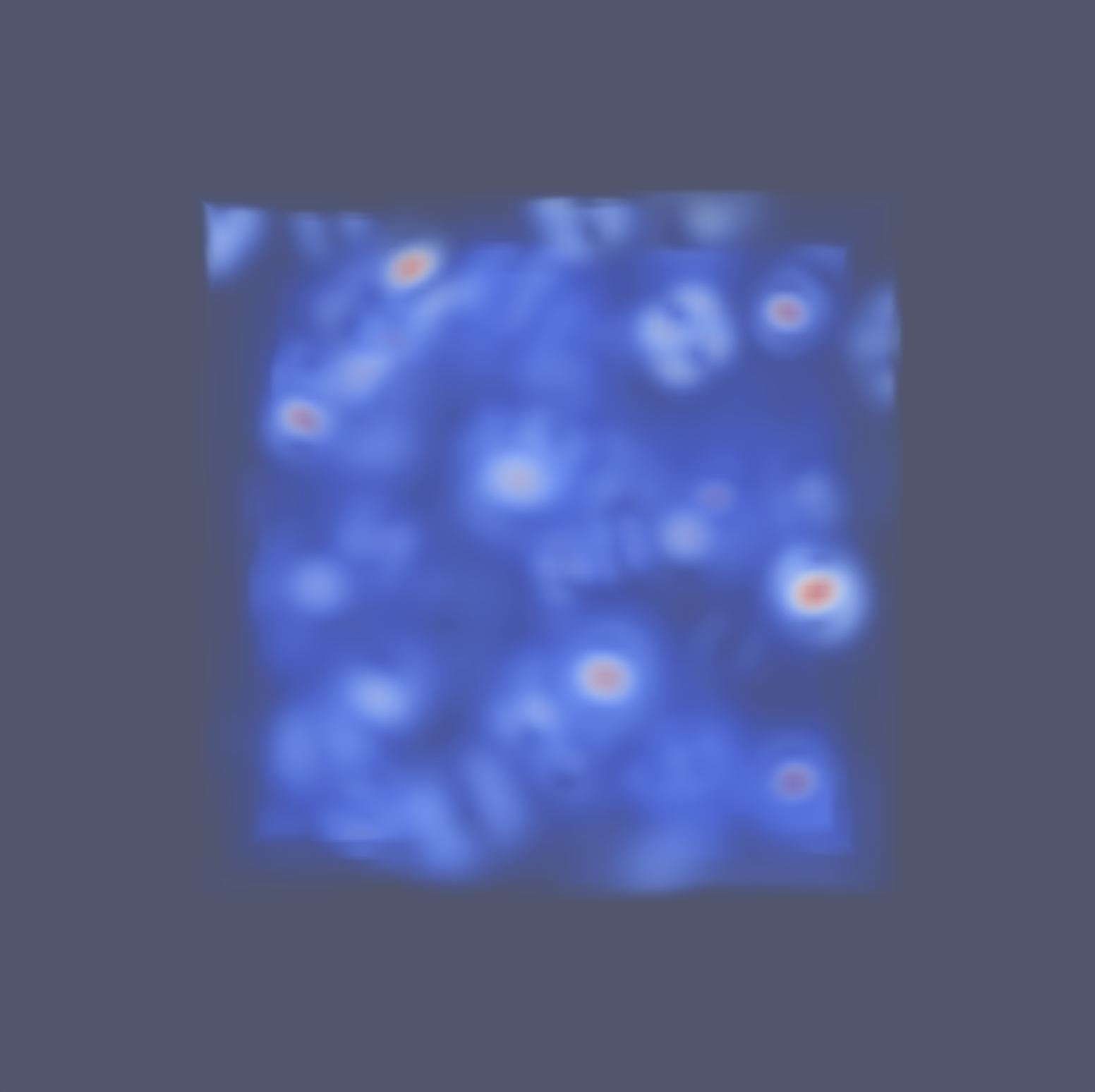}{Curl SIREN}%
  \hfill%
  \formattedgraphics{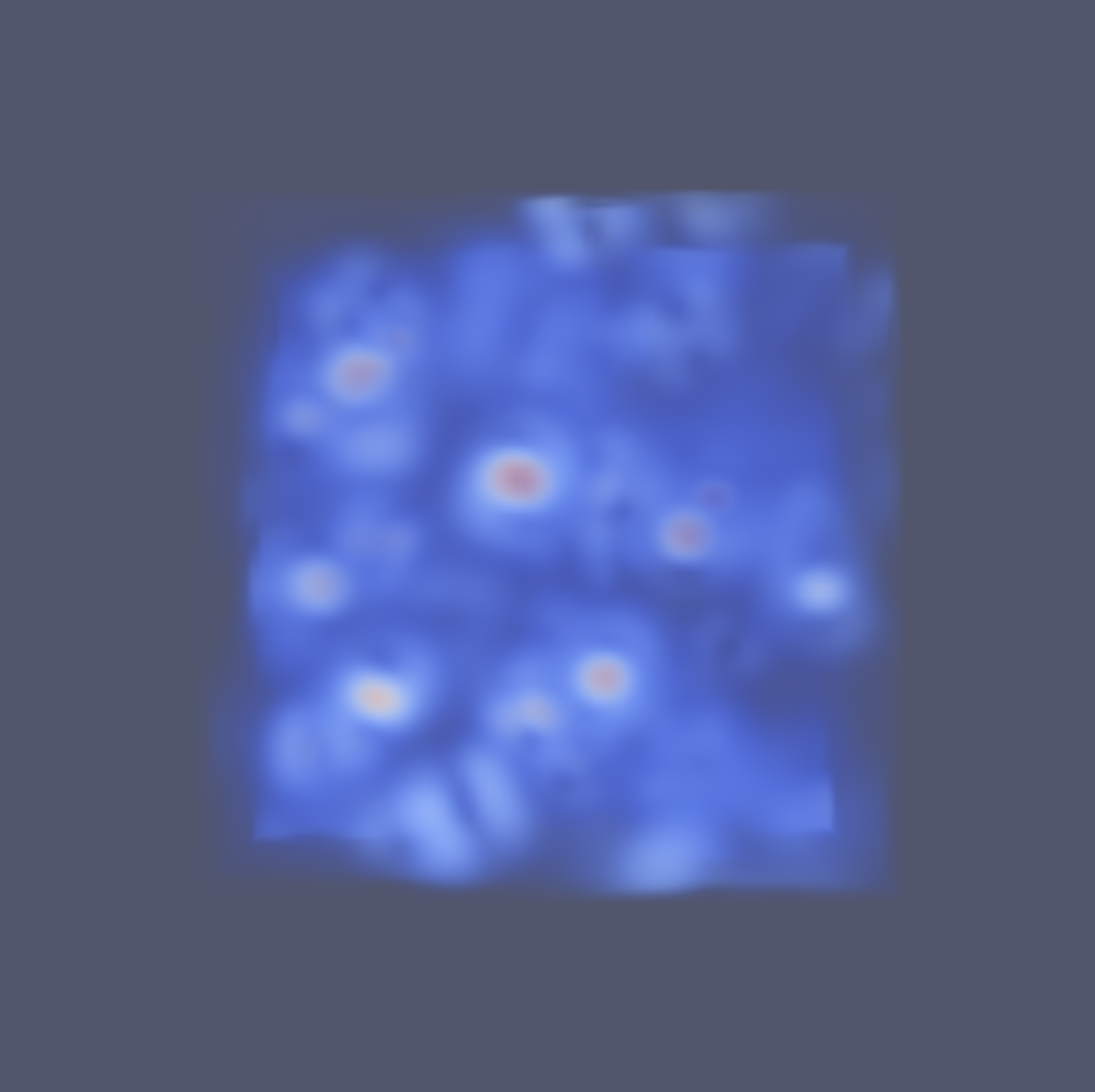}{Regular RBF}%
  \hfill%
  \formattedgraphics{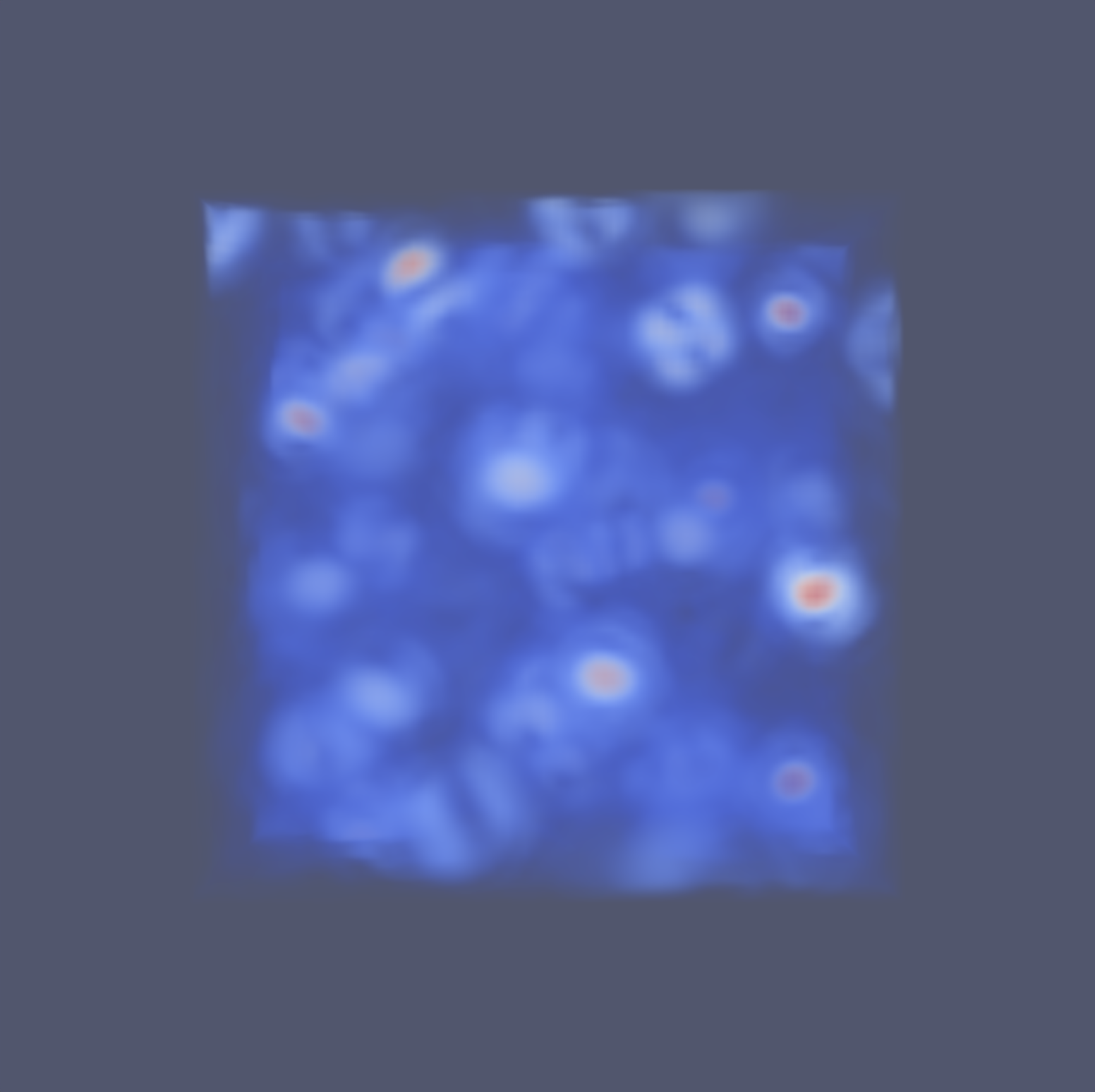}{DFK-Poly6}%
  \hfill%
  \formattedgraphics{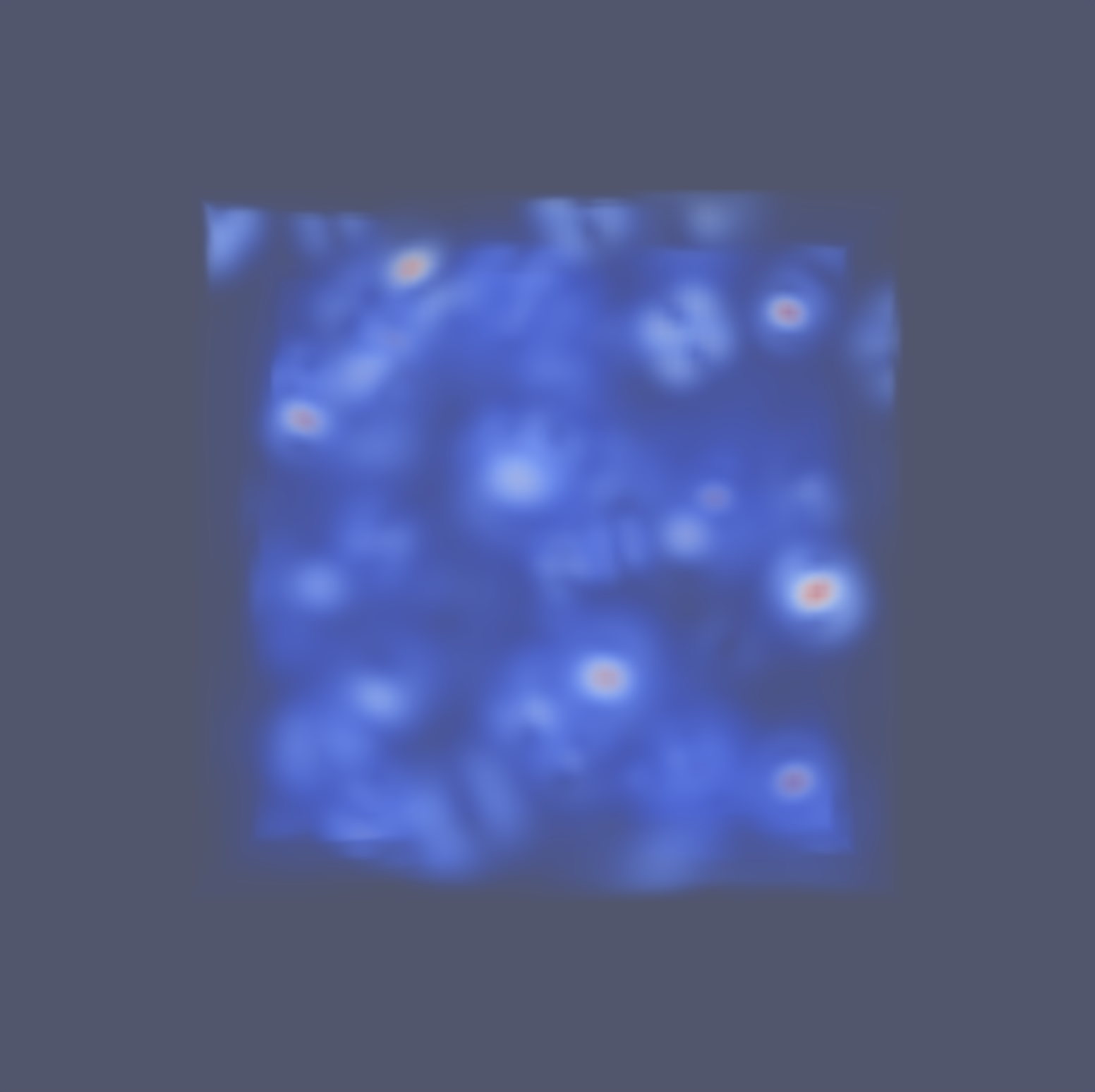}{Curl Kernel}%
  \hfill%
  \formattedgraphics{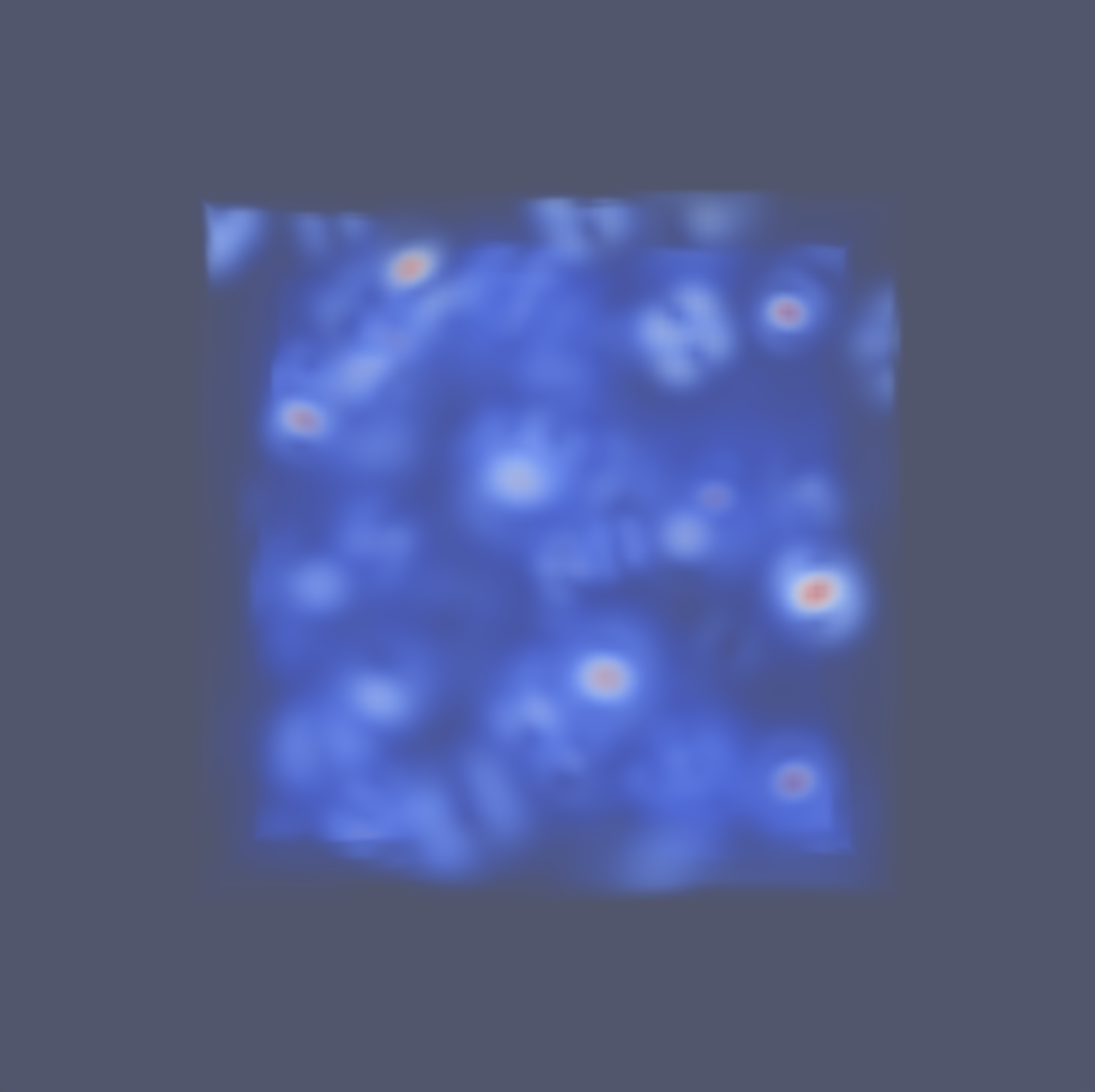}{\bfseries DFK-Wen4}%
  \\
  \vspace{-1em}
  \caption{Fitting experiments of the analytic vortices, with the resulting vorticity fields rendered. Physics-embedded methods clearly demonstrate superior fitting capabilities over physics-informed ones. Additionally, kernel-based approaches excel at capturing local details compared to neural networks-base representrations. Among all tested approaches, DFK-Wen4 achieves the lowest fitting error, as shown in Tab.~\ref{tab:losses}.}
  \label{fig:analytic}
  \Description{analytic experiments}
\end{figure}

\section{Problem Statement}

From a continuous perspective, the flow field reconstruction tasks we consider can be unified into the following constrained optimization problem:
\begin{align}
  \mathop{\arg\min}_{\bm{u}(\bm{x},t)\in\mathcal{F}}\quad&\mathcal{L}_\mathrm{obs}\left[\bm{u}(\bm{x},t),f_\mathrm{obs}(\bm{x},t)\right]\text{,}\\
  \text{subject to}\quad&\del\cdot\bm{u}=0\text{,\quad}\bm{x}\in\Omega\text{,}\label{eqn:div-free}\\
  \text{and}\quad & \bm{u}=\bm{u}_\mathrm{s}\text{,\qquad}\bm{x}\in\partial\Omega\text{,}\label{eqn:no-slip}
\end{align}
where $\Omega$ is the fluid domain, and $\partial\Omega$ is its boundary.
The search space $\mathcal{F}=L^2(\Omega\times[0,T],\mathbb{R}^d)$ consists of $d$-dimensional square-integrable vector functions defined over the spatiotemporal domain.

The objective function $\mathcal{L}_\mathrm{obs}$, which quantifies the observational loss, depends on both the flow field $\bm{u}$, to be optimized, and an observed input field $f_\mathrm{obs}$.
When the observations directly correspond to the flow field, though possibly incomplete within $\Omega$, the objective function is given by
\begin{equation}
  \label{eqn:fitting_loss}
  \mathcal{L}_\mathrm{obs}[\bm{u},\bm{u}_\mathrm{D}]=\frac{1}{V_\mathrm{D}}\int_0^T\int_{\Omega_\mathrm{D}}\Vert\bm{u}-\bm{u}_\mathrm{D}\Vert\,\mathrm{d}V_\mathrm{D}\,\mathrm{d}t\text{,}
\end{equation}
where $\Omega_\mathrm{D}$ is the supervised region with volume $V_\mathrm{D}$.
Alternatively, if the observations originate from a passive field, such as soot concentration $\sigma$, the advection equation $\partial\sigma/\partial t+\bm{u}\cdot\del\sigma=0$ must be incorporated into the objective function, yielding
\begin{equation}
  \label{eqn:advection_loss}
  \mathcal{L}_\mathrm{obs}[\bm{u},\sigma]=\frac{1}{V}\int_0^T\int_{\Omega}\left\Vert\frac{\partial\sigma}{\partial t}+\bm{u}\cdot\del\sigma\right\Vert\mathrm{d}V\,\mathrm{d}t\text{.}
\end{equation}
The constraints given in
Eqs. \eqref{eqn:div-free} and \eqref{eqn:no-slip} arise from the continuity equation and the no-slip boundary condition for incompressible flows, respectively, where the solid velocity $\bm{u}_\mathrm{s}$ is assumed to be zero unless specified otherwise.

\paragraph{Physics-informed losses}
Conventionally, the divergence-free and boundary conditions are \emph{relaxed} by incorporating additional penalty terms $\mathcal{L}_\mathrm{div}$ and $\mathcal{L}_\mathrm{bou}$, defined as
\begin{align}
  \mathcal{L}_\mathrm{div}&=\frac{1}{V}\int_0^T\int_\Omega\Vert\del\cdot\bm{u}\Vert\,\mathrm{d}V\,\mathrm{d}t\text{,}\\
  \mathcal{L}_\mathrm{bou}&=\frac{1}{A}\int_0^T\int_{\partial\Omega}\Vert\bm{u}-\bm{u}_\mathrm{s}\Vert\,\mathrm{d}A\,\mathrm{d}t\text{,}\label{eqn:boundary_loss}
\end{align}
where $A$ denotes the boundary area of $\Omega$.
The resulting optimization problem is then formulated as
\begin{equation}
  \mathop{\arg\min}_{\bm{u}(\bm{x},t)\in\mathcal{F}}\quad\mathcal{L}=\mathcal{L}_\mathrm{obs}+\lambda_\mathrm{div}\mathcal{L}_\mathrm{div}+\lambda_\mathrm{bou}\mathcal{L}_\mathrm{bou}\text{,}
\end{equation}
where $\lambda_\mathrm{div}$ and $\lambda_\mathrm{bou}$ represent tunable weighting factors. The total objective function $\mathcal{L}$ aligns with loss formulation used in previous work \cite{Chu2022,Wang2024,Yu2023}.

\paragraph{Physics-embedded approaches} 
In contrast, we demonstrate that the divergence-free condition, $\del\cdot\bm{u}=0$
can be more effectively enforced by embedding it directly into the search space. By employing divergence-free representations, the search space is restricted to valid solutions that inherently satisfy the continuity equation for incompressible flows, ensuring that the value of $\mathcal{L}_\mathrm{div}$ remains identically zero.
This reformulates the optimization problem as
\begin{equation}
  \mathop{\arg\min}_{\bm{u}(\bm{x},t)\in\mathcal{P}[\mathcal{F}]}\quad\mathcal{L}=\mathcal{L}_\mathrm{obs}+\lambda_\mathrm{bou}\mathcal{L}_\mathrm{bou}\text{,}
\end{equation}
in which $\mathcal{P}[\mathcal{F}]$ denotes the reduced search space of divergence-free fields.
This embedding approach guarantees that the reconstructed flow fields are both data-driven and physically realistic, leading to more accurate results.

\paragraph{Loss discretization}
In practice, observations are typically provided as discrete sample points. When these data points are uniformly distributed within $\Omega_\mathrm{D}$ (or $\partial\Omega_\mathrm{D}$), the spatial integrals divided by $V$ (or $A$) in Eqs. (\ref{eqn:fitting_loss}--\ref{eqn:boundary_loss}) can generally be approximated by averaging over the data points.
Similarly, the outer temporal integral can be expressed as a summation over discrete time steps, and the time derivatives of physical quantities are computed using finite difference schemes.

\begin{figure}[t]
  \centering%
  \newcommand{\formattedgraphics}[2]{\begin{overpic}[width=.138\linewidth,trim=10cm 8cm 10cm 12cm,clip]{#1}\put(3,89){\sffamily\scriptsize\color{white} #2}\end{overpic}}%
  \formattedgraphics{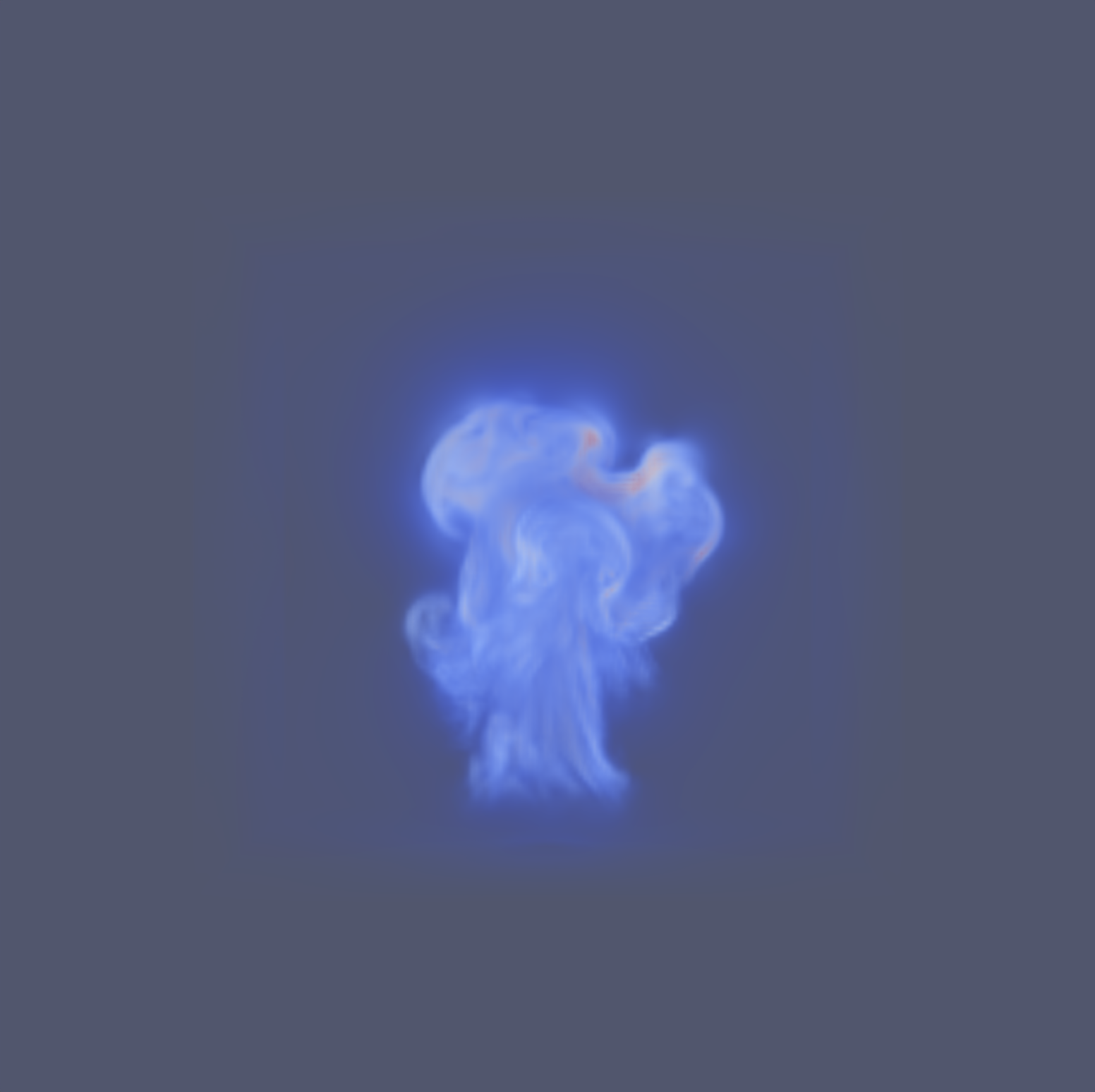}{Ground Truth}%
  \hfill%
  \formattedgraphics{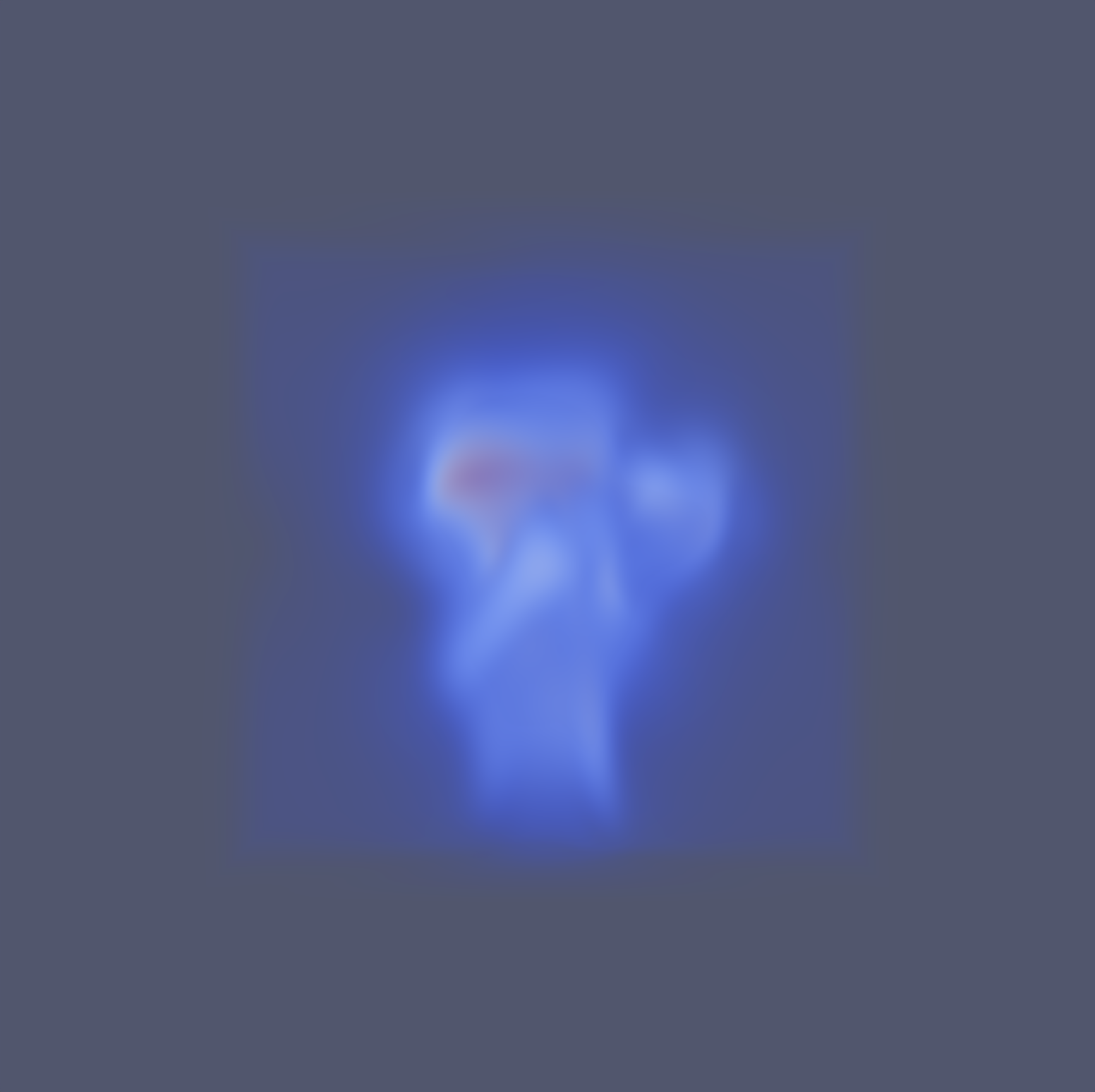}{SIREN}%
  \hfill%
  \formattedgraphics{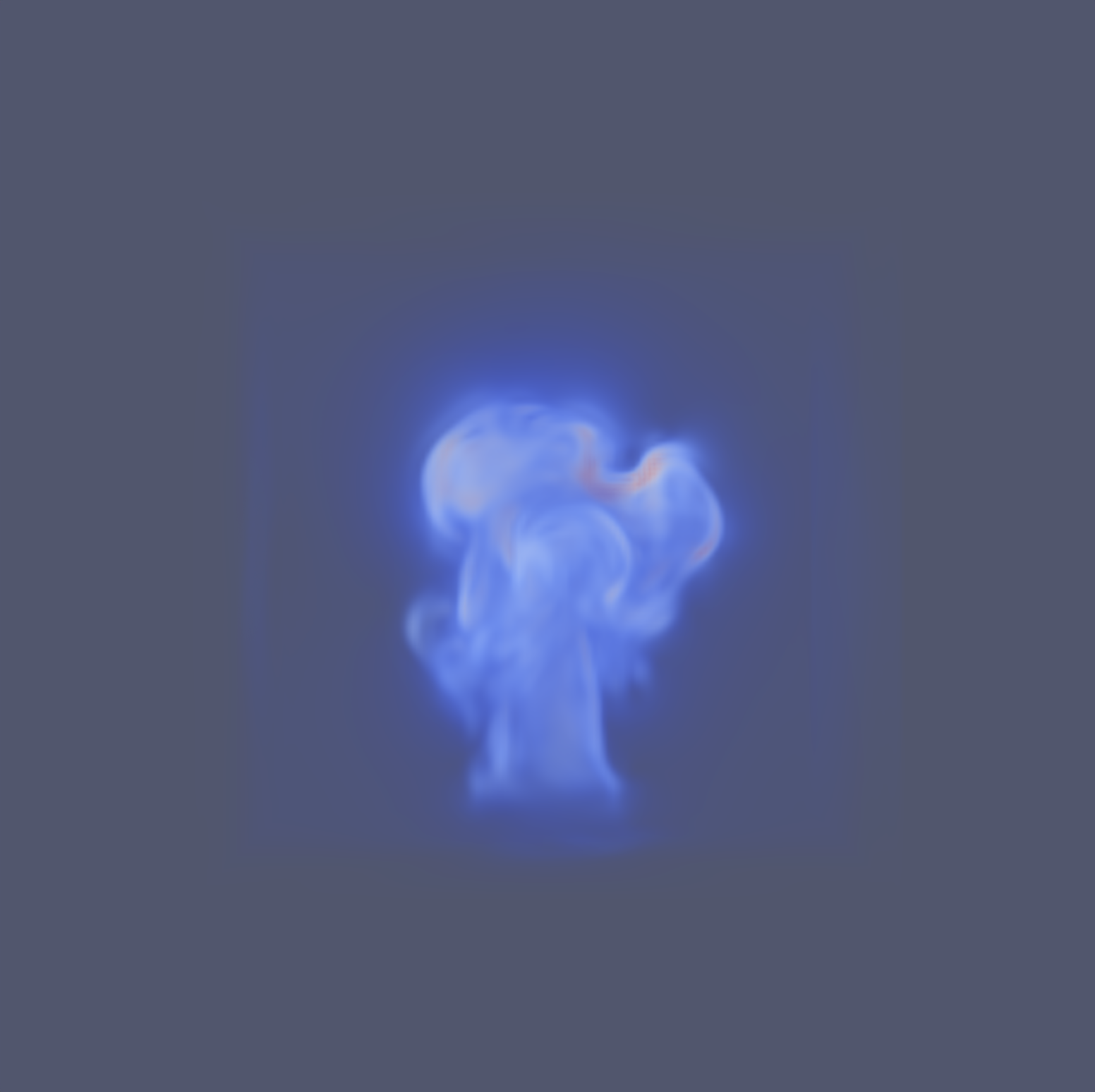}{Curl SIREN}%
  \hfill%
  \formattedgraphics{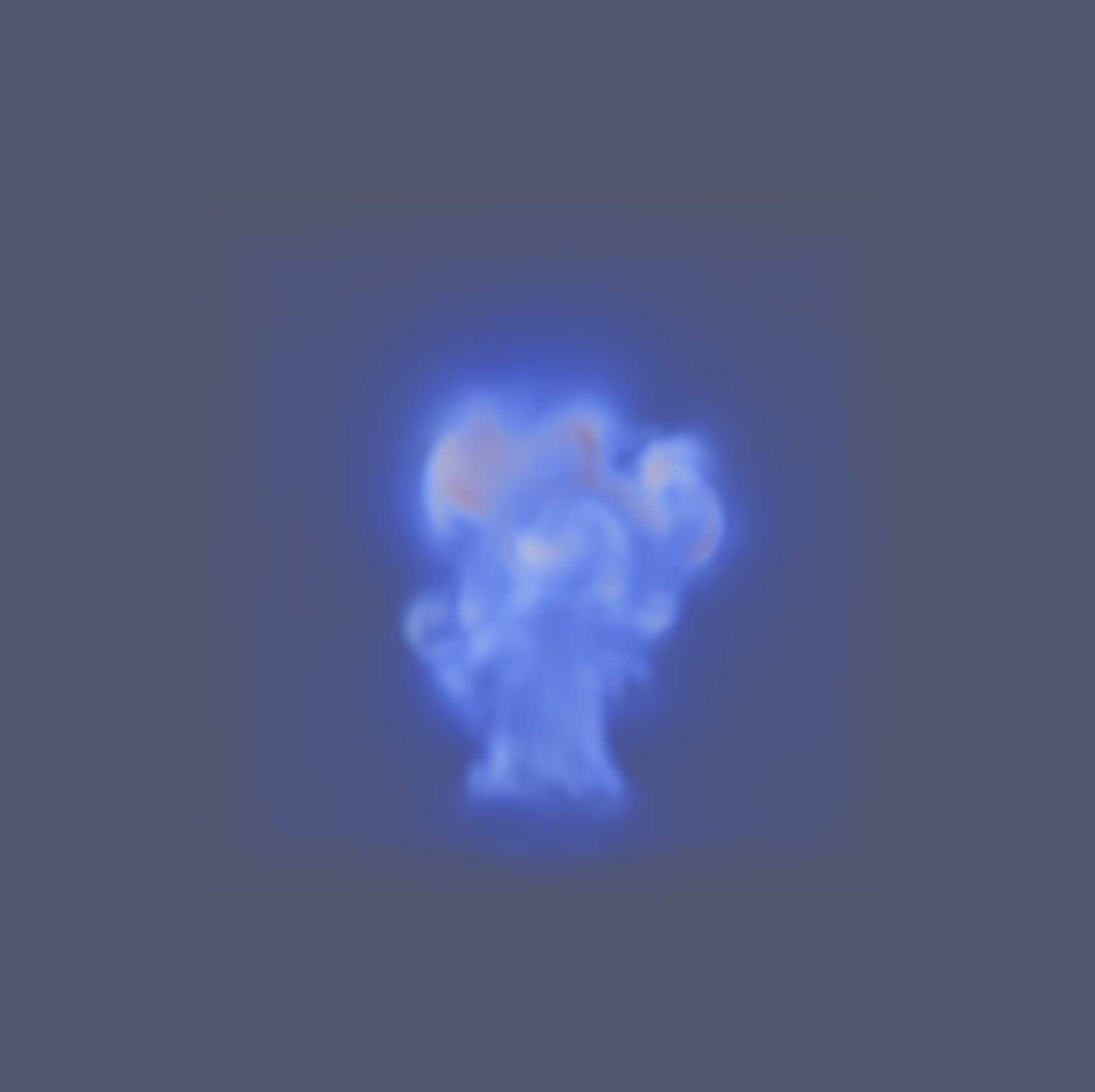}{Regular RBF}%
  \hfill%
  \formattedgraphics{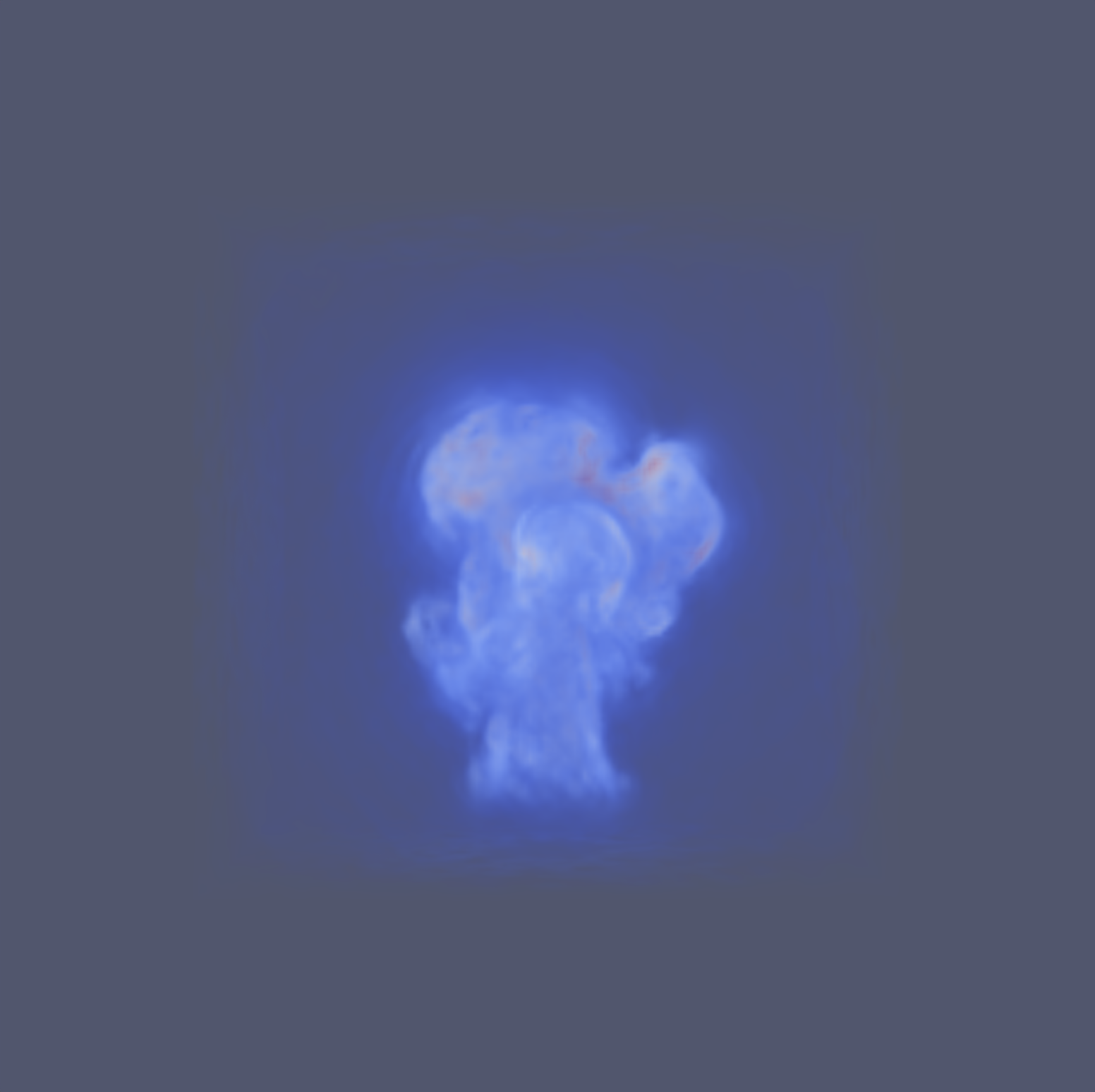}{DFK-Poly6}%
  \hfill%
  \formattedgraphics{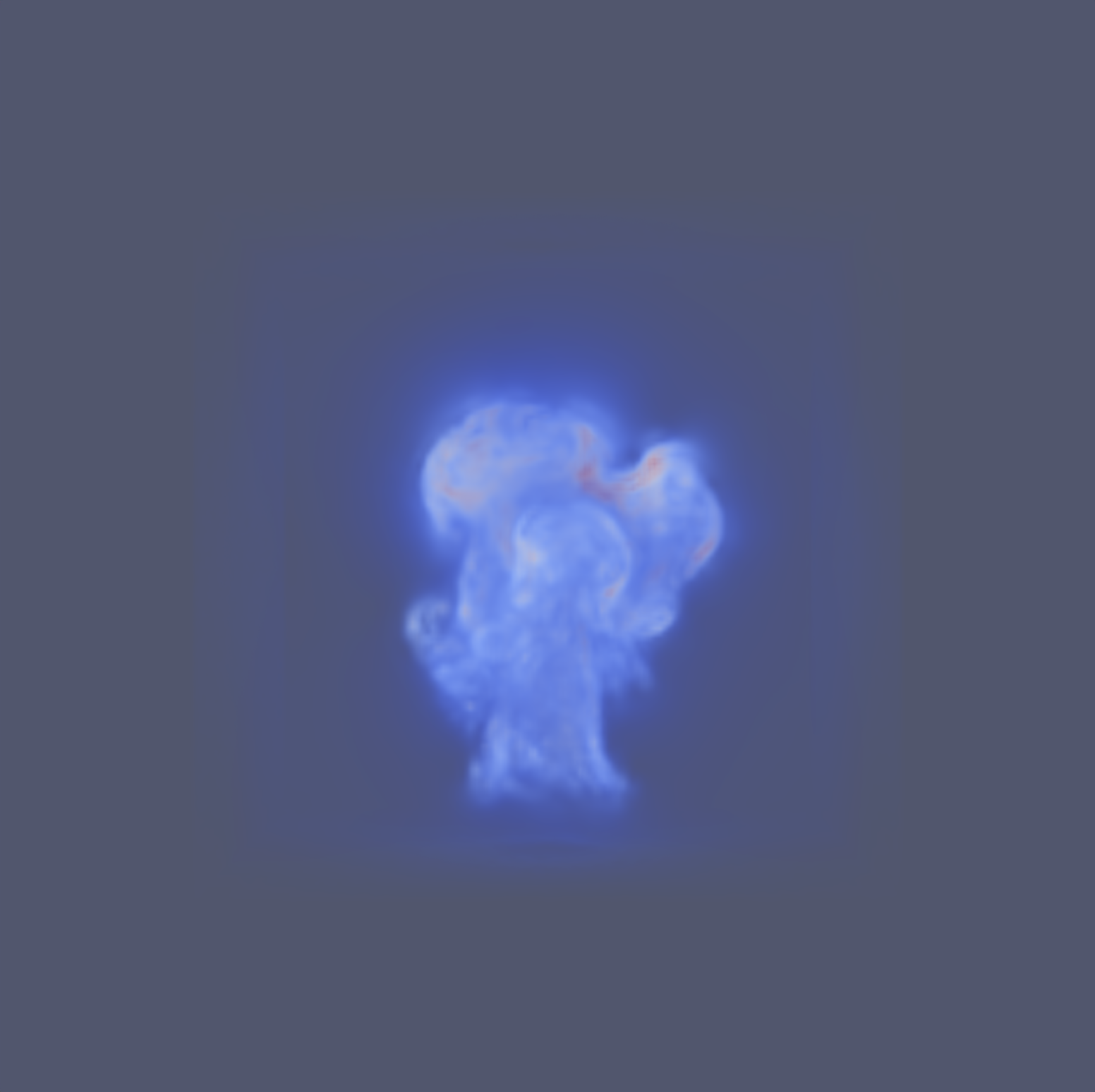}{Curl Kernel}%
  \hfill%
  \formattedgraphics{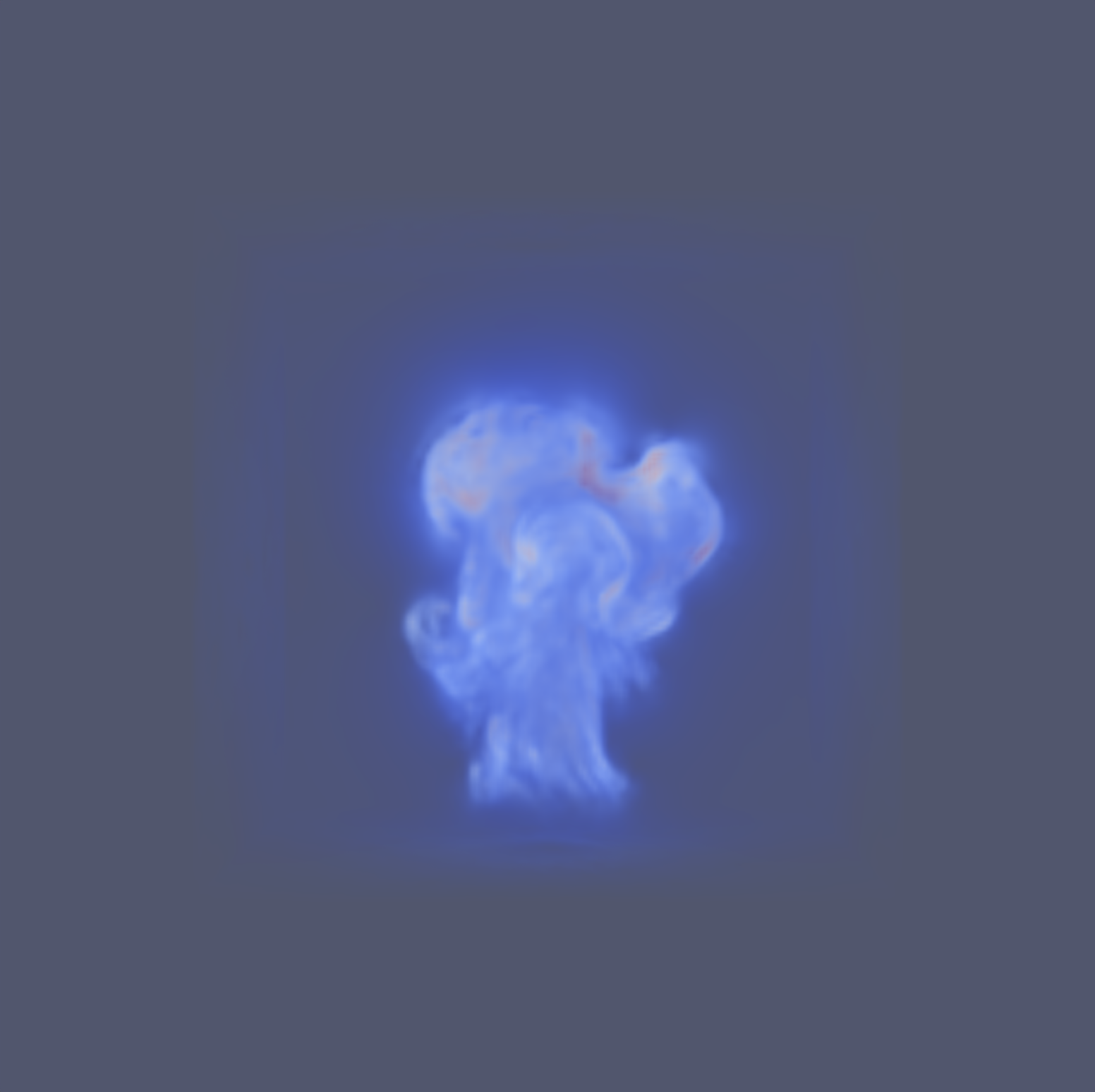}{\bfseries DFK-Wen4}%
  \\
  \vspace{-1em}
  \caption{Fitting experiments of the simple plume, with the resulting vorticity fields rendered. SIREN, Curl SIREN, and Regular RBF produce overly smooth results, failing to capture finer details. In contrast, Curl Kernel and DFK-Wen4 closely match the ground truth (see Tab.~\ref{tab:losses}), delivering accurate and high-fidelity representations, with only $4.3\%$ of DoFs used compared to the raw data.}
  \label{fig:plume}
  \Description{plume experiments}
\end{figure}

\section{Divergence-Free Kernels}

We present the essential formulae for constructing the physics-consistent search space using divergence-free matrix-valued kernels, specifically DFKs-Wen4, along with a theoretical analysis of their advantages over other kernel-based representations.

\subsection{Construction of DFKs-Wen4}
\label{sec:con_flowpek}

Given the spatial dimension $d$, we consider the $i$-th kernel, located at $\bm{x}_i\in\mathbb{R}^d$ with a size of $h_i\in\mathbb{R}$, to be associated with a scalar-valued kernel defined as
\begin{equation}
  \phi_i(\bm{x})=
  \phi\left(\frac{\Vert\bm{x}-\bm{x}_i\Vert}{h_i}\right)\text{,}
\end{equation}
where $\phi(r)$ is a radial basis function.
From this scalar kernel $\phi_i(\bm{x})$, we derive a matrix-valued kernel $\bm{\psi}_i:\mathbb{R}^d\to\mathbb{R}^{d\times d}$ by applying a second-order differential operator \cite{Narcowich1994}:
\begin{equation}
  \label{eqn:df-constr}
  \bm{\psi}_i(\bm{x})=\left(-\bm{I}\del^2+\del\del^\top\right)\phi_i(\bm{x})\text{,}
\end{equation}
in which $\bm{I}$ is the identity matrix and $\del\del^\top$ denotes the Hessian operator.
For any choice of $\phi(r)$ and any weight vector $\bm{\omega}_i\in\mathbb{R}^d$, it can be shown (see \S{A.2} in the supplementary document) that
\begin{equation}
  \label{eqn:construct}
  \bm{\psi}_i(\bm{x})\,\bm{\omega}_i=\del\times(\del\phi_i\times\bm{\omega}_i)\text{.}
\end{equation}
Therefore, if we express a flow field $\bm{u}(\bm{x})$ as the summation
\begin{equation}
  \bm{u}(\bm{x})=\sum_i{\bm{\psi}_i(\bm{x})\,\bm{\omega}_i}\text{,}
\end{equation}
the divergence of the field, $\del\cdot\bm{u}$, is guaranteed to be zero.

Considering the trade-off between the flow field smoothness and computational complexity, we adopt the $\mathcal{C}^4$-continuous piecewise-polynomial radial function proposed by \citet{Wendland1995}\text{,}
\begin{equation}
  R_\mathrm{Wen4}(r)=(1-r)_+^6(35r^2+18r+3)\text{,}
\end{equation}
to server as $\phi(r)$, where $(\cdot)_+$ denotes $\max(\cdot, 0)$.
The concrete formulations of DFKs-Wen4 and their derivatives that are useful for both fluid mechanics and optimization algorithms are provided in \S{B}.

\subsection{Properties Analysis}
\label{sec:analysis}

Although the procedure outlined in \S\ref{sec:con_flowpek} is not the only way to obtain divergence-free kernel-based representations, we argue that DFKs-Wen4 are particularly well-suited for flow field reconstruction due to their desirable properties.

\begin{figure}[t]
  \centering
  \setlength{\imagewidth}{.196\textwidth}
  \subcaptionbox{\label{fig:illu_curl}Curl Kernel;}{\includegraphics[width=\imagewidth]{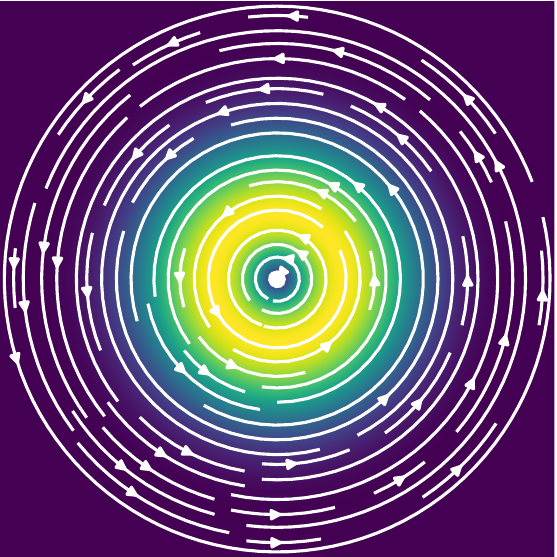}}%
  \hfill%
  \subcaptionbox{\label{fig:illu_gauss}DFK-Gauss;}{\includegraphics[width=\imagewidth]{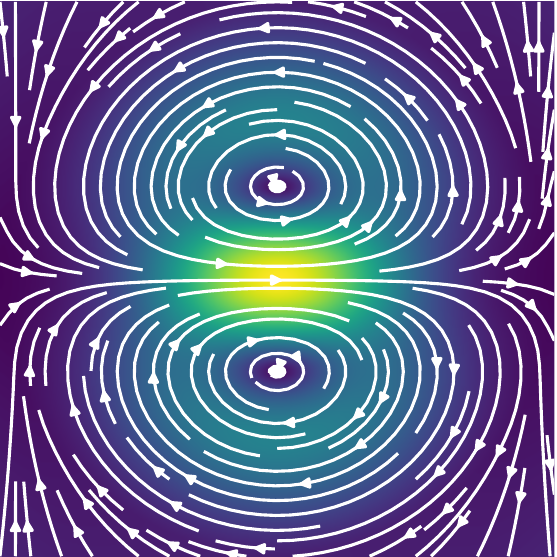}}%
  \hfill%
  \subcaptionbox{\label{fig:illu_poly6}DFK-Poly6;}{\includegraphics[width=\imagewidth]{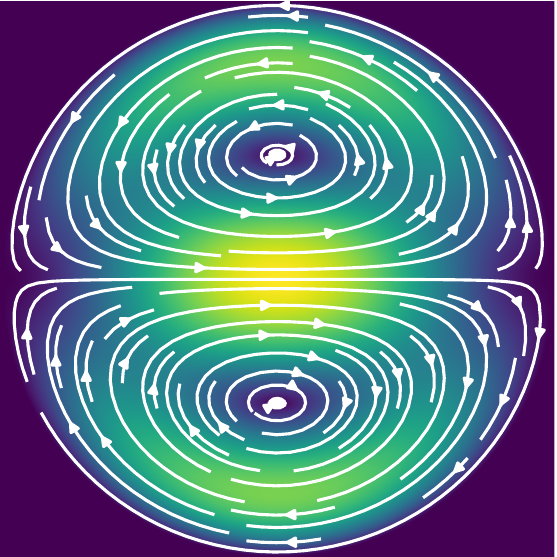}}%
  \hfill%
  \subcaptionbox{\label{fig:illu_flowpek}\textbf{DFK-Wen4};}{\includegraphics[width=\imagewidth]{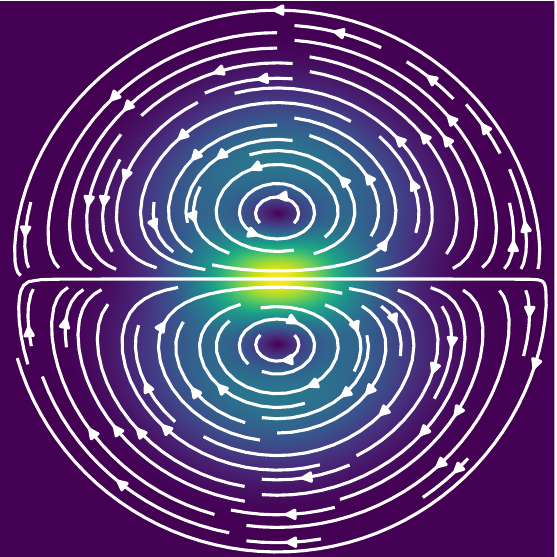}}%
  \hfill%
  \subcaptionbox{\label{fig:illu_vort}Vorticity of (\subref{fig:illu_flowpek}).}{\includegraphics[width=\imagewidth]{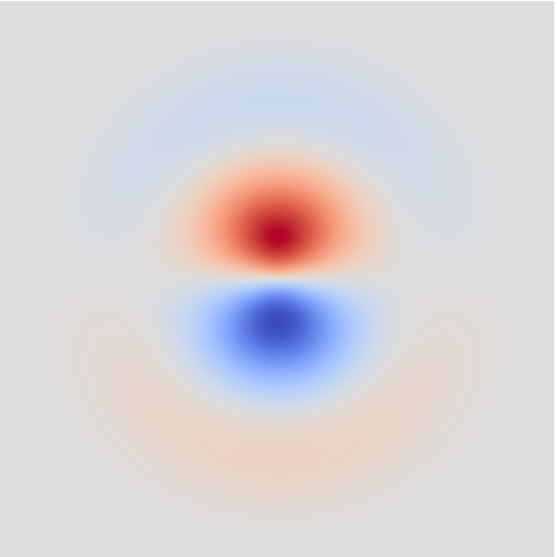}}\\
  \vspace{-1em}
  \caption{\label{fig:illustration}
  Illustration of different divergence-free kernels in 2D. The left four images display streamline plots, where brighter background colors indicate higher velocity.
  For Curl Kernel (\subref{fig:illu_curl}), the vector field is constructed as $(\partial/\partial y,-\partial/\partial x)\,R_\mathrm{Wen4}$, which generalizes to $\del\times R_\mathrm{Wen4}\bm{\omega}$ in 3D.
  For DFKs (\subref{fig:illu_gauss}--\subref{fig:illu_flowpek}), the fields are formulated as $(-\bm{I}\del^2+\del\del^\top)\,\phi\,\bm{\omega}$, where $\phi$ takes the forms $R_\mathrm{Gau}=\exp{(-9r^2/2)}$, $R_\mathrm{Poly6}=(1-r^2)^3_+$, and $R_\mathrm{Wen4}$, respectively. For comparison, $\bm{\omega}$ is set to $(1,0)$.
  Note that these plots also serve as 2D cross-sections of their 3D counterparts.
  The rightmost image (\subref{fig:illu_vort}) shows the corresponding vorticity field of DFK-Wen4, with cool and warm colors indicating opposite rotation directions.}
  \Description{different kernels}
\end{figure}

\subsubsection{Dipolarity}
Unlike vortex flows derived from the curl of a radial vector potential (or the stream function in 2D), as shown in Fig.~\ref{fig:illu_curl}, the kernel constructed via Eq.~\eqref{eqn:construct} exhibits distinct dipole characteristics, with two vortices rotating in opposite directions, as illustrated in Figs.~\ref{fig:illu_gauss}--\ref{fig:illu_flowpek}.
This dipolar structure is a fundamental phenomenon that arises when a flow interacts with a blocking body (e.g., around a cylinder or in the formation of K\'{a}rm\'{a}n vortex streets). Furthermore, for a divergence-free vector field, a dipole, rather than a single vortex, represents the leading term in its multipole expansion\footnote{Commonly used in the representation of another type of divergence-free vector field, the magnetic field \cite{Griffiths2017}. Streamlines in Figs.~\ref{fig:illu_gauss}--\ref{fig:illu_flowpek} closely resemble the magnetic field lines near a magnetic dipole moment.},
providing a more compact and accurate representation of incompressible flow fields.

\subsubsection{Compact Support}

DFKs-Wen4 inherit the compact support property of Wendland functions, enabling localized optimizations that avoid global interdependence. This is ideal for representing regions with zero velocity, such as solid boundaries or stationary fluids.
In contrast, Gaussian functions, often used in classification tasks,
are less suited for flow field reconstruction.
As shown in Fig.~\ref{fig:illu_gauss}, divergence-free kernels based on Gaussian functions (DFKs-Gauss) influence the entire space, increasing computational complexity and hindering local details capture.
Truncating the Gaussian function at a certain radius disrupts continuity, making it unsuitable for this application.

\subsubsection{Positive Definiteness}
It is a key concept in kernel-based interpolation \cite{Wendland2004}, ensuring that the interpolation matrix derived from distinct data points is invertible and the interpolation problem has a unique solution.
For kernels of the form $\bm{\psi}=(-\bm{I}\del^2+\del\del^\top)\,\phi$, it can be proven that if $\phi$ is positive definite, $\bm{\psi}$ will also be positive definite (see \S{A.1} in the supplementary document).  
In flow field reconstruction, positive definiteness means that 
as long as the kernel sizes and positions align with the data points' distribution
the flow field will be accurately represented. A set of weights always exists to ensure the reconstructed flow field passes through all data points. Intuitively, 
this property improves the convexity and smoothness of the optimization landscape.
In contrast, non-positive definite kernels, such as DFKs based on the Poly6 function commonly used in SPH methods \cite{Muller2003} (DFKs-Poly6), tend to produce flow fields with steeper gradients near local extrema (see Fig.~\ref{fig:illu_poly6}).



\subsubsection{Differentiability}
By utilizing $\mathcal{C}^4$-continuous Wendland radial functions, DFKs-Wen4 ensure the existence and continuity of second-order derivatives, which are essential for upstream and downstream tasks.
The Navier--Stokes equations that govern fluid motion involve second-order partial derivatives of the flow field for the computation of viscosity,
and the gradient of vorticity (the curl of the flow field, see Fig.~\ref{fig:illu_vort}) forms the foundation of vortex methods \cite{Bridson2015}. As a $\mathcal{C}^2$-continuous representation, DFKs-Wen4 guarantee differentiability for most reconstruction tasks based on physical flow field equations.

\begin{figure}[t]
  \centering%
  \newcommand{\formattedgraphics}[2]{\begin{overpic}[width=.162\linewidth,trim=1cm 1cm 1cm 1cm,clip]{#1}\put(4,90){\sffamily\scriptsize #2}\end{overpic}}%
  \formattedgraphics{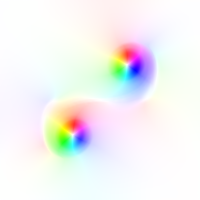}{Input}%
  \hfill%
  \formattedgraphics{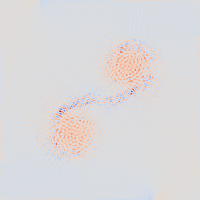}{Divergence}%
  \hfill%
  \formattedgraphics{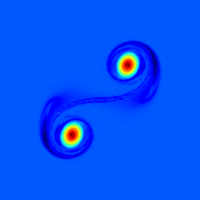}{\color{white} Vorticity (G.T.)}%
  \hfill%
  \formattedgraphics{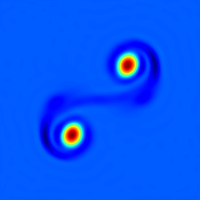}{\color{white} Curl SIREN}%
  \hfill%
  \formattedgraphics{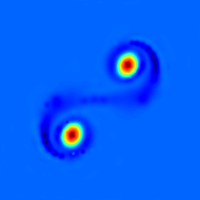}{\color{white} Curl Kernel}%
  \hfill%
  \formattedgraphics{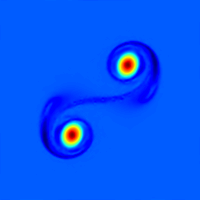}{\color{white}\bfseries DFK-Wen4}%
  \\
  \vspace{-1em}
  \caption{Projection experiments of the Taylor vortex. The input is illustrated using HSV color encoding, with its divergence plotted. The vorticity fields of ground truth and experimental results are presented with \emph{jet} color mapping, where our method provides the most accurate projection. The PSNR/SSIM values are as follows: Curl SIREN: 33.68/0.965; Curl Kernel: 29.80/0.970; \textbf{DFK-Wen4:} \textbf{39.30}/\textbf{0.994}.}
  \label{fig:taylor}
  \Description{taylor experiments}
\end{figure}

\section{Implementation}
\label{sec:implementation}

We implemented the flow field reconstruction framework based on divergence-free kernels using the PyTorch library.
To maintain a balanced distribution of kernels, we developed a custom C++ module that employs fast Poisson disk sampling \cite{Bridson2007} to determine the initial kernel positions.
Let $N$ denote the number of sampled kernels. To encourage that every point within the region of interest is influenced by multiple kernels, we initialize the kernel radii uniformly using the following formula:
\begin{equation}
  \label{eqn:init_radius}
  h_i=\eta\left[\frac{\Gamma\left(1+\frac{d}{2}\right)V}{N\pi^{\frac{d}{2}}}\right]^{\frac{1}{d}}
\end{equation}
where $\Gamma(\cdot)$ is the gamma function, and $\eta$ is a user-defined scaling factor.
Additionally, the initial kernel weights are set to zero.

Similar to 3DGS \cite{Kerbl2023}, in most application scenarios, the position $\bm{x}_i$, radius $h_i$, and weight $\bm{\omega}_i$ of each DFK are treated as trainable parameters, optimized using stochastic gradient descent with an Adam optimizer \cite{Kingma2017} and an exponential learning rate scheduler.
However, for large-scale scenarios, such as multi-frame flow field inference tasks, both time and memory costs become prohibitive. To mitigate this, we fix the positions and radii of the kernels and optimize only the weights.
To further boost performance, we developed a custom C++ hash grid module to precompute and cache influenced data point--kernel pairs and extended PyTorch with Taichi \cite{Hu2019} to enable parallel computation of each pair's contribution.
In this case, we continue using the Adam optimizer but switch to full-batch gradient descent, reducing the learning rate by $10\%$ once a plateau is detected.

\paragraph{Comparison methods}
For ablation studies, we implemented several divergence-free kernel-based representations, including Curl Kernels, DFKs-Poly6, and DFKs-Wen4, as well as Regular (scalar-valued) RBFs that mimic the Gaussian kernel using compactly supported Wenland's $\mathcal{C}^2$
polynomial $R_\mathrm{Wen2}=(1-r)^4_+(4r+1)$.
The same initialization and training strategies outlined above are applied.
For comparison with INRs, we trained
neural network models that use periodic activation functions, dubbed sinusoidal representation networks (SIRENs) \cite{Sitzmann2020}, to capture fine details.
In addition to standard SIRENs with physics-informed divergence penalties, we introduced Curl SIREN architectures, which model flow fields by taking the curl of network outputs \cite{Richter2024}, thereby ensuring inherent divergence-free behavior. 

\begin{figure}[t]
  \centering
  \newcommand{\formattedgraphics}[2]{\begin{overpic}[width=.18\linewidth,trim=5cm 5cm 5cm 5cm,clip]{#1}\put(4,88){\sffamily\scriptsize\color{white} #2}\end{overpic}}
  \formattedgraphics{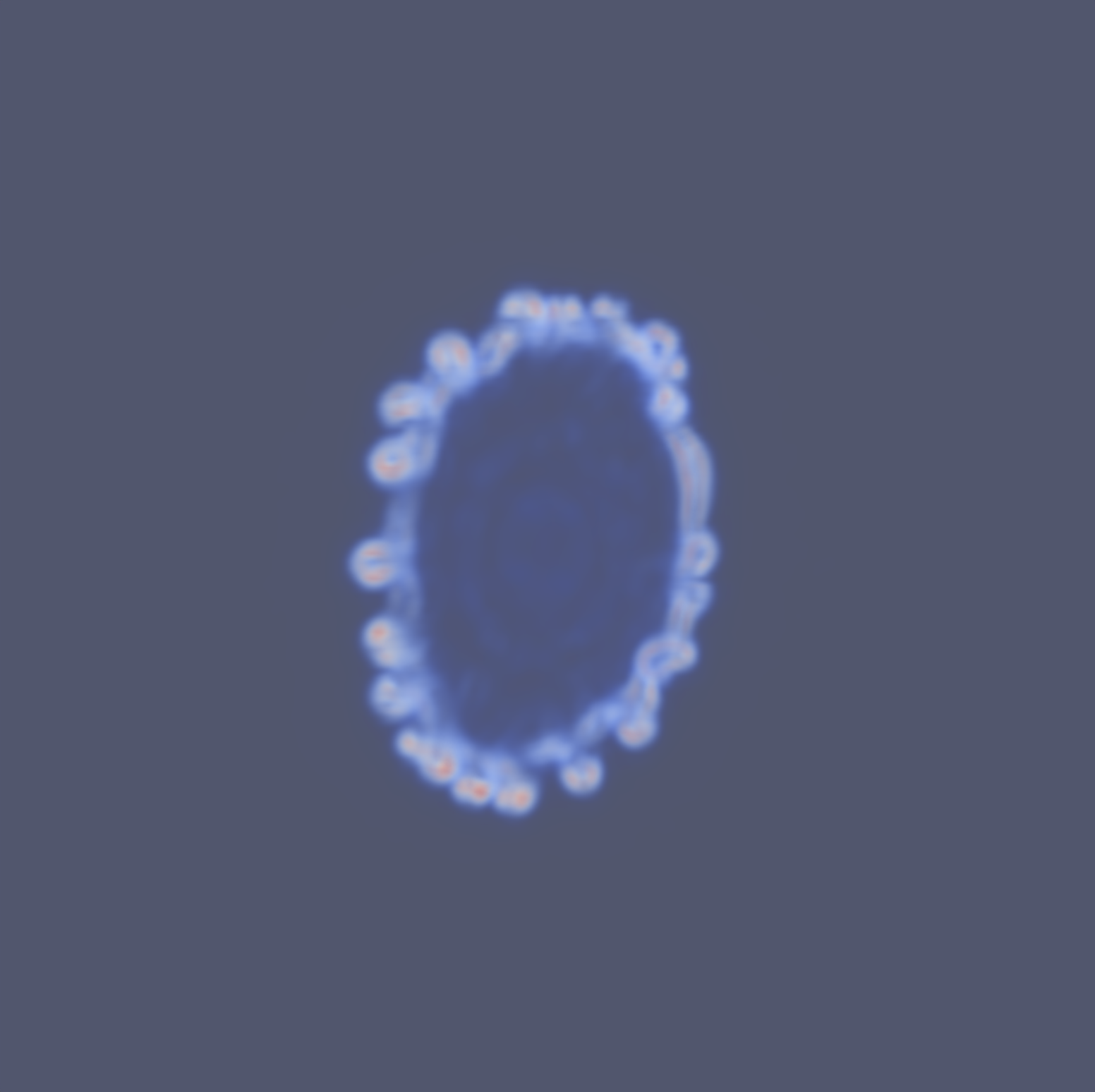}{Input}%
  \formattedgraphics{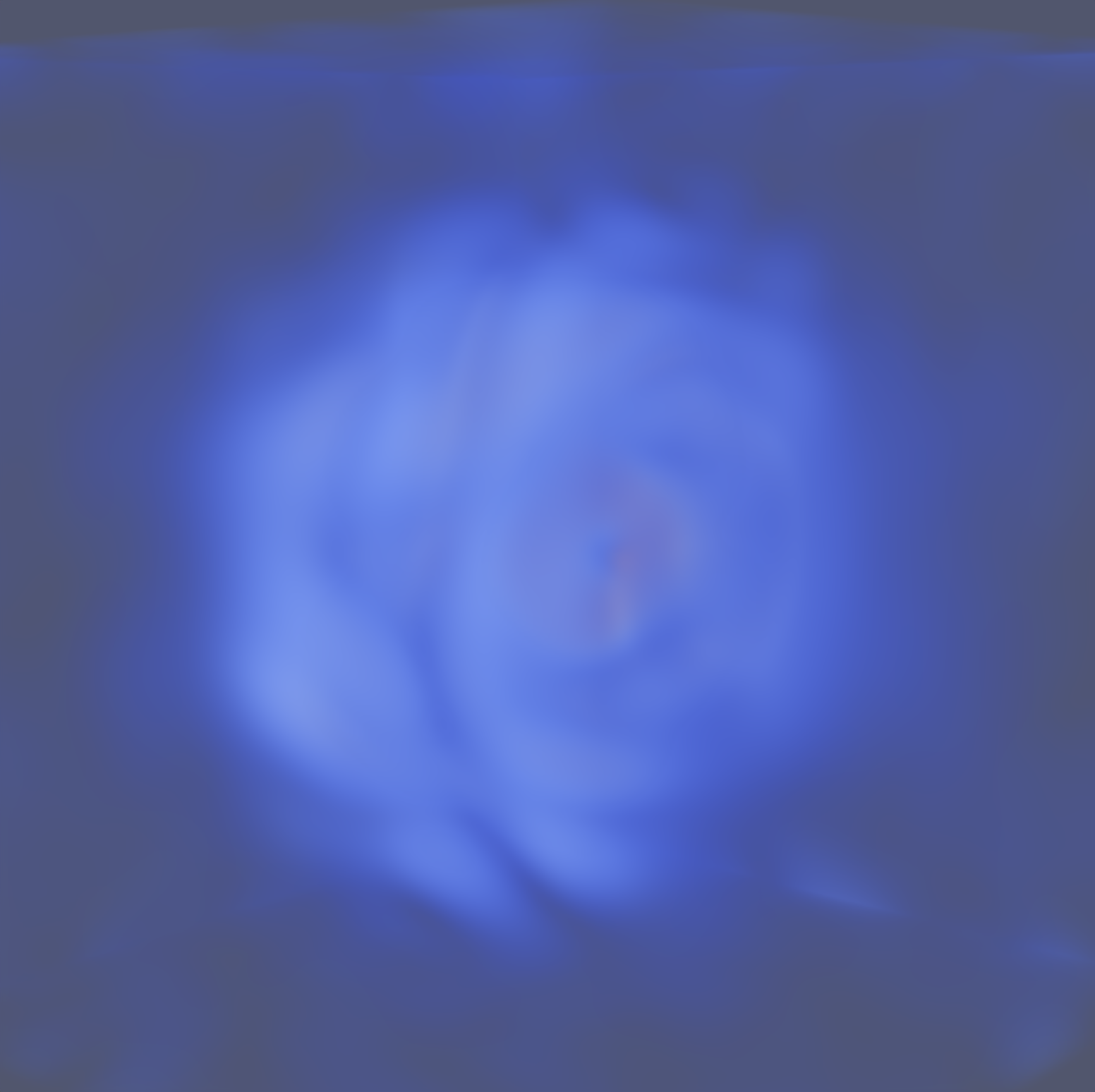}{Curl SIREN}%
  \formattedgraphics{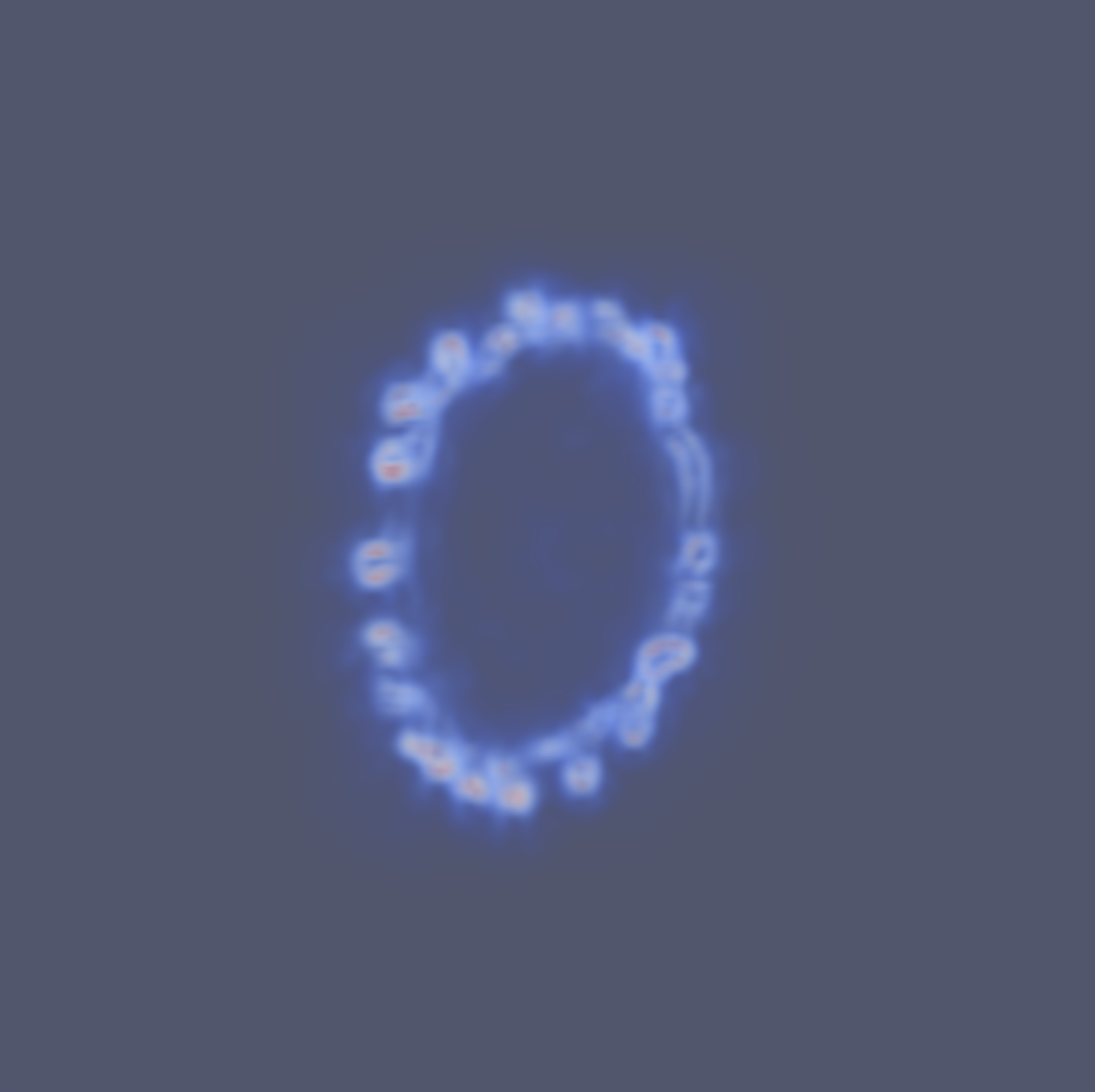}{Curl Kernel}%
  \formattedgraphics{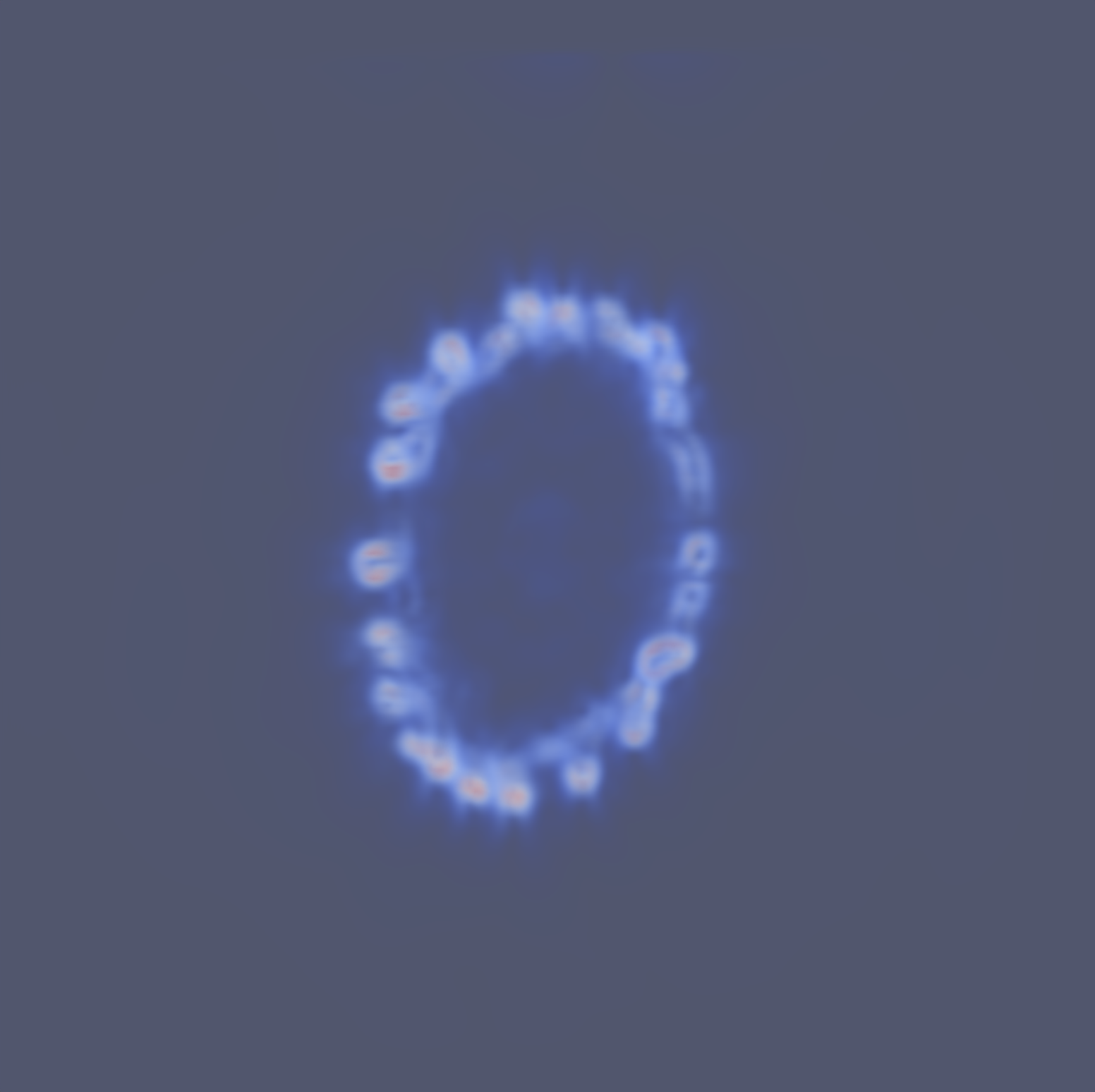}{\bfseries DFK-Wen4}%
  \\
  \vspace{-1em}
  \caption{Projection experiments of the vortex ring collision. The vortice fields of input and results are rendered. As demonstrated in Tab.~\ref{tab:losses}, DFK-Wen4 has the lowest loss in this case, though Curl Kernel also provides comparable visualization. Note that the vorticity field output by the INR is very diffuse.}
  \label{fig:rings}
  \Description{rings experiments}
\end{figure}

\section{Experiments}
\label{sec:experiments}

All experiments are conducted on a Windows 11 system equipped with an AMD Ryzen 9 7950X processor and an NVIDIA GeForce RTX 4090 GPU (24GB VRAM).
Unless otherwise specified, all methods are trained using the same number of epochs and batch size for each scenario.
We ensure that DFKs-Wen4 consistently have the fewest trainable parameters compared to the other methods by adjusting the numbers of initialized kernels.
A summary of the test case details and the final loss values at convergence can be found in Tab.~\ref{tab:losses}.
For further experimental settings and statistics, please refer to \S{C} in the supplementary document.

\begingroup
\setlength{\intextsep}{0pt}%
\setlength{\columnsep}{5pt}%
\begin{wrapfigure}{r}{0.24\linewidth}
\centering
\vspace{-0.5pt}
\includegraphics[width=\linewidth]{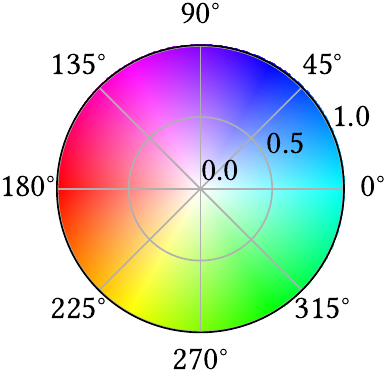}
\end{wrapfigure}
\paragraph{Flow visualization}
We use three primary approaches to visualize flow fields.
(1) \emph{Vorticity mapping}:
We calculate the vorticity of the flow and map regions with higher vorticity to greater opacity and warmer colors to highlight intense rotational motion.
(2) \emph{HSV color encoding}:
Drawing inspiration from optical flow techniques \cite{Baker2011}, we map the direction of the 2D velocity to hue and its (relative) magnitude to saturation, as illustrated in the inset figure. For 3D flow fields, we project the data onto a specified 2D plane in advance.
(3) \emph{Line integral convolution (LIC)} \cite{Cabral1993}: This 2D technique represents the flow's streamlines, offering a clear visualization of its structure.

\endgroup

\begin{table}[t]
  \centering
  \caption{Basic specifications and final loss values at convergence of experiments. For each case, the data scale is given in the number of data points, and the result with the minimal loss is colored green.
  The DFKs-Wen4 won in nearly all the experiments except for \emph{plume} and \emph{missile}.
  See discussions in the main text.}
  \label{tab:losses}
  \def\best{\color{Green}}
  \vspace{-1em}
  \mbox{}\clap{\resizebox{\textwidth}{!}{%
  \begin{tabular}{cccccccccc}
      \hline\hline
      Section & Case & $d$ & Data scale & SIREN & Curl SIREN & Regular RBF & DFK-Poly6 & Curl Kernel &\bfseries DFK-Wen4 \\
      \hline\hline
      \ref{sec:fitting} & K\'{a}rm\'{a}n & 2 & $104.4\,\mathrm{k}$ & \num{3.950e-02} & \num{5.289e-03} & \num{2.150e-03} & \num{2.758e-03} & \num{1.483e-03} &\best\num{5.421e-04}\\
      \ref{sec:fitting} & analytic & 3 & $512.0\,\mathrm{k}$ & \num{4.819e-02} & \num{7.580e-03} & \num{1.176e-02} & \num{1.395e-02} & \num{8.445e-03} & \best\num{7.661e-03}\\
      \ref{sec:fitting} & plume & 3 & $2.097\,\mathrm{M}$ & \num{3.783e-03} & \num{1.940e-03} & \num{2.515e-03} & \num{2.309e-03} & \best\num{1.517e-03} & \num{2.039e-03}\\
      \hline
      \ref{sec:projection} & Taylor & 2 & $40.00\,\mathrm{k}$ &  & \num{9.082e-04} &  &  & \num{1.357e-03} & \best\num{2.611e-04}\\
      \ref{sec:projection} & collision & 3 & $2.097\,\mathrm{M}$ &  & \num{1.069e-03} &  &  & \num{4.982e-04} & \best\num{3.372e-04}\\
      \hline
      \ref{sec:completion} & flows (\SI{0}{\degree}) & 2 & $210.4\,\mathrm{k}$ &  & \num{1.062e-04} &  & & & \best\num{7.124e-05}\\
      \ref{sec:completion} & flows (\SI{45}{\degree}) & 2 & $210.4\,\mathrm{k}$ &  & \num{1.654e-04} &  & & & \best\num{1.034e-04}\\
      \ref{sec:completion} & flows (\SI{90}{\degree}) & 2 & $210.4\,\mathrm{k}$ &  & \num{2.131e-04} &  & & & \best\num{1.264e-04}\\
      \ref{sec:completion} & missile & 2 & $1.012\,\mathrm{M}$ &  & \best\num{2.608e-04} &  & \num{4.140e-03}  & \num{1.134e-03} & \num{3.306e-04}\\
      \ref{sec:completion} & bullet & 3 & $5.504\,\mathrm{M}$ &  & \num{1.174e-03} &  &  & \num{4.612e-04} & \best\num{3.954e-04}\\
      \hline
      \ref{sec:super-resolution} & turb. A & 2 & $4.096\,\mathrm{k}$ &  & \num{1.327e-02} &  & \num{2.109e-03}  & \num{2.494e-03} & \best\num{4.950e-04}\\
      \ref{sec:super-resolution} & turb. B & 2 & $4.096\,\mathrm{k}$ &  & \num{8.306e-03} &  & \num{1.908e-03}  & \num{2.595e-03} & \best\num{4.936e-04}\\
      \ref{sec:super-resolution} & obstacle & 3 & $135.2\,\mathrm{k}$ &  & \num{5.472e-03} &  & \num{6.670e-03}  & \num{4.789e-03} & \best\num{2.302e-03}\\
      \hline
      \ref{sec:inference-data} & rising & 3 & $116.0\,\mathrm{M}$ & \num{6.212e-02} & \num{6.436e-02} &  & & & \best\num{4.632e-02}\\
      \ref{sec:inference-data} & teapot & 3 & $134.0\,\mathrm{M}$ & \num{8.319e-02} & \num{8.693e-02} &  & & & \best\num{5.447e-02}\\
      \hline
      \ref{sec:inference-videos} & scalar & 3 & $232.4\,\mathrm{M}$ & \num{6.099e-02}  &  &  & & & \best\num{5.032e-02}\\
      \hline\hline
  \end{tabular}}}
\end{table}

\subsection{Storage of Complete Flow Fields}

We begin by validating DFKs' representation capabilities through fitting tasks with dense, complete data ($\Omega_\mathrm{D}=\Omega$).
By optimizting DFKs' positions, radii, and weights, we are able to store high-fidelity flow fields with far fewer degrees of freedom (DoFs) than the raw data.

\subsubsection{Fitting of Incompressible Flows}
\label{sec:fitting}
For dense incompressible flow field data, whether experimentally or synthetically generated, fitting accuracy on such data reflects the suitability of the chosen representation for reconstruction. We compare SIREN, Curl SIREN, Regular RBF, DFK-Poly6, Curl Kernel, and DFK-Wen4 on pure fitting tasks. The loss functions follow Eq.~\eqref{eqn:fitting_loss} for $\mathcal{L}_\mathrm{obs}$, with $\lambda_\mathrm{div} = 0.5$ (for non-divergence-free methods) and $\lambda_\mathrm{bou} = 1$.

\paragraph{K\'{a}rm\'{a}n vortex street (2D)}

We simulate the steady incoming flow around a cylinder using the method of \citet{Narain2019} on a $512\times204$ grid, capturing the periodic shedding of counter-rotating vortex pairs, regularly arranged on both sides of the obstacle. After the formation of the vortex street, we select a single frame for fitting. 
In this scenario, our method employs only \num{5367} kernels ($\sim27\,\mathrm{k}$ parameters, compression ratio as low as $8.6\%$) to achieve the least fitting error.
The vorticity fields presented in Fig.~\ref{fig:karman} clearly illustrate that DFKs-Wen4 yield results with superior clarity and accuracy, effectively capturing the flow dynamics with greater precision.

\paragraph{Analytic vortices (3D)}
We generate a pointwise divergence-free vector field, going beyond the mere incompressibility requirement in the finite-volume sense typically used in simulations. Specifically, we compute the curl of the following vector potential:
\begin{equation}
  \bm{A}=\begin{pmatrix}
  (1-x^2)\,(1-y^2)\,(1-z^2)\\
  (1-x^2)\,(1-y^2)\,(1-z^2)+\sin{\pi x}\sin{\pi y}\sin{\pi z}\\
  (1-x^2)\,(1-y^2)\,(1-z^2)
  \end{pmatrix}\text{,}
\end{equation}
and \num{100} vortex particles are then randomly seeded in the resulting field. 
The illustration in Fig.~\ref{fig:analytic} indicates that our method with \num{21117} kernels initialized has the minimal fitting loss.

\paragraph{Simple plume (3D)}
Using the method proposed by \citet{Fedkiw2001} and the \emph{mantaflow} framework\footnote{Maintained by Nils Thuerey et al. http://mantaflow.com/}, we simulate a rising smoke plume on a $128^3$ grid and selecte a single frame for fitting.
As shown in Fig.~\ref{fig:plume}, Curl Kernel and DFK-Wen4 achieve the best results, each comprising $\sim27\,\mathrm{k}$ parameters, far less than the simulation output ($\sim6.3\,\mathrm{M}$).

\begin{figure}[t]
  \centering
  \setlength{\tabcolsep}{0.5pt}
  \setlength{\imagewidth}{.155\linewidth}
   \renewcommand{\arraystretch}{0.}
    {\small\begin{tabular}{m{0.3cm}<{\centering}m{\imagewidth}<{\centering}m{\imagewidth}<{\centering}m{\imagewidth}<{\centering}}
        & \SI{0}{\degree} & \SI{45}{\degree} & \SI{90}{\degree}\\
        \specialrule{0em}{1pt}{1pt}
        \rotatebox{90}{Curl SIREN} &
        \includegraphics[width=\imagewidth]{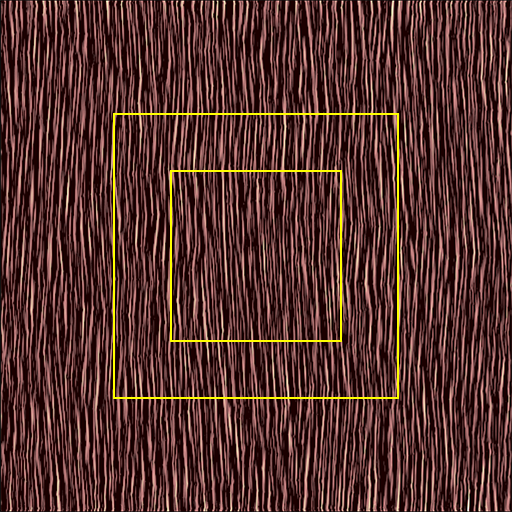} & 
        \includegraphics[width=\imagewidth]{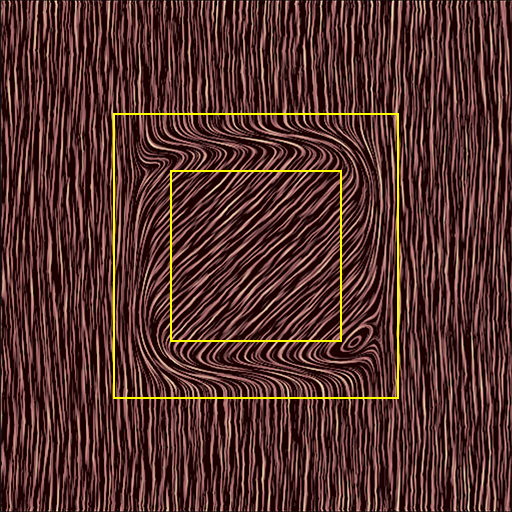} &
        \includegraphics[width=\imagewidth]{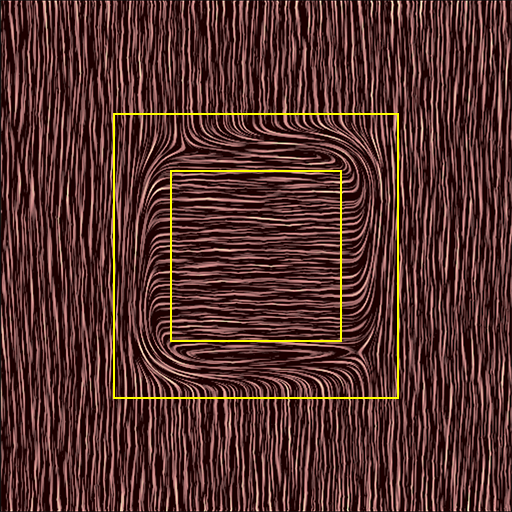}\\
        \specialrule{0em}{0pt}{1pt}
        \rotatebox{90}{\bfseries DFK-Wen4} &
        \includegraphics[width=\imagewidth]{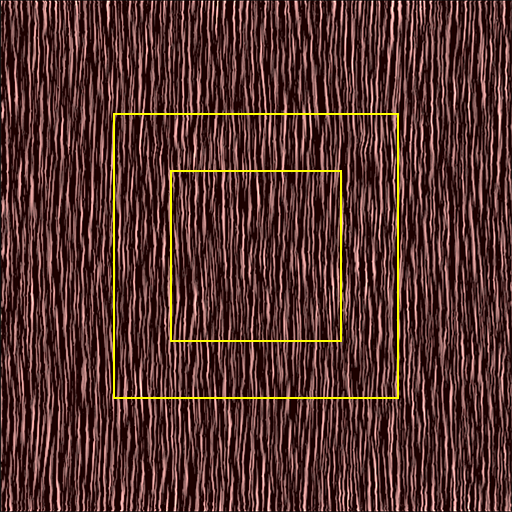} &
        \includegraphics[width=\imagewidth]{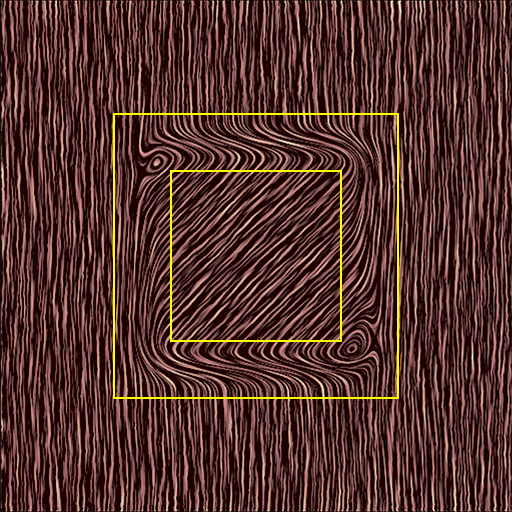} &
        \includegraphics[width=\imagewidth]{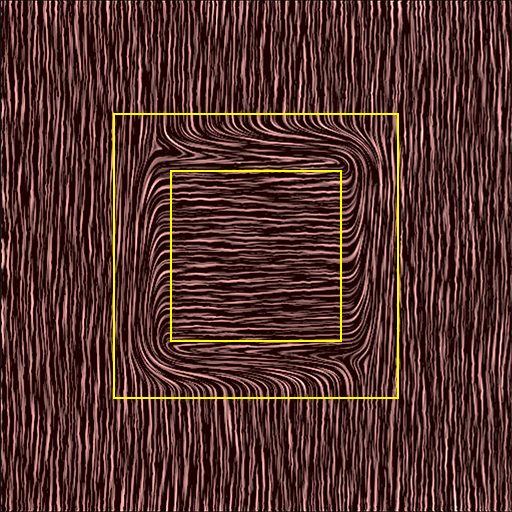}\\
    \end{tabular}}%
  \vspace{-1em}
  \caption{Inpainting experiments of the laminar flows, visualized using the LIC method. The region between the yellow rectangles is the area where the flow is completed. DFK-Wen4 fits the existing data better and produces more symmetric results. We use large initial kernel radii to encourage smoother inpainting.}
  \label{fig:laminar}
  \Description{laminar expriments}
\end{figure}

\subsubsection{Leray Projection}
\label{sec:projection}
A key advantage of DFKs are their ability to provide a local, adaptive Helmholtz decomposition basis.
The vector field expressed by Eq.~\eqref{eqn:construct} corresponds precisely to the divergence-free component of $-\del^2\phi_i(\bm{x})\,\bm{\omega}$. Therefore, for any given vector field, we can simply apply $-\del^2\phi_i$ (a regular RBF) to fit the data without divergence-free constraints, naturally performing the Leray projection afterwards. 
Unlike the traditional grid-based projection \cite{Stam1999}, this approach ensures that the resulting flow field is divergence-free at every point.

As noted by \citet{Richter2024}, other divergence-free representations can achieve similar results with proper adjustments. For instance, starting from the Curl Kernel, one can construct a curl-free Gradient Kernel, and their combination can be used for fitting, yeilding only the Curl Kernel component:
\begin{equation}
\bm{v}(\bm{x})=\del\times\left[R_\mathrm{Wen4}(\bm{x})\,\bm{\omega}_1\right]+\omega_2\,\del R_\mathrm{Wen4}(\bm{x})\text{.}
\end{equation}
The concept is akin to the Curl SIREN approach. Although these methods introduce additional coefficients, we compare them with DFKs to assess their relative performance.

\paragraph{Taylor vortex (2D)}
We use velocity field from a specific frame of the Taylor vortex example produced by the optimization-based fluid solver \cite{xing2024gridfreefluidsolverbased}, which does not guarantee divergence free, as the input for the projection step. Ideally, the projection should eliminate the divergence while preserving the vorticity. As shown in Fig.~\ref{fig:taylor}, our method with only \num{5416} kernels successfully recovers a velocity field with a vorticity distribution closest to the ground truth.

\paragraph{Vortex ring collision (3D)}
We project the velocity field from a specific frame of a 3D animation generated by the simulator of \citet{xing2024gridfreefluidsolverbased}. In this example, two vortex rings gradually approach each other and eventually shred into many smaller rings upon collision. As shown in Fig.~\ref{fig:rings}, our method achieves the best result with only \num{8544} kernels initialized.

\subsection{Completion of Quasi-Static Flow Data}

Due to measurement limitations, acquiring complete flow field data is often impractical. Typical challenges include missing data in specific regions or low-resolution sampling. In these cases, the optimization objective remains the same as in the fitting task, but $\Omega_\mathrm{D}$ is only a subset of $\Omega$ now.

As demonstrated in \S\ref{sec:fitting}, representations that enforce incompressibility via penalty terms perform significantly worse than alternative methods. Therefore, in this section, we focus our comparison exclusively on divergence-free approaches.
Note that there is no additional regularization except for the divergence-free constraint. We choose $\eta$ in Eq.~\eqref{eqn:init_radius} to be sufficiently large in order to encourage smoother completion.

\begin{figure}[t]
  \centering%
  \newcommand{\formattedgraphics}[2]{\begin{overpic}[width=.245\linewidth]{#1}\put(3,51){\sffamily\scriptsize #2}\end{overpic}}%
  \formattedgraphics{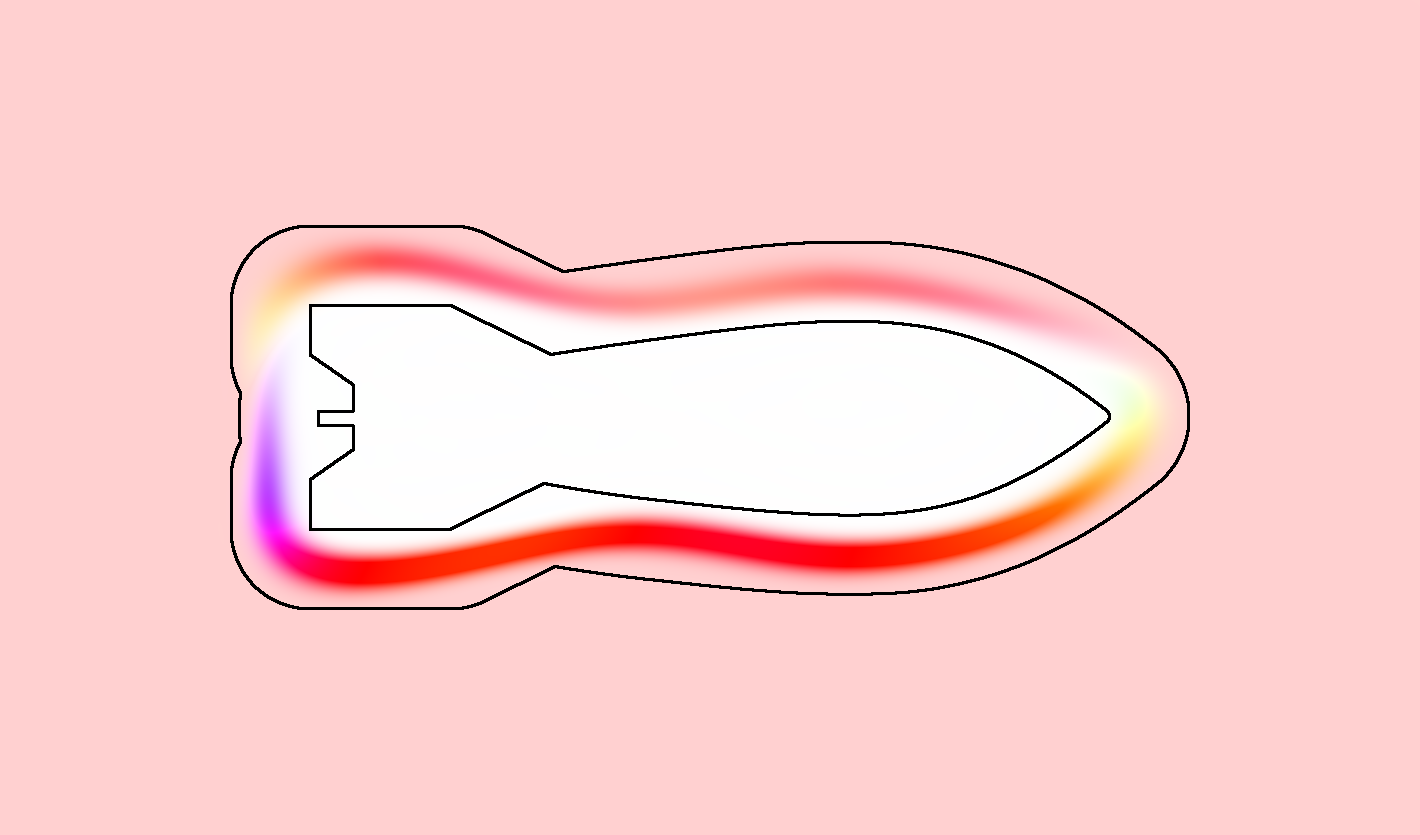}{Curl SIREN}%
  \hfill%
  \formattedgraphics{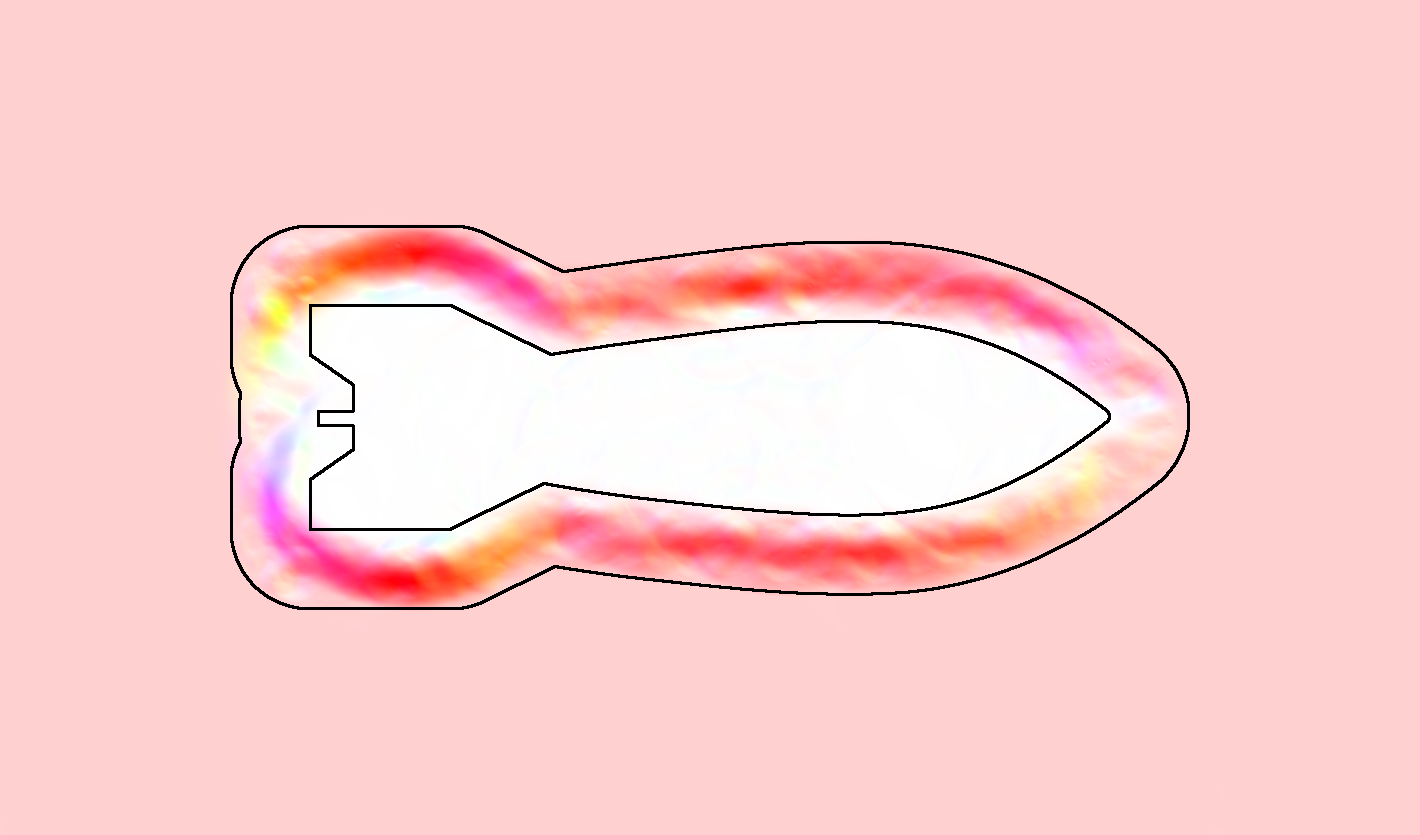}{DFK-Poly6}%
  \hfill%
  \formattedgraphics{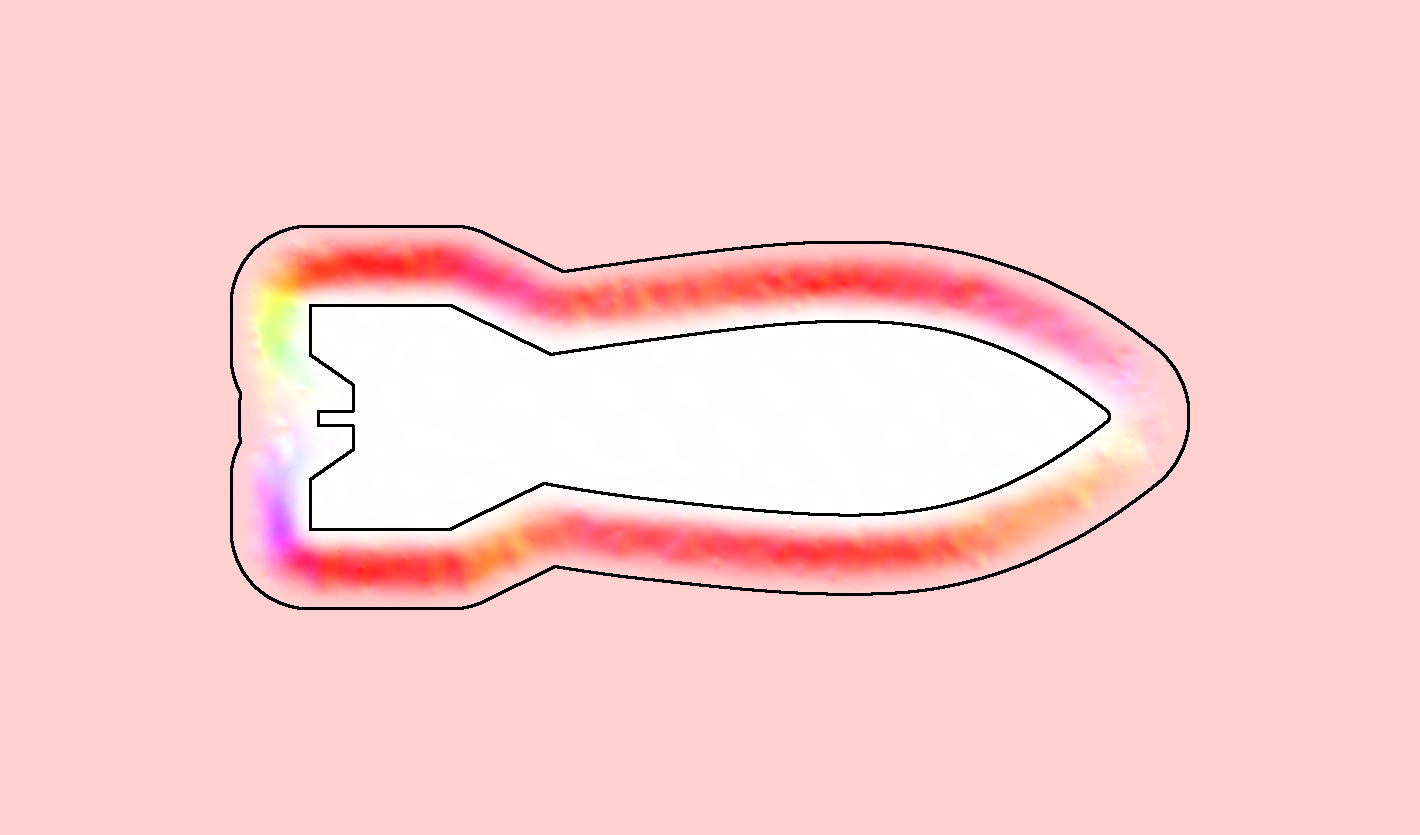}{Curl Kernel}%
  \hfill%
  \formattedgraphics{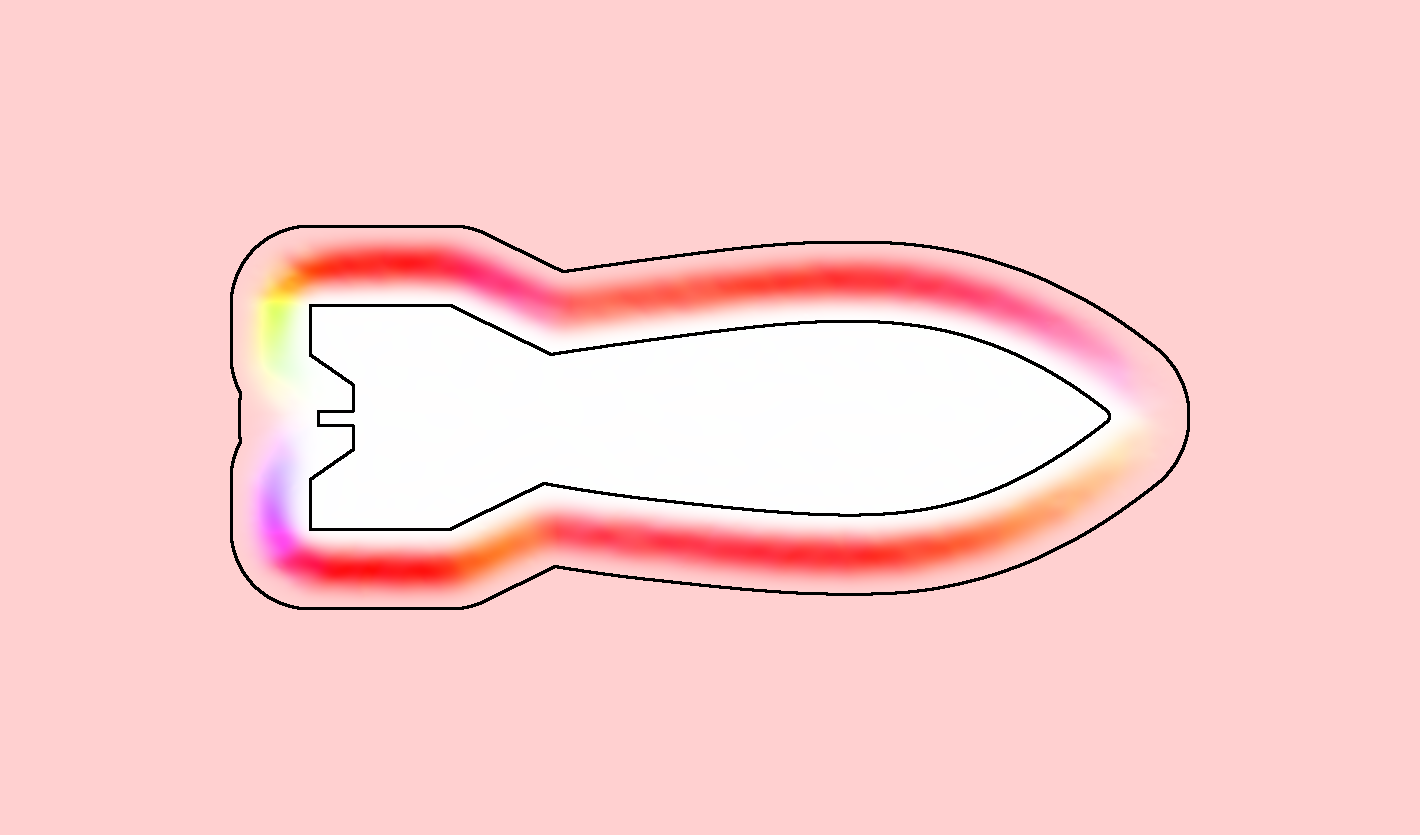}{\bfseries DFK-Wen4}%
  \\
  \vspace{-1em}
  \caption{Inpainting experiments of the missile, visualized using HSV color encoding, reveal clear differences in performance. Compared to DFK-Wen4, Curl SIREN fails to adhere to the solid boundary and lacks symmetry. Meanwhile, both DFK-Poly6 and Curl Kernel exhibit rough velocity fields due to kernel-based artifacts. We choose the same, large initial radii for all these kernels.}
  \label{fig:missile}
  \Description{missile experiments}
\end{figure}

\begin{figure}[t]
  \centering
  \newcommand{\formattedgraphics}[2]{\begin{overpic}[width=.23\linewidth,trim=5cm 15cm 5cm 12cm,clip]{#1}\put(3,53){\sffamily\scriptsize\color{white} #2}\end{overpic}}
  \formattedgraphics{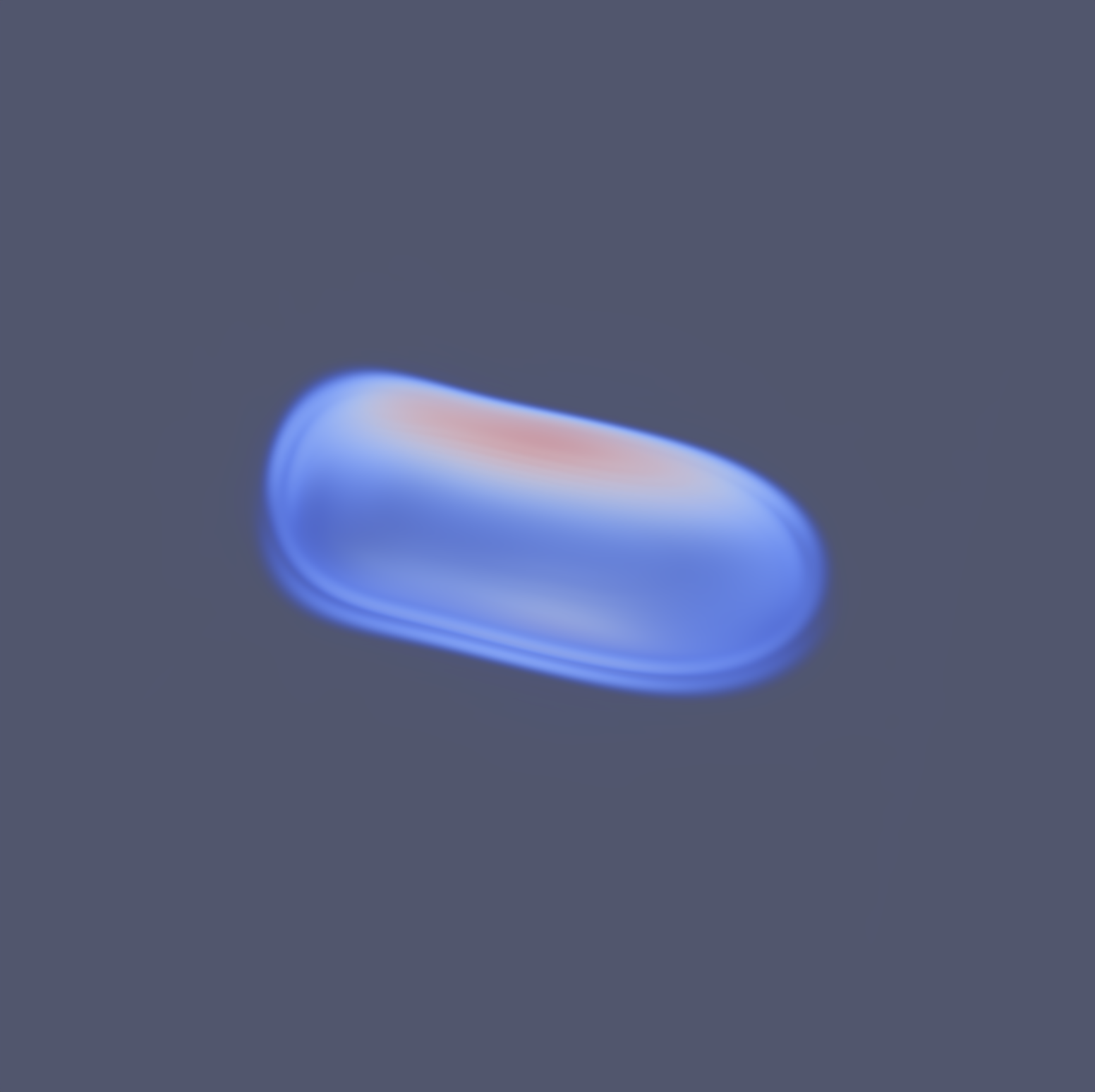}{Curl SIREN}%
  \formattedgraphics{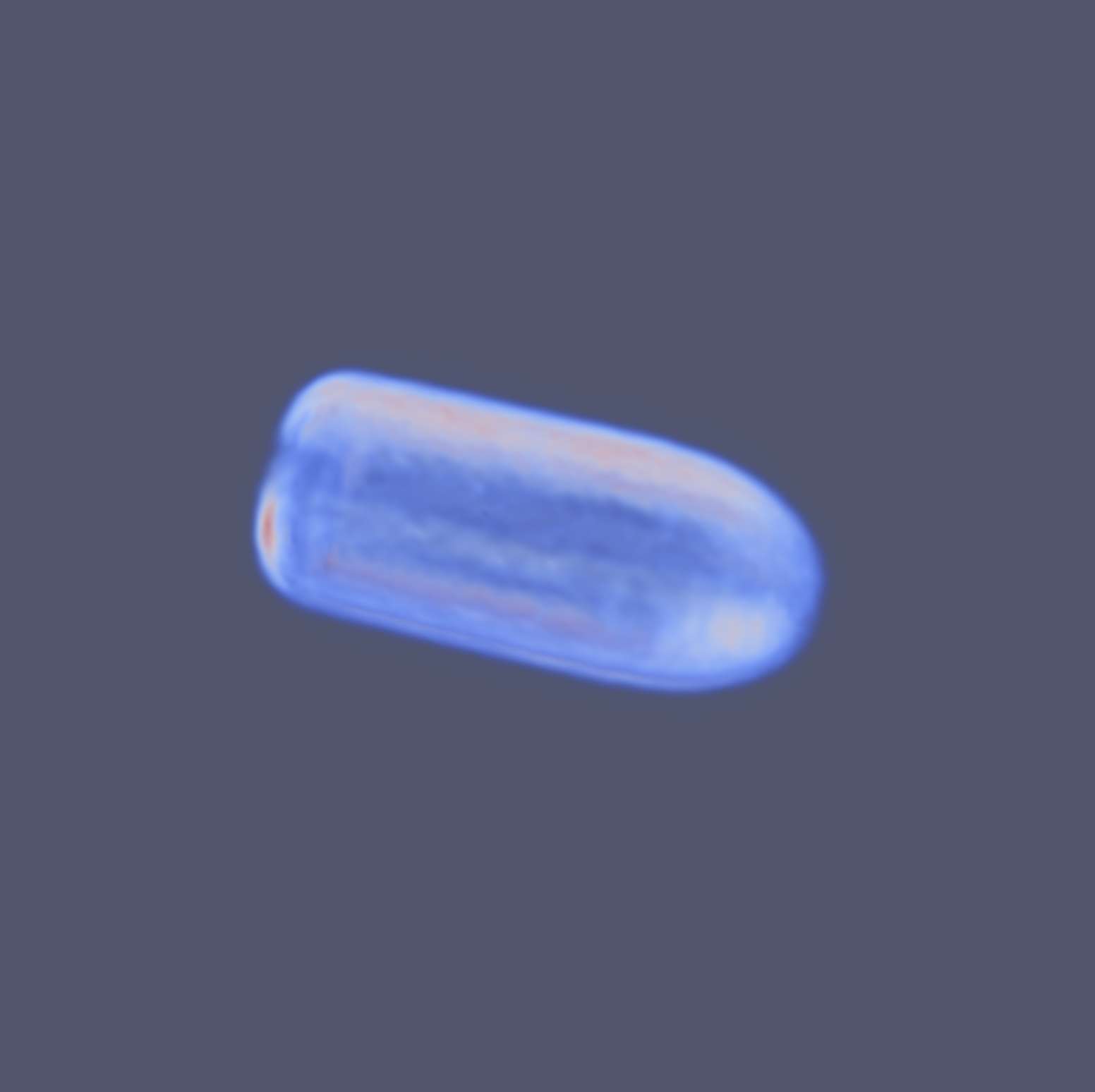}{Curl Kernel}%
  \formattedgraphics{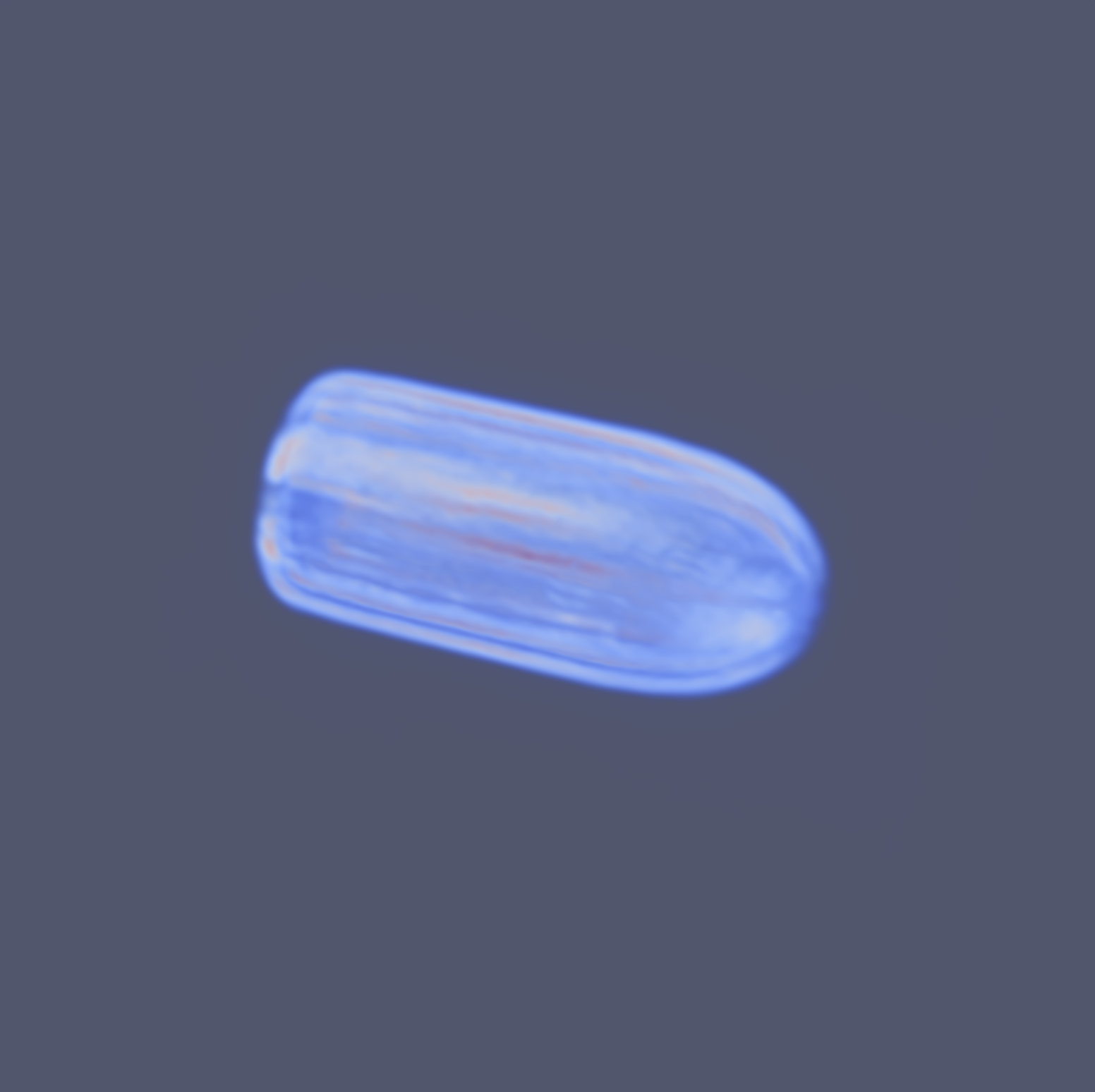}{\bfseries DFK-Wen4}%
  \\
  \vspace{-1em}
  \caption{Inpainting experiments of the bullet with the vorticity fields rendered. Curl SIREN fails to keep the rotational symmetry and performs poorly in fitting the shape of the solid boundary, while the result of Curl Kernel is very rough. DFK-Wen4 has the lowest fitting loss (see Tab.~\ref{tab:losses}) and overcomes these shortcomings.}
  \label{fig:bullet}
  \Description{bullet experiments}
\end{figure}

\subsubsection{Inpainting}
\label{sec:completion}

Inpainting is a powerful technique for interpolating unknown regions of a flow field from known data \cite{Ozdemir2024}. It enables the smooth integration of user-defined regions into existing flow fields and can automatically generate local perturbations caused by the introduction of new solid boundaries.

\paragraph{Laminar flows (2D)}
Inspired by \citet{Ozdemir2024}, we explore the task of ``stitching'' two uniform flow fields with a specified angular offset. As shown in Fig.~\ref{fig:laminar}, the regions inside the inner yellow rectangle and outside the outer yellow rectangle exhibit different velocity directions, while the area in between is the target region for intelligent inpainting. We evaluate the performance of Curl SIREN and DFK-Wen4 for completing under three angular offsets: \SI{0}{\degree}, \SI{45}{\degree}, and \SI{90}{\degree}. The results demonstrate that DFKs-Wen4 (\num{5416} kernels) not only better preserve the known velocity field but also generate more symmetric and accurate completions compared to Curl SIREN.

\paragraph{Missile (2D)}
In wind tunnel experiments, disturbances occur in the background flow due to obstacles, where data completion can provide a high-quality initial guess.
For this scenario, we define a target completion region around a 2D missile, marked by the area between two black contours in Fig.~\ref{fig:missile}. Outside this region, the flow remains uniformly leftward, while inside, the velocity is constrained to zero due to the solid boundary. Although Curl SIREN achieves the lowest loss value, it compromises boundary adherence. In contrast, our method with \num{5390} kernels not only aligns better with the boundary but also produces the smoothest and most symmetric inpainting, providing a more physically realistic flow reconstruction.

\paragraph{Bullet (3D)}
We conduct experiments to complete the flow field around a 3D bullet as well. As shown in Fig.~\ref{fig:bullet}, Curl SIREN fails to maintain symmetry, and the result of Curl Kernel is very rough.  DFK-Wen4 (\num{25325} kernels) generates the smoothest and most symmetric result. 

\subsubsection{Super-Resolution}
\label{sec:super-resolution}
Flow field super-resolution is an effective technique for overcoming challenges of data collection and storage while improving the accuracy and computational efficiency of numerical simulations.
By incorporating the continuity equation of incompressible flows as a physical prior, this method enables the recovery of fine-scale fluid details, even from highly compressed or low-resolution data.

\paragraph{Turbulence A/B (2D)}
We generate two turbulent flow fields caused by multiple velocity sources following \citet{Stam1999} on a $512\times512$ Cartesian grid, and then downsampled them to $64\times64$ data points.
As shown in Figs. \ref{fig:multisrc} and \ref{fig:mixing}, the DFK-Wen4 method excels at recovering even the finest details of the original fields, preserving both the smoothness and clarity of the flow structures.

\paragraph{Spherical obstacle (3D)}
To assess super-resolution with obstacles, we test using $32^3$ data sampled from a smoke simulation on a $128^3$ grid, where a sphere is positioned at the center.
Our method with \num{42002} points initialized delivers the best accuracy as shown in Tab.~\ref{tab:losses} and Fig.~\ref{fig:obsplume}.

\begin{figure}[t]
  \centering
  \setlength{\imagewidth}{0.163\textwidth}
    \newcommand{\formattedgraphics}[2]{%
      \begin{tikzpicture}
        \node[anchor=south west, inner sep=0] at (0,0){\includegraphics[width=\imagewidth,trim=7cm 10cm 3cm 0,clip]{#1}};
        \draw[red] (0.25\imagewidth, 0.2\imagewidth) rectangle (0.7\imagewidth, 0.65\imagewidth);
        \node[anchor=west] at (.01\imagewidth, .92\imagewidth) {\sffamily\scriptsize #2};
        \end{tikzpicture}%
    }
  \formattedgraphics{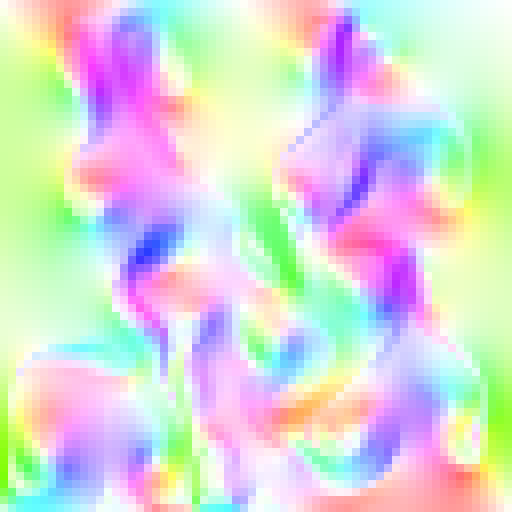}{Input}%
  \hfill
  \formattedgraphics{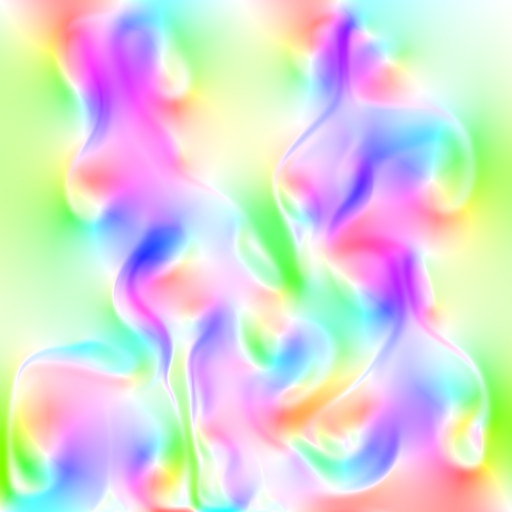}{Ground Truth}%
  \hfill
  \formattedgraphics{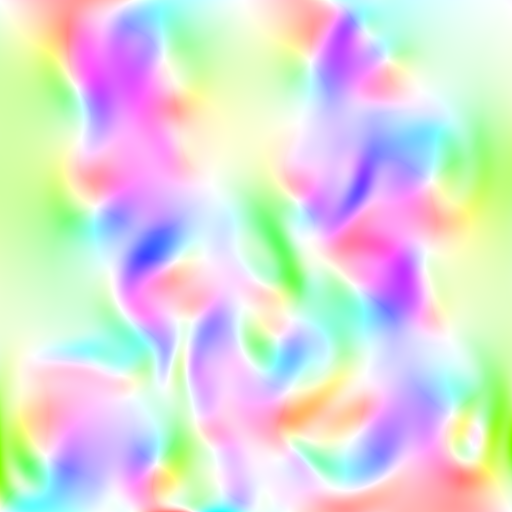}{Curl SIREN}%
  \hfill
  \formattedgraphics{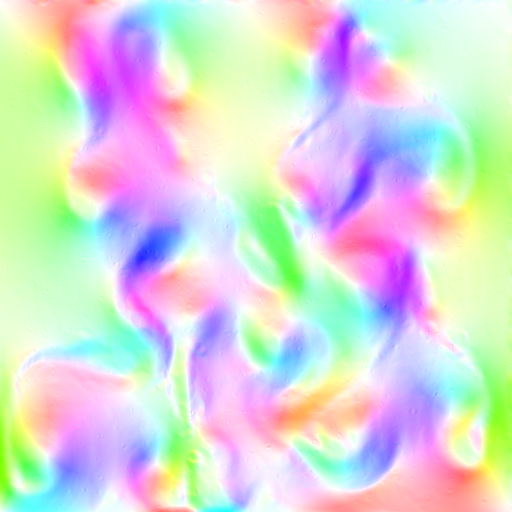}{DFK-Poly6}%
  \hfill
  \formattedgraphics{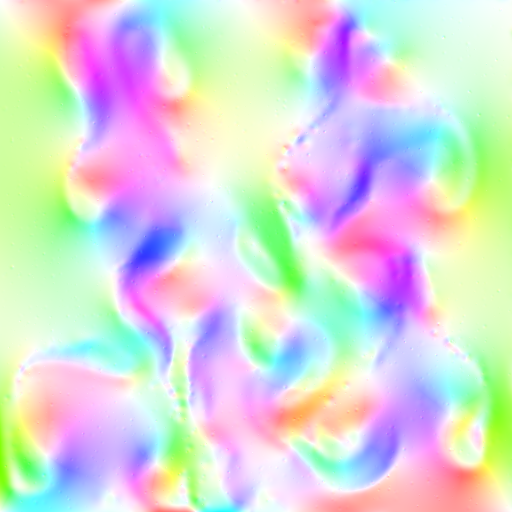}{Curl Kernel}%
  \hfill
  \formattedgraphics{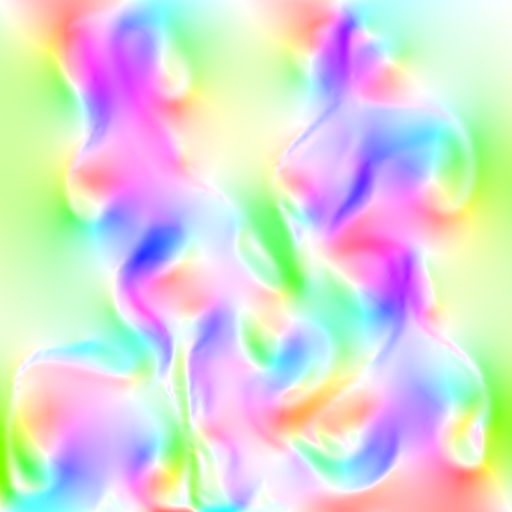}{\bfseries DFK-Wen4}%
  \\
  \vspace{-1em}
  \caption{Super-resolution experiments of the turbulence A, visualized using HSV color encoding, show notable performance differences.
  DFK-Wen4 is the only approach that recovers the extremly anisotropic, sharp structure in the red box.
  Note that the input field is very chaotic and asymmetric.
  The PSNR/SSIM values are as follows: Curl SIREN: 31.14/0.948; DFK-Poly6: 36.40/0.950; Curl Kernel: 37.08/0.965; \textbf{DFK-Wen4:} \textbf{40.93}/\textbf{0.987}.}
  \label{fig:multisrc}
  \Description{multisrc experiments}
\end{figure}
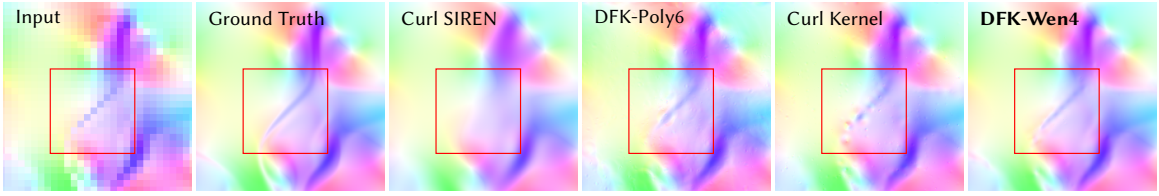

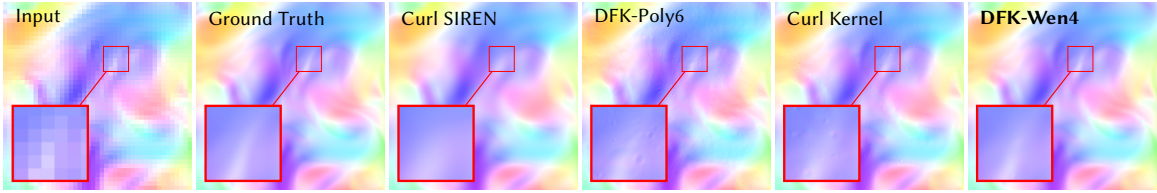
\begin{figure}[t]
  \centering
  \setlength{\imagewidth}{0.163\textwidth}
    \newcommand{\formattedgraphics}[2]{%
      \begin{tikzpicture}[spy using outlines={rectangle, magnification=3, connect spies}]
        \node[anchor=south west, inner sep=0] at (0,0){\includegraphics[width=\imagewidth,trim=4cm 2cm 2cm 4cm,clip]{#1}};
        \spy [red,size=.4\imagewidth] on (.6\imagewidth,.7\imagewidth) in node at (.25\imagewidth,.25\imagewidth);
        \node[anchor=west] at (.01\imagewidth, .92\imagewidth) {\sffamily\scriptsize #2};
        \end{tikzpicture}%
    }
  \formattedgraphics{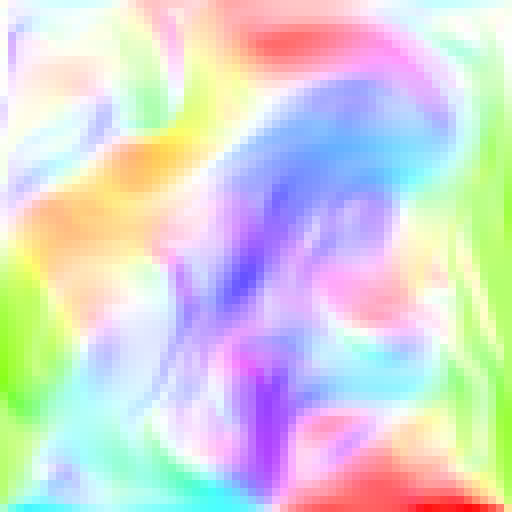}{Input}%
  \hfill
  \formattedgraphics{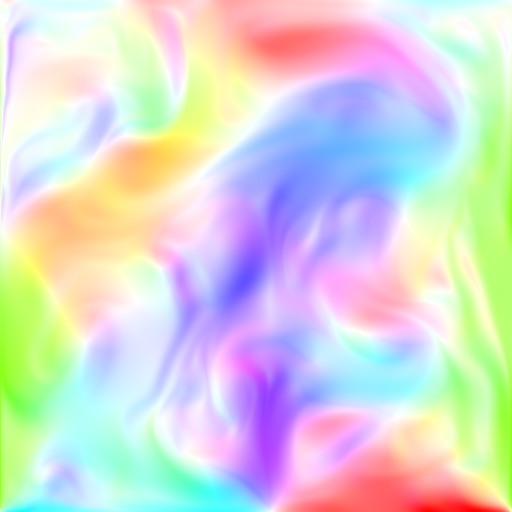}{Ground Truth}%
  \hfill
  \formattedgraphics{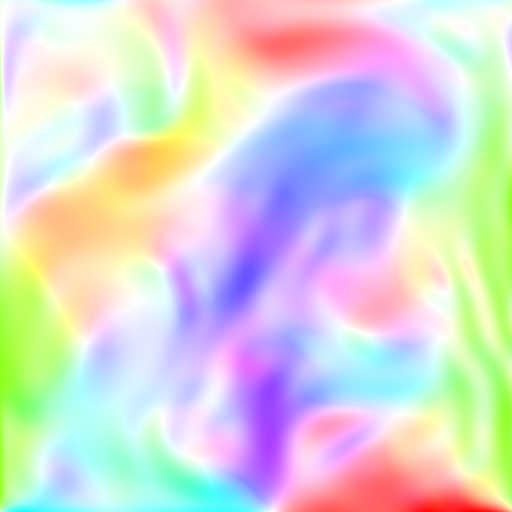}{Curl SIREN}%
  \hfill
  \formattedgraphics{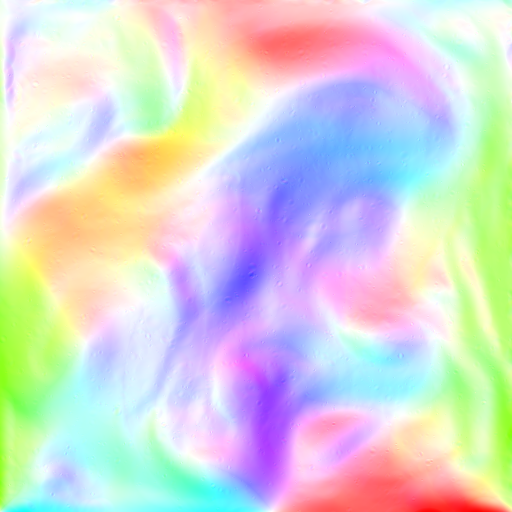}{DFK-Poly6}%
  \hfill
  \formattedgraphics{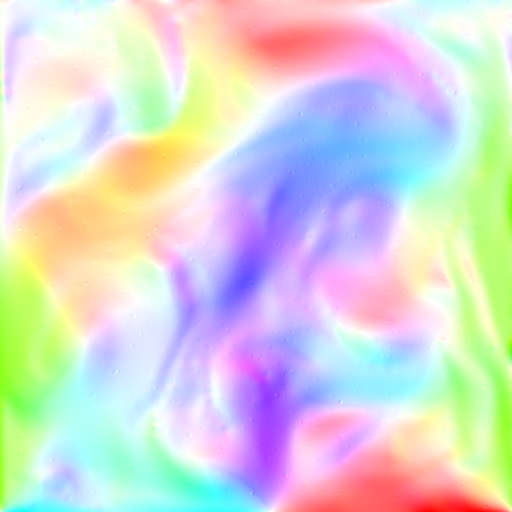}{Curl Kernel}%
  \hfill
  \formattedgraphics{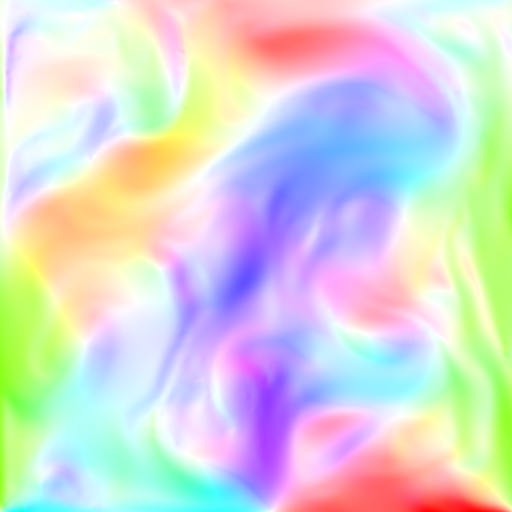}{\bfseries DFK-Wen4}%
  \\
  \vspace{-1em}
  \caption{Super-resolution experiments of turbulence B. The results are visualized using HSV color encoding, among which DFK-Poly6 and Curl Kernel, though recover the flow field globally, bring visible kernel-based artifacts.
  By optimizing DFKs-Wen4 with different radii and weights, we can successfully recover such chaotic and anisotropic patterns with high accuracy.
  The PSNR/SSIM values are as follows: Curl SIREN: 36.31/0.970; DFK-Poly6: 37.10/0.956; Curl Kernel: 38.10/0.969; \textbf{DFK-Wen4:} \textbf{41.05}/\textbf{0.991}.}
  \label{fig:mixing}
  \Description{mixing experiments}
\end{figure}

\begin{figure}[t]
  \centering
  \newcommand{\formattedgraphics}[2]{\begin{overpic}[width=.195\linewidth]{#1}\put(4,88){\sffamily\scriptsize\color{white} #2}\end{overpic}}
  \formattedgraphics{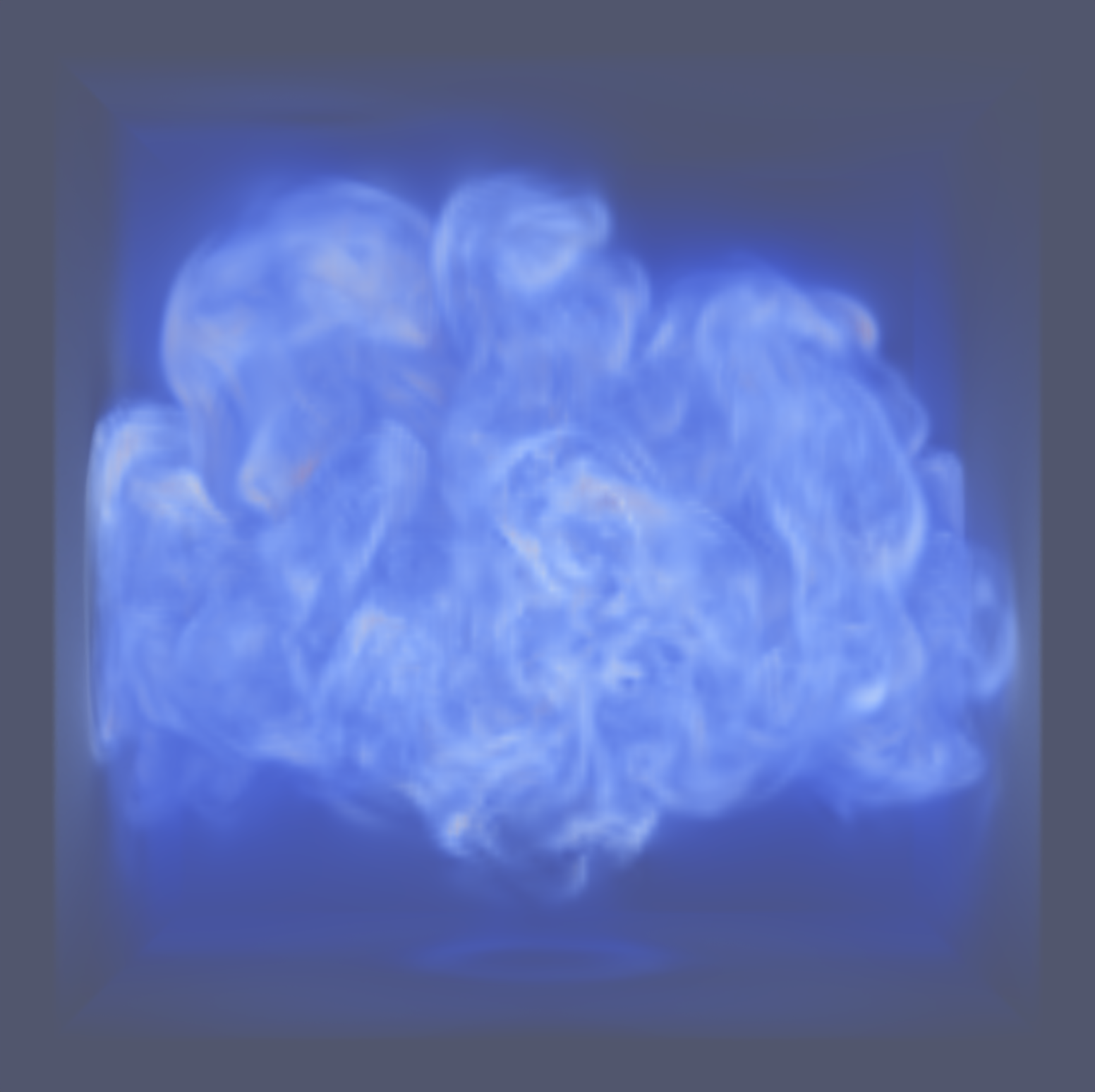}{Ground Truth}%
  \hfill
  \formattedgraphics{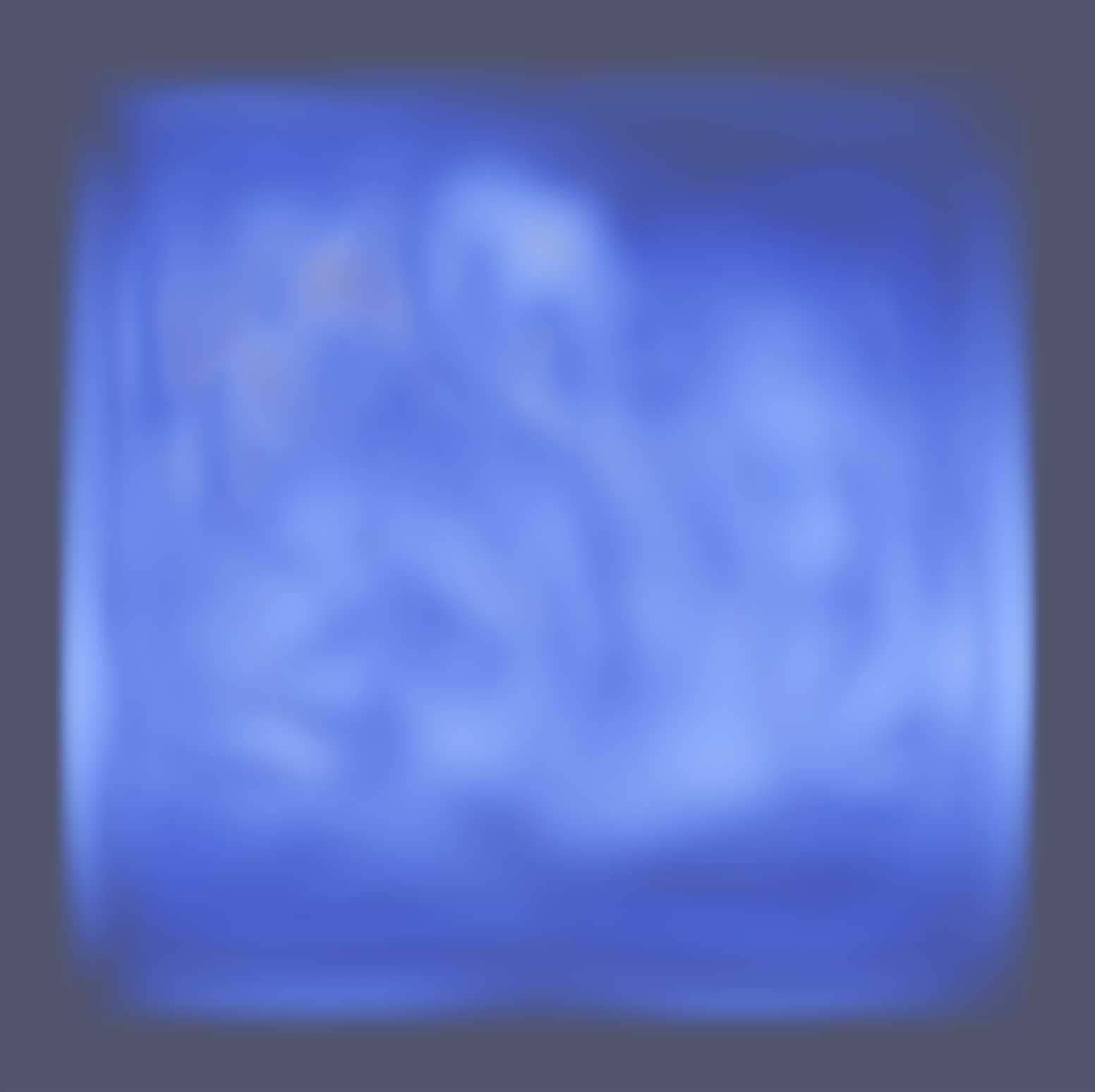}{Curl SIREN}%
  \hfill
  \formattedgraphics{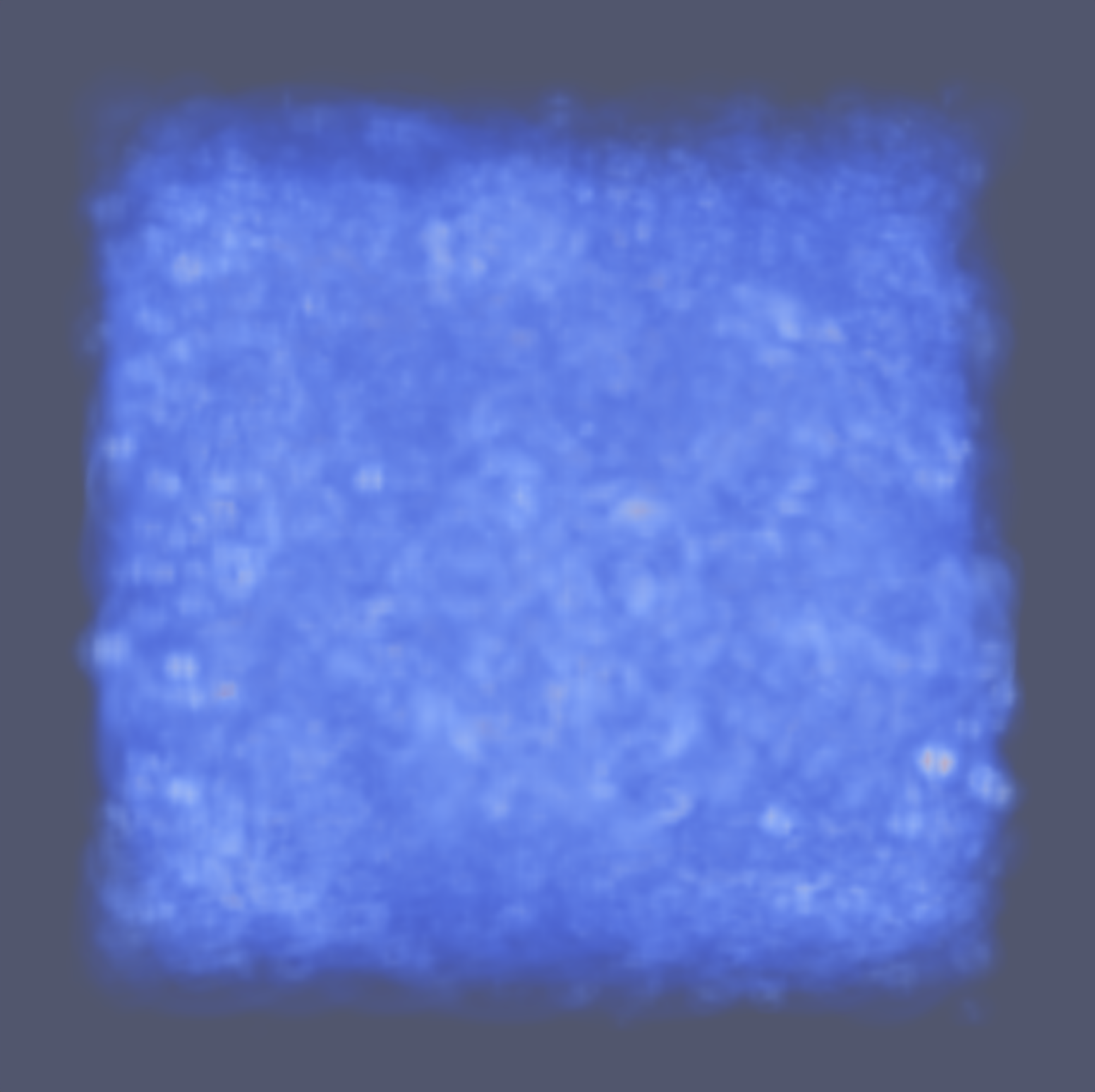}{DFK-Poly6}%
  \hfill
  \formattedgraphics{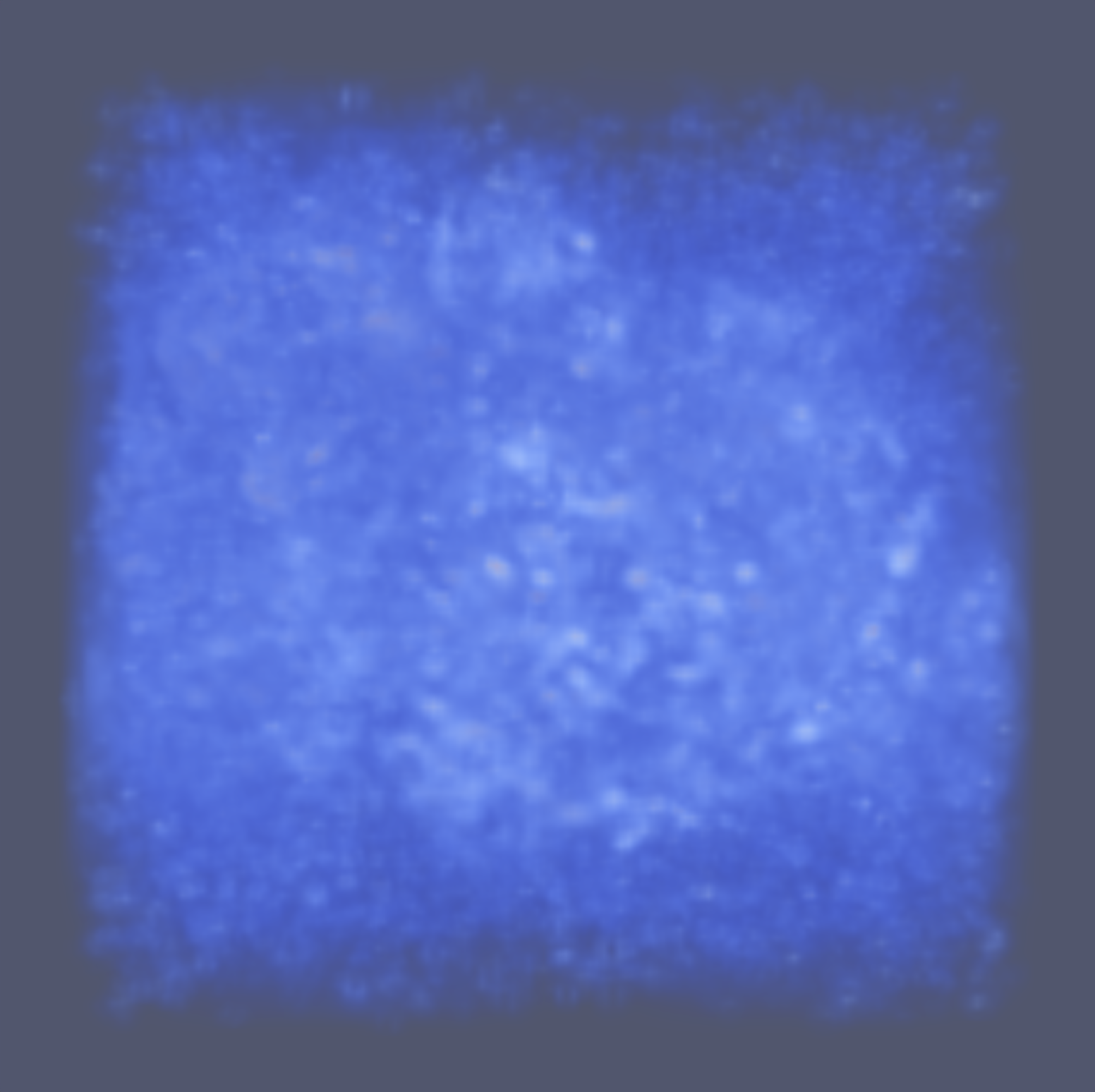}{Curl Kernel}%
  \hfill
  \formattedgraphics{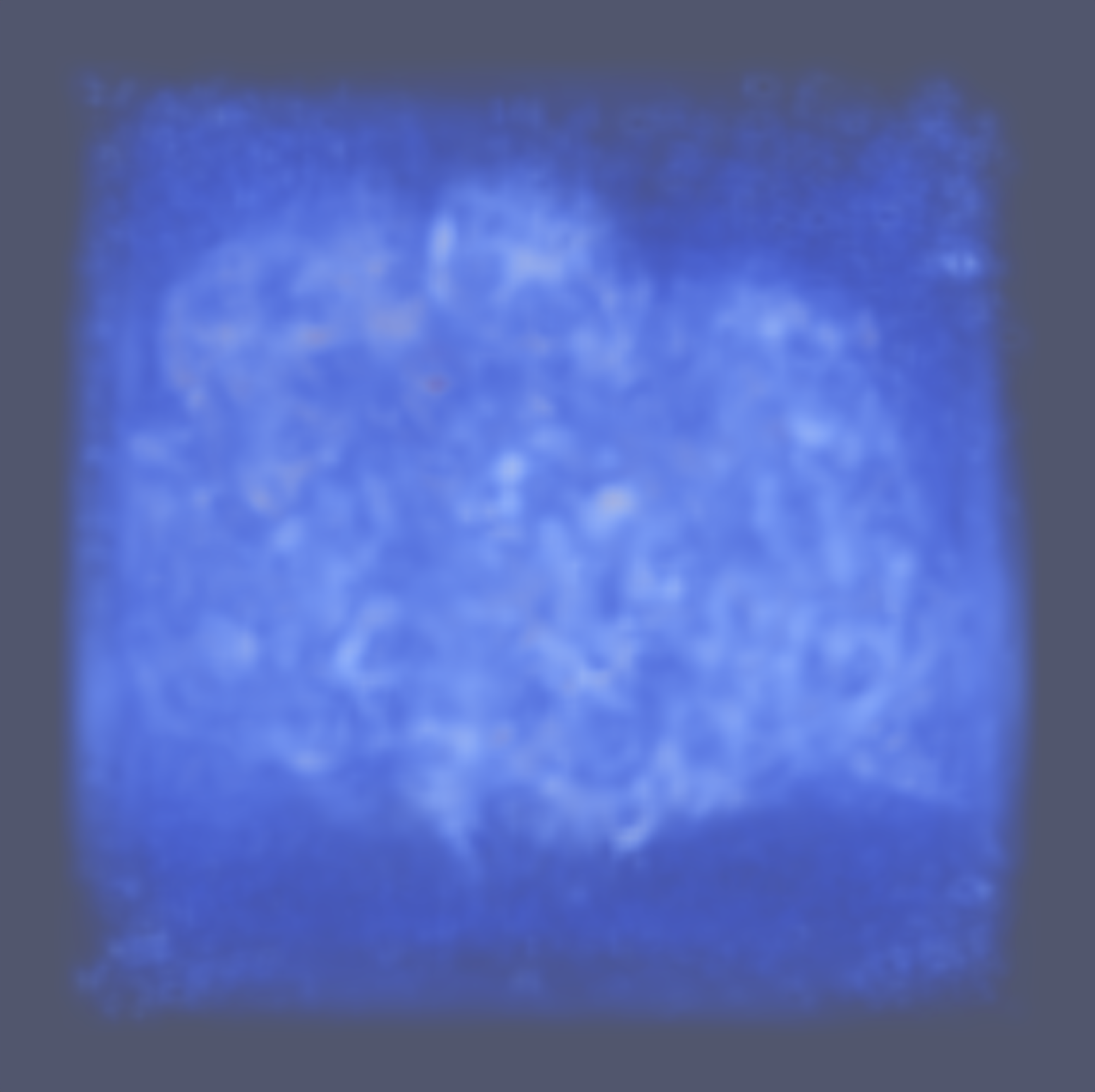}{\bfseries DFK-Wen4}%
  \\
  \vspace{-1em}
  \caption{Super-Resolution experiments of the flow field with a spherical obstacle. The vorticity fields are rendered. DFK-Wen4 shows the best super-resolution capability, especially in recovering the detailed vortices, while the result of Curl SIREN is overly diffuse and those of other kernels are noisy.}
  \label{fig:obsplume}
  \Description{obsplume experiments}
\end{figure}

\subsection{Inference of Time-Continuous Flows}

\label{sec:reconstruction}

Directly measuring a flow field is often challenging in practice. However, passive scalar fields advected by the flow, such as soot density or smoke color, are typically easier to capture with higher accuracy. When diffusion effects are negligible, the flow field can be reconstructed from the continuous evolution of these passive fields using the advection equation. In this context, the observational loss, $\mathcal{L}_\mathrm{obs}$, is defined as shown in Eq.~\eqref{eqn:advection_loss}.

To mitigate velocity noise in regions with weakly supervision and ensure the temporal continuity of the fields, we introduce regularization and discontinuity penalty terms into the loss function:
\begin{gather}
  \mathcal{L}=\mathcal{L}_\mathrm{obs}+\lambda_\mathrm{div}\mathcal{L}_\mathrm{div}+\lambda_\mathrm{bou}\mathcal{L}_\mathrm{bou}+\lambda_\mathrm{reg}\mathcal{L}_\mathrm{reg}+\lambda_\mathrm{con}\mathcal{L}_\mathrm{con}\text{,}\\
  \mathcal{L}_\mathrm{reg}=\frac{1}{V}\int_0^T\int_{\Omega}\left\Vert\bm{u}\right\Vert\mathrm{d}V\,\mathrm{d}t\text{,}\\
  \mathcal{L}_\mathrm{con}=\frac{1}{V}\int_0^T\int_{\Omega}\left\Vert\frac{\partial\bm{u}}{\partial t}\right\Vert\mathrm{d}V\,\mathrm{d}t\text{,}
\end{gather}
where the time derivative of $\bm{u}$ is estimated using central differences for second-order accuracy.
We found that setting
$\lambda_\mathrm{bou}=1$ and $\lambda_\mathrm{div}=\lambda_\mathrm{reg}=\lambda_\mathrm{con}=0.1$ provides a well-balanced weighting.
Note that for divergence-free representations, the value of $\lambda_\mathrm{div}$ is ineffective because $\mathcal{L}_\mathrm{div}$ is always zero.

In dynamic scenarios, a separate SIREN model is trained for each frame in the INR-based approaches.
In contrast, kernel-based approaches use fixed kernel positions and radii that are shared across frames, while each frame is assigned a unique set of weights.

\begin{figure}[t]
  \centering
  \newcommand{\formatteddensity}[2]{\begin{overpic}[height=.195\linewidth,trim=9cm 2cm 9cm 0cm,clip]{#1}\put(3,5){\sffamily\scriptsize #2}\end{overpic}}
  \newcommand{\formattedvelocity}[2]{\begin{overpic}[width=.13\linewidth]{#1}\put(3,5){\sffamily\scriptsize #2}\end{overpic}}
  \formatteddensity{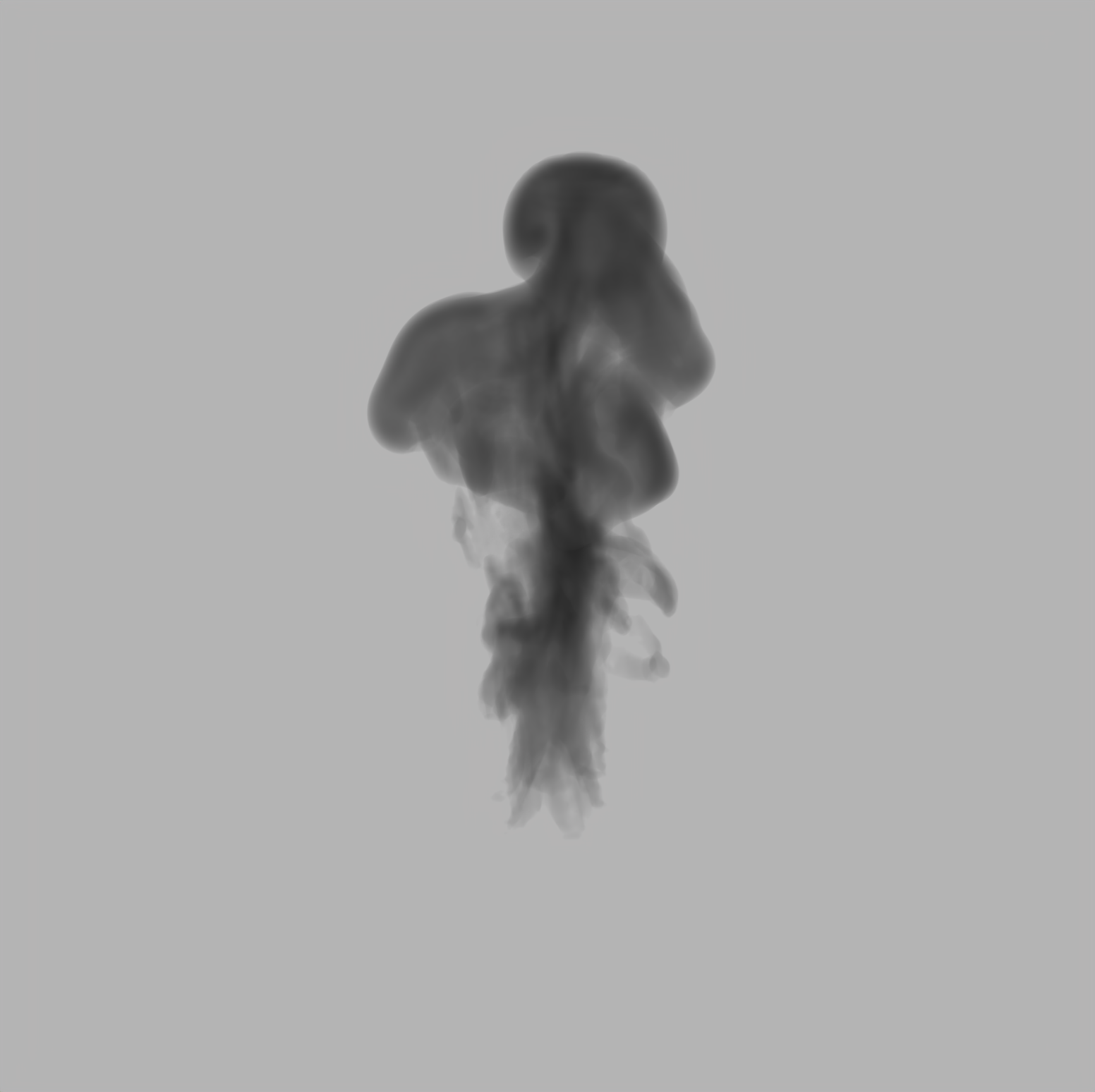}{Input ($t=\SI{1.2}{\second}$)}%
  \formattedvelocity{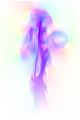}{Ground Truth}%
  \formattedvelocity{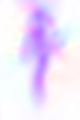}{SIREN}%
  \formattedvelocity{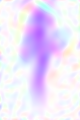}{Curl SIREN}%
  \formattedvelocity{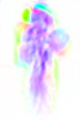}{\textbf{DFK-Wen4}}%
  \\
  \formatteddensity{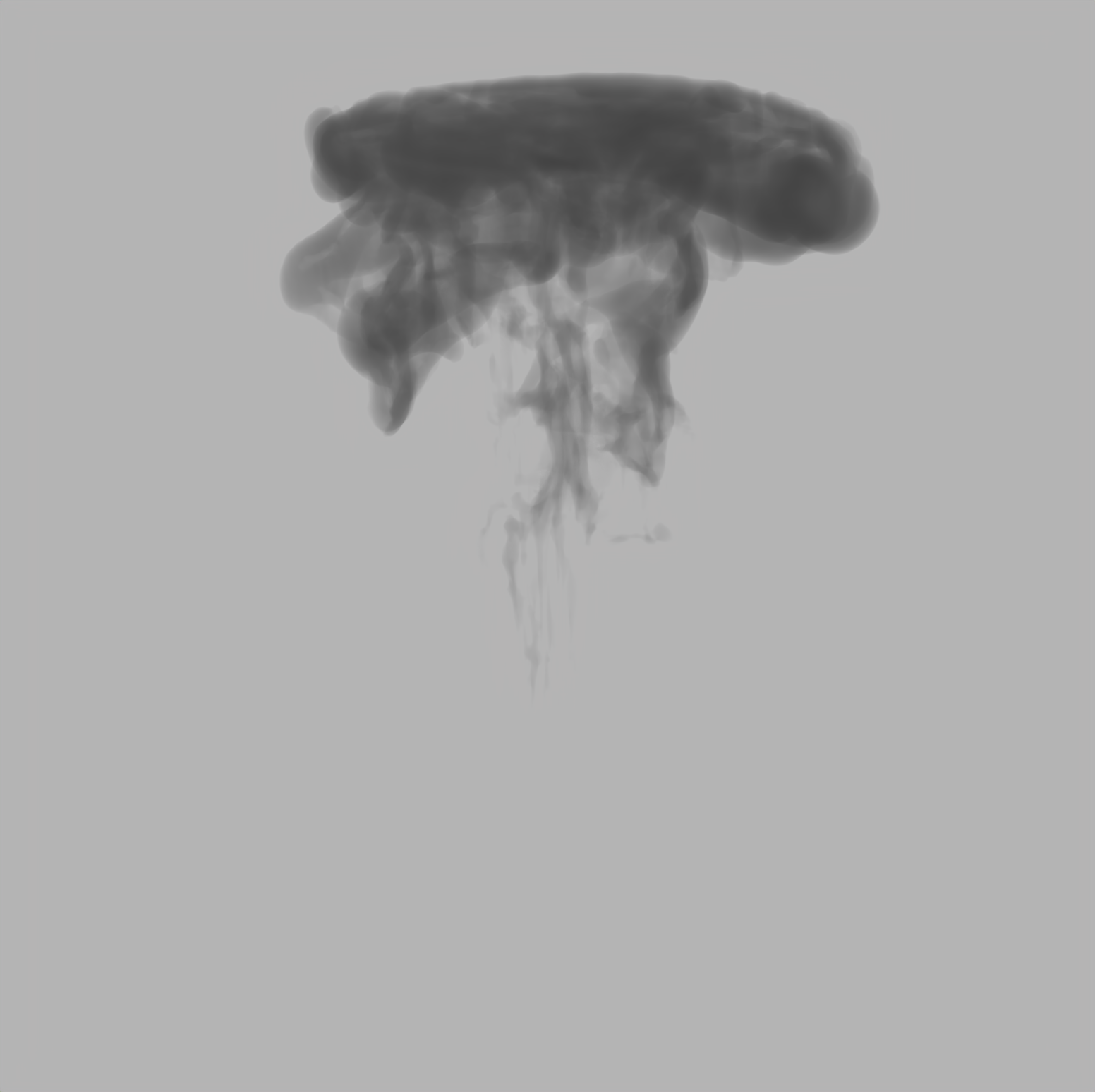}{Input ($t=\SI{3.2}{\second}$)}%
  \formattedvelocity{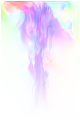}{Ground Truth}%
  \formattedvelocity{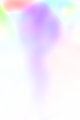}{SIREN}%
  \formattedvelocity{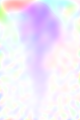}{Curl SIREN}%
  \formattedvelocity{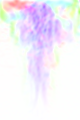}{\textbf{DFK-Wen4}}%
  \\
  \vspace{-1em}
  \caption{Inference experiments of the rising smoke. The leftmost image depicts the rendered visualization of the input density field. The reconstructed velocity fields are shown using HSV color encoding, presented as slices viewed along the positive $x$-axis. The results demonstrate the detailed inference capabilities of DFK-Wen4, which uniquely captures rich flow structures.}
  \label{fig:rising}
  \Description{rising experiments}
\end{figure}

\subsubsection{From Passive Field Data}
\label{sec:inference-data}
We compare the performance of SIREN, Curl SIREN, and DFK-Wen4 in inferring the flow field from complete and accurate passive field data.

\paragraph{Rising (3D)}
We perform a smoke simulation on a $80\times120\times80$ Cartesian grid for $150$ frames and use the multi-frame density data (soot concentration) to infer the dynamic flow field. As illustrated in Fig.~\ref{fig:rising}, the DFK-Wen4 method (\num{26375} kernels, only $40\%$ of parameters compared to INRs) uniquely succeeds in inferring detailed vortices within the field, while network-based methods produce overly smooth results, missing critical fine-grained structures.

\paragraph{Teapot (3D)}
In a similar setup, the passive field data is acquired from the simulated motion of smoke over $100$ frames on a $192\times128\times128$ grid, with a teapot placed at the center of the scene (Fig.~\ref{fig:teapot}).
As demonstrated in the figure, while SIREN-based methods capture the global behaviors of the fluid, they fail to account for the presence of the obstacles and struggle to reconstruct the thin structures within the flow field.
Despite the hard divergence-free constraint, our method achieves $30\%$ less advection loss than INRs with only $60\%$ of trainable parameters.

\subsubsection{From Multi-View Videos}
\label{sec:inference-videos}
By integrating with a NeRF \cite{Mildenhall2021} frontend, our method can also infer the flow field from multi-view videos sequences of dynamic smoke. Here, the implementation follows PINF-Smoke \cite{Chu2022}.

\paragraph{Scalar flow (3D)}
For inferring flow fields from real captures, we use the ScalarFlow dataset \cite{Eckert2019}, which contains videos of buoyancy-driven rising smoke plumes.
Five cameras with fixed positions were evenly distributed across a \SI{120}{\degree} arc centered on the rising smoke. We use the middle $120$ frames from each camera view for fluid reconstruction. 
After training of PINF-Smoke, its resulting density field is fed as input into our method for further processing.
Fig.~\ref{fig:scalar} illustrates the reconstruction results, in which our method shows the best ability in capturing local vortices,
while the others only recover a general rising direction. The total loss of DFK-Wen4 is roughly $20\%$ less with only $60\%$ of trainable parameters compared to INRs.

\begin{figure}[t]
  \centering
  \setlength{\imagewidth}{.171\linewidth}
  \newcommand{\formatteddensity}[2]{\begin{overpic}[height=.113\linewidth,trim=0cm 6cm 0cm 6cm,clip]{#1}\put(3,5){\sffamily\tiny #2}\end{overpic}}
  \newcommand{\formattedvelocity}[2]{%
  \begin{overpic}[height=.113\linewidth]{#1}\put(3,5){\sffamily\tiny #2}\end{overpic}%
  }
  \formatteddensity{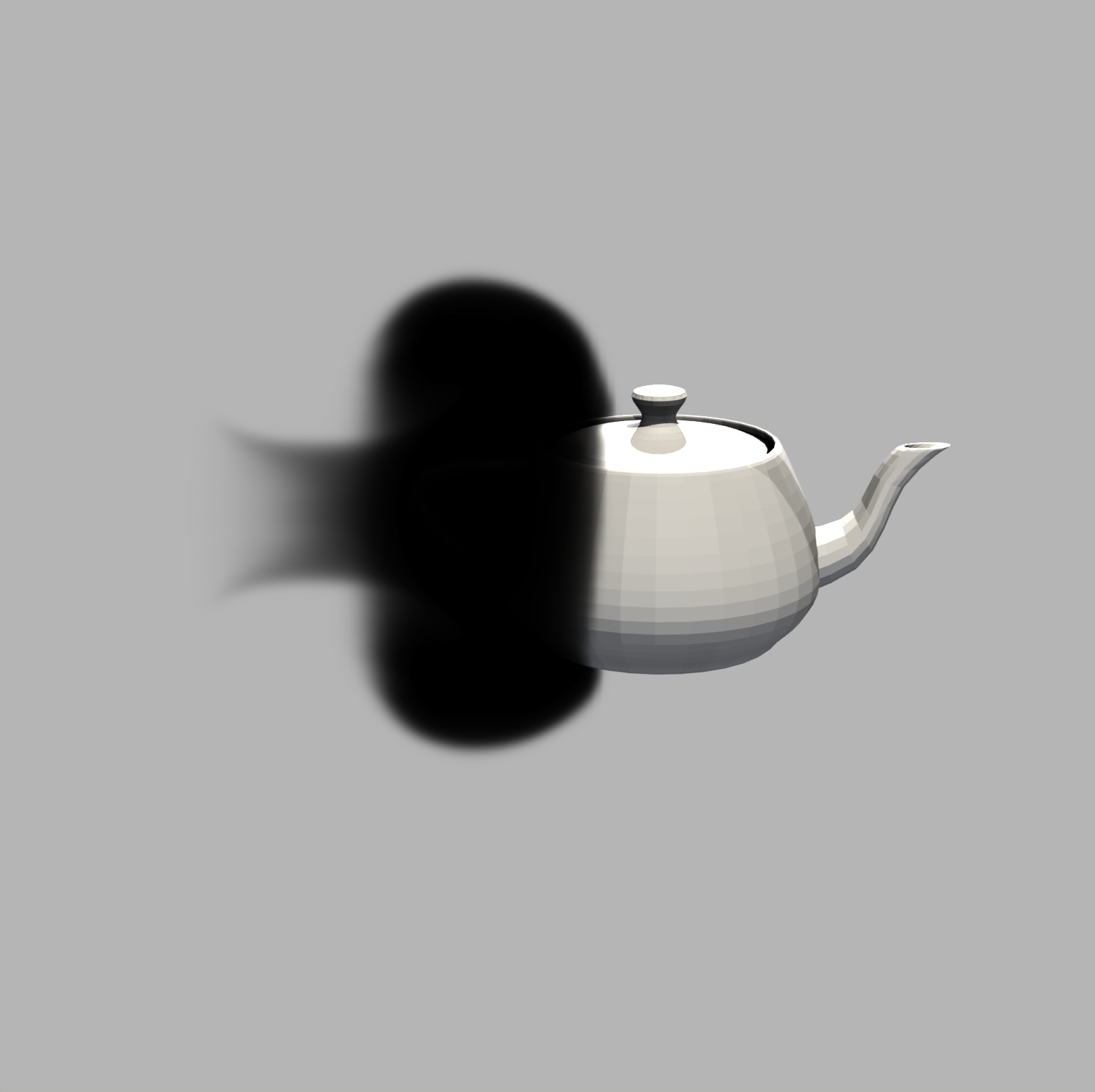}{Input ($t=\SI{0.92}{\second}$)}%
  \hfill
  \formatteddensity{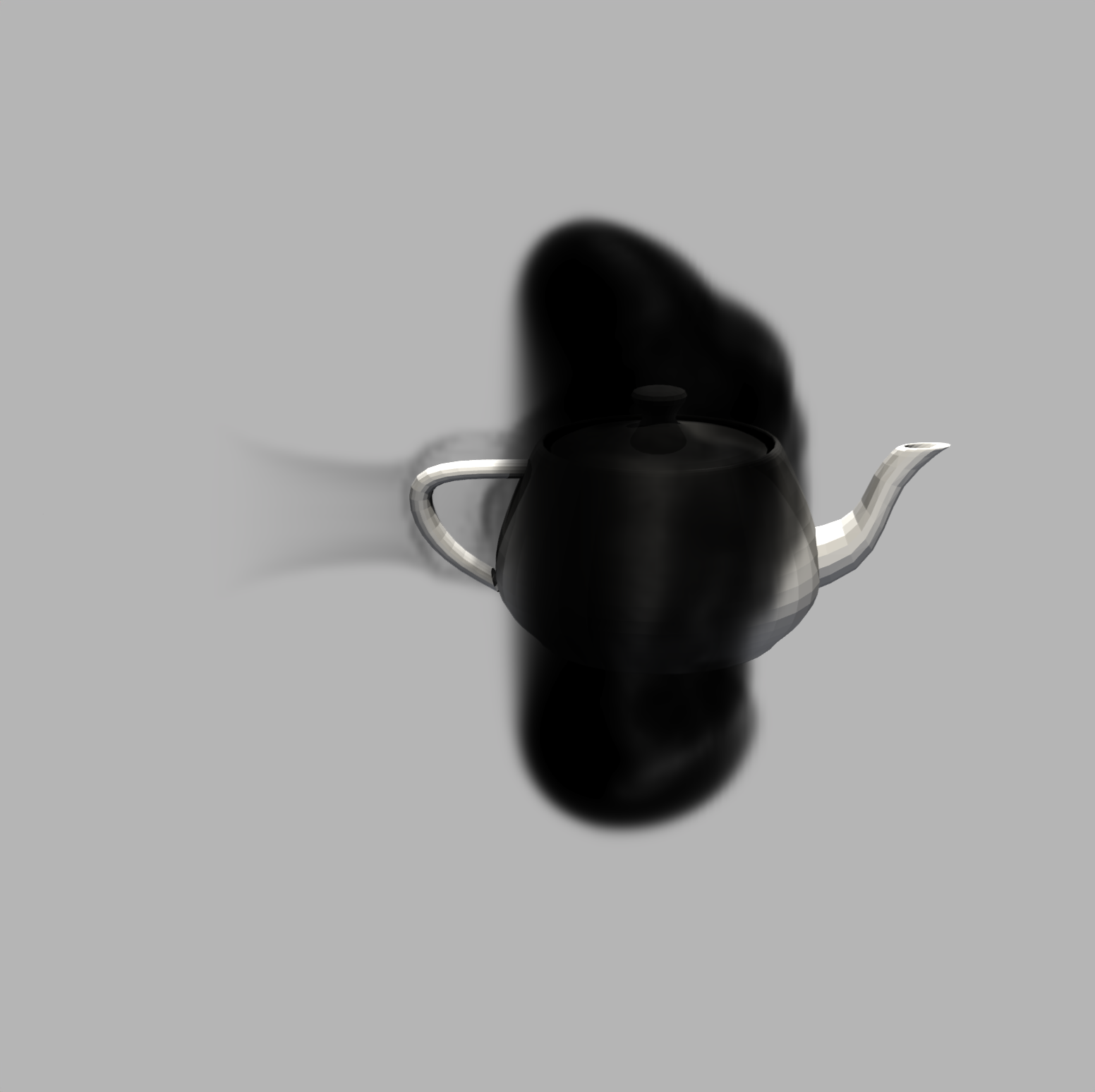}{Input ($t=\SI{2.0}{\second}$)}%
  \hfill
  \formattedvelocity{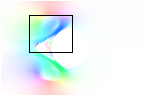}{G.T. ($t=\SI{1.2}{\second}$)}%
  \hfill
  \formattedvelocity{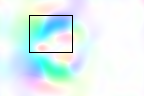}{SIREN ($t=\SI{1.2}{\second}$)}%
  \hfill
  \formattedvelocity{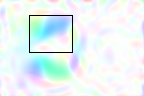}{Curl SIREN ($t=\SI{1.2}{\second}$)}%
  \hfill
  \formattedvelocity{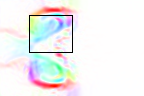}{\textbf{DFK-Wen4} ($t=\SI{1.2}{\second}$)}%
  \\
  \vspace{-1em}
  \caption{Inference experiments of the teapot. The two leftmost images illustrate the motion direction of the smoke and the position of the teapot. The flow fields are visualized using the HSV color encoding, with slices viewed along the negative $y$-axis (i.e., top-down view). Among the methods, only DFK-Wen4 accurately captures the influence of the obstacle, recovering the thin structures around it (marked with a black frame).
  }
  \label{fig:teapot}
  \Description{teapot experiments}
\end{figure}

\begin{figure}[t]
  \centering
  \newcommand{\formatteddensity}[2]{\begin{overpic}[height=.239\linewidth,trim=0cm 0cm 0cm 7cm,clip]{#1}\put(3,92){\sffamily\tiny\color{white} #2}\end{overpic}}
  \newcommand{\formattedvelocity}[2]{\begin{overpic}[height=.239\linewidth,trim=0cm 0cm 0cm 0cm,clip]{#1}\put(3,5){\sffamily\tiny #2}\end{overpic}}
  \formatteddensity{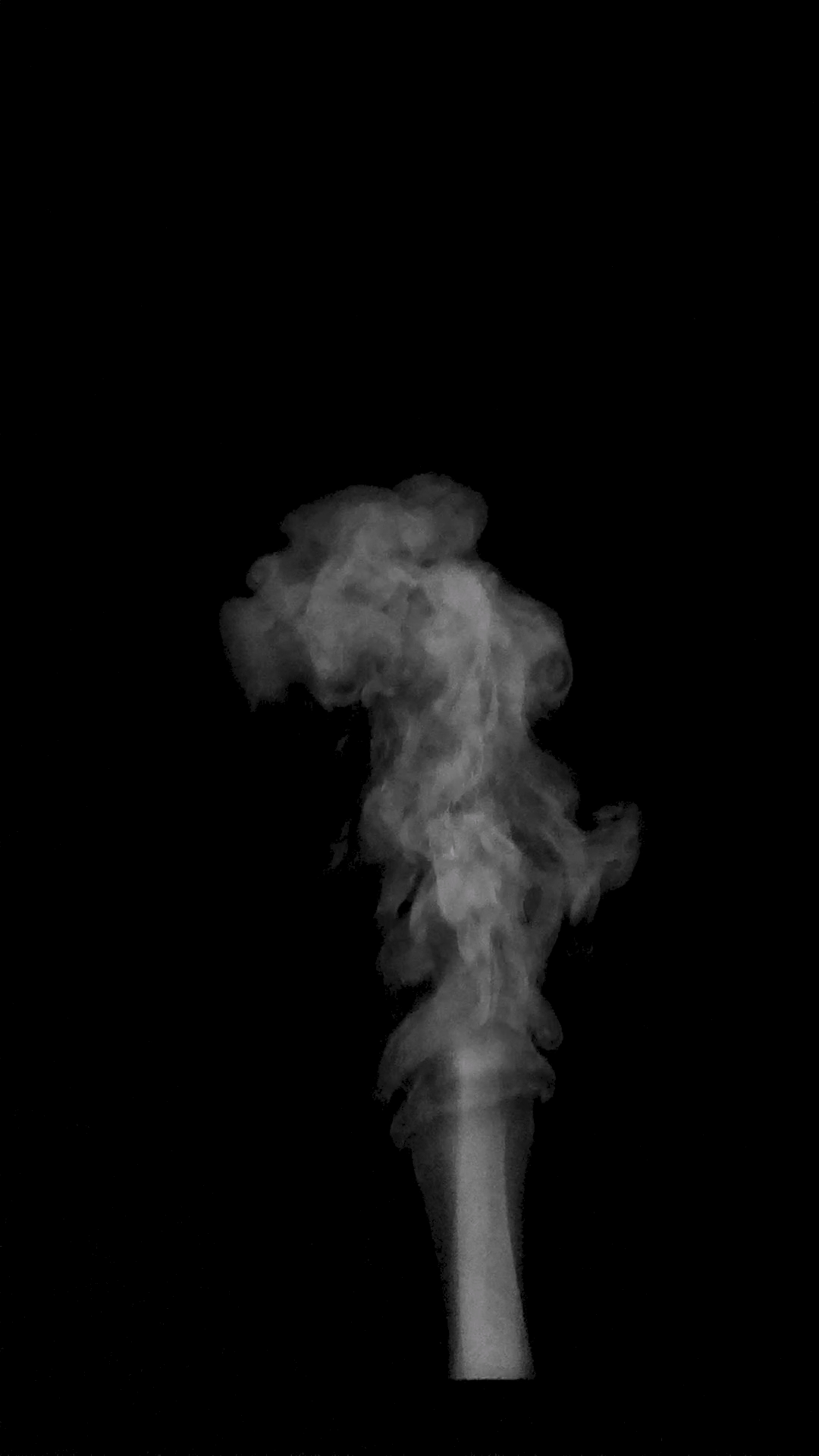}{Camera 1 ($t=\SI{4.8}{\second}$)}%
  \hfill
  \formatteddensity{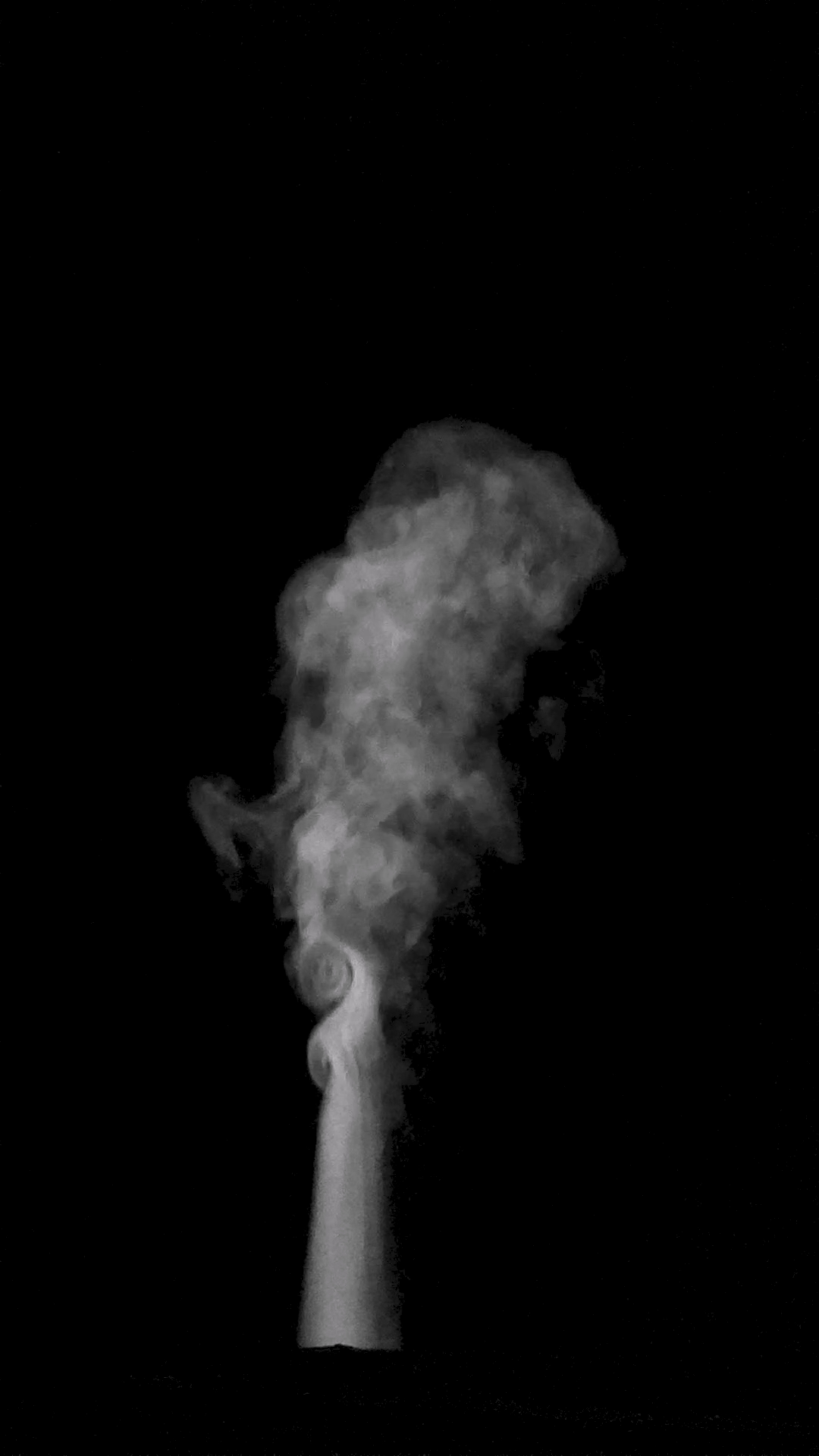}{Camera 2 ($t=\SI{4.8}{\second}$)}%
  \hfill
  \begin{overpic}[height=.239\linewidth,trim=6cm 2cm 6cm 0cm,clip]{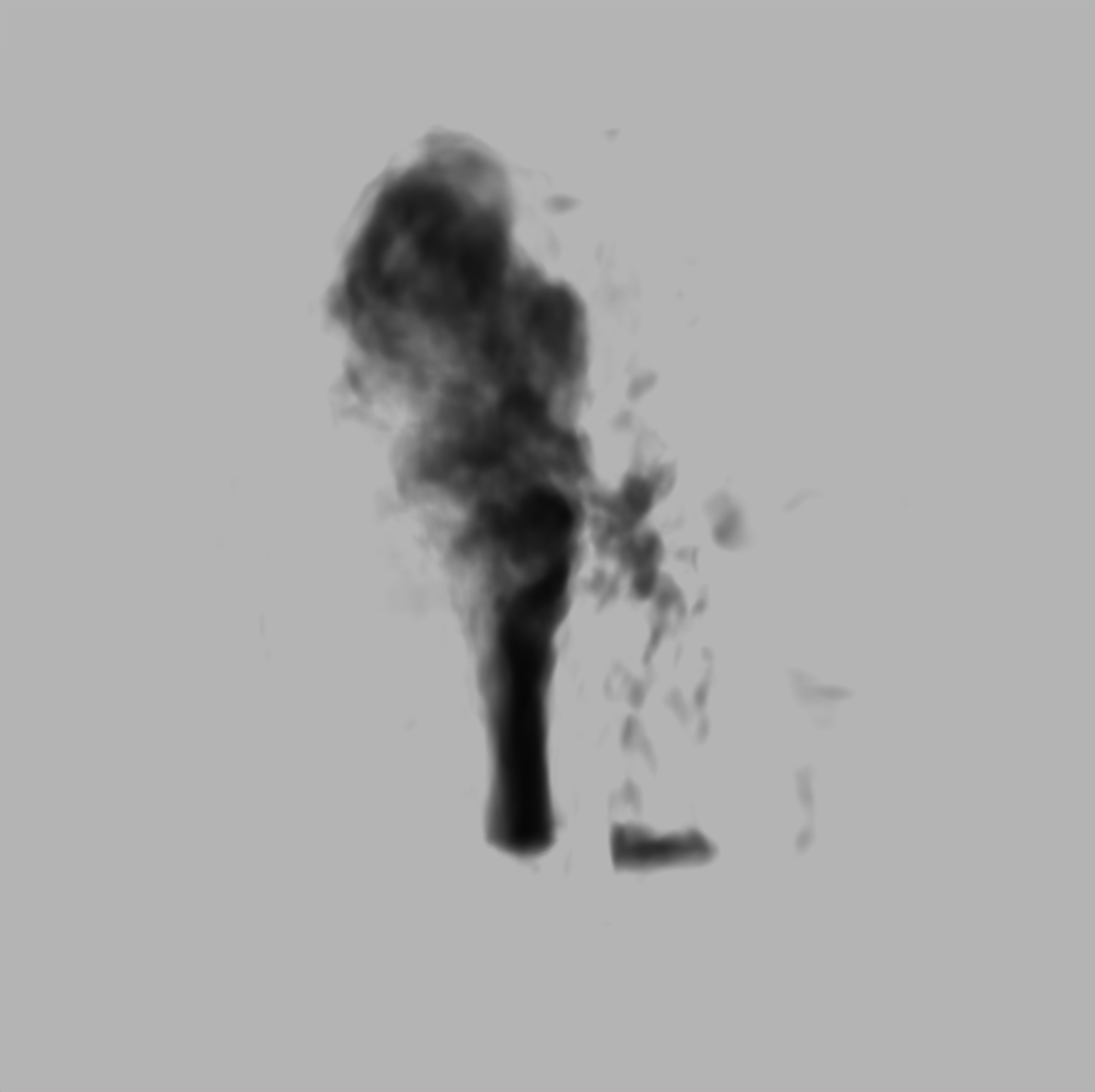}\put(3,5){\sffamily\tiny Input ($t=\SI{4.8}{\second}$)}\end{overpic}
  \hfill
  \formattedvelocity{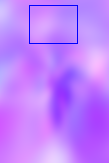}{PINF Smoke ($t=\SI{4.8}{\second}$)}%
  \hfill
  \formattedvelocity{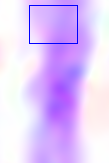}{SIREN ($t=\SI{4.8}{\second}$)}%
  \hfill
  \formattedvelocity{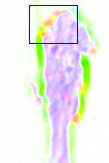}{\textbf{DFK-Wen4} ($t=\SI{4.8}{\second}$)}%
  \\
  \vspace{-1em}
  \caption{Inference experiments of the scalar flow. The two leftmost images are real-world captured from two of the five cameras, while the third one presents the rendered visualization of the input density field reconstructed using NeRF.
  Among the comparison results, only DFK-Wen4 recovers detailed vorticities at the edge of the smoke (marked with a blue frame). Note that PINF Smoke does not include the regularization term used in our loss function.}
  \label{fig:scalar}
  \Description{scalar experiments}
\end{figure}

\section{Conclusion \& Discussions}

In this paper, we develop a new kernel-based framework for high-fidelity and physically accurate reconstruction of incompressible flow fields. Compared to recent approaches, our framework uniquely employs matrix-valued kernel functions---specifically DFKs-Wen4---as analytically divergence-free approximations of the velocity field.
DFKs-Wen4 offer several distinct advantages over alternative representations of incompressible velocity fields: their dipole nature aligns with fundamental flow solutions, enabling accurate modeling of core fluid dynamics; their compact support is ideal for capturing local vortex structures and handling boundary conditions; their positive definiteness ensures superior convergence rates, while their differentiability allows for smooth and continuous modeling of flow fields.
Across a variety of flow field reconstruction tasks, from fitting and extrapolation to inference, the framework employing DFKs-Wen4 consistently outperforms competing approaches. Consequently, we believe DFKs-Wen4 hold great promise as a foundational representation capable of replacing INRs for flow field reconstruction.

Beyond their application in inverse problems, DFKs-Wen4, with their pointwise divergence-free expressive power, strong fitting capabilities, and rapid convergence, exhibit exceptional potential for forward simulations. Experiments in \S\ref{sec:projection} have demonstrated their effectiveness in pressure computation, suggesting the feasibility of a fully pointwise incompressible mesh-free fluid simulation framework. Such a framework could unify representations for both forward and inverse tasks, fostering seamless integration of fluid generation and understanding.

\paragraph{Limitations}
A key limitation of this work is that our validation data is mainly derived from numerical simulations rather than real-world measurements.
Expanding the DFK-Wen4 representation to tackle interdisciplinary applications involving real-world data is a natural next step.
Besides, for simplicity, only no-slip boundary conditions are considered in this paper. 
While free-slip conditions can be incorporated by introducing dot products into $\mathcal{L}_\mathrm{bou}$, modeling free-surface flows remains an open challenge that warrants further investigation.

\paragraph{Future work}
Currently, the RBF underlying DFKs-Wen4 (i.e., Wendland's $\mathcal{C}^4$ polynomial) is isotropic. Inspired by recent advancements in 3DGS \cite{Kerbl2023}, adapting DFKs to use anisotropic ellipsoids could further enhance their expressive capabilities.
Moreover, the DFKs presented in this paper are restricted to flow fields in Euclidean space.
Extending DFKs to surface-based kernels \cite{Narcowich2007} and applying them to flows on manifolds%
---such as those found in soap films or planetary atmospheres---opens up a compelling avenue for future exploration.

\bibliographystyle{ACM-Reference-Format}
\bibliography{refs}

\end{document}


\title{Representing Flow Fields with Divergence-Free Kernels for Reconstruction (Supplementary Document)}

\author{Xingyu Ni}
\email{nixy@pku.edu.cn}
\orcid{0000-0003-1127-2848}
\affiliation{
\institution{School of Computer Science, Peking University}
\city{Beijing}
\country{China}
}

\author{Jingrui Xing}
\email{xjr01@hotmail.com}
\orcid{0000-0001-7219-9969}
\affiliation{
\institution{School of Intelligence Science and Technology, Peking University}
\city{Beijing}
\country{China}
}

\author{Xingqiao Li}
\email{lixingqiao@pku.edu.cn}
\orcid{0000-0002-8131-6140}
\affiliation{
\institution{School of Intelligence Science and Technology, Peking University}
\city{Beijing}
\country{China}
}

\author{Bin Wang}
\email{binwangbuaa@gmail.com}
\orcid{0000-0001-9496-772X}
\affiliation{
\institution{Independent}
\city{Beijing}
\country{China}
}
\authornote{corresponding authors}

\author{Baoquan Chen}
\email{baoquan@pku.edu.cn}
\orcid{0000-0003-4702-036X}
\affiliation{
\institution{State Key Laboratory of General Artificial Intelligence, Peking University}
\city{Beijing}
\country{China}
}
\authornotemark[1]

\maketitle

\appendix

\section{Matrix-Valued Radial Basis Functions}

As proposed by \citet{Narcowich1994}, matrix-valued radial basis functions (RBFs) are used to represent a vector field $\bm{f}:\mathbb{R}^d\to\mathbb{R}^d$ by
\begin{equation}
  \widetilde{\bm{f}}(\bm{x})=\sum_{k=1}^N{\bm{\psi}(\bm{x}-\bm{x}_k)\,\bm{\alpha}_k}\text{,}
\end{equation}
where $S=\{\bm{x}_1,\bm{x}_2,\ldots,\bm{x}_N\}$ represents a set of scattered points, with corresponding vector weights $\bm{\omega}_1,\bm{\omega}_2,\ldots,\bm{\omega}_N\in\mathbb{R}^d$,
and $\bm{\psi}:\mathbb{R}^d\to\mathbb{R}^{d\times d}$ is the matrix-valued kernel.

\subsection{Positive Definiteness}

A continuous complex-valued matrix function $\psi:\mathbb{R}^d\to\mathbb{C}^{n\times n}$ is said to be \emph{positive semi-definite} if, for any set of pairwise distinct points $S=\{\bm{x}_1,\bm{x}_2,\ldots,\bm{x}_N\}\subset\mathbb{R}^d$ and any set of vector coefficients $\bm{\alpha}_1,\bm{\alpha}_2,\ldots,\bm{\alpha}_N\in\mathbb{C}^n$, the following inequality holds:
\begin{equation}
  \label{eqn:pd_func_matrix}
  0\le\sum_{j=1}^N{\sum_{k=1}^N{\overline{\bm{\alpha}_j}^{\top}\bm{\psi}(\bm{x}_j-\bm{x}_k)\,\bm{\alpha}_k}}\in\mathbb{R}\text{.}
\end{equation}
Furthermore, $\bm{\psi}(\bm{x})$ is said to be \emph{positive definite} if equality holds if and only if $\bm{\alpha}_1=\bm{\alpha}_2=\cdots=\bm{\alpha}_N=\bm{0}$.

A common approach to constructing matrix-valued positive definite functions is to apply second-order differential operators to scalar-valued positive definite functions \cite{Wendland2009}. To analyze the positive definiteness of such matrix-valued functions, we first establish Lemma~\ref{lem:drv2}.

\begin{lemma}
  \label{lem:drv2}
  Let $\bm{x}\in\mathbb{R}^d$ have components $x^1,x^2,\ldots,x^d$.
  Suppose $\psi\in\mathcal{C}^2(\mathbb{R}^d)\cap\mathcal{L}^1(\mathbb{R}^d)$ is a complex-valued function with Fourier transform $\widehat{\psi}(\bm{\omega})$. Then, the Fourier transform of its second-order partial derivatives satisfy
  \begin{equation}
    \mathcal{F}\left[\frac{\partial^2\psi}{\partial x^p\partial x^q}\right](\bm{\omega})=-\omega^m\omega^n\widehat{\psi}(\bm{\omega})\text{,\qquad}m,n=1,2,\ldots,d\text{.}
  \end{equation}
\end{lemma}
\begin{proof}
The computation of the Fourier transform proceeds as follows:
\begin{align}
  \mathcal{F}\left[\frac{\partial^2\psi}{\partial x^p\partial x^q}\right](\bm{\omega})
  &=\int_{\mathbb{R}^d}{\frac{\partial}{\partial x^p}\left[\frac{\partial\psi}{\partial x^q}(\bm{x})\right]\mathrm{e}^{-i\bm{x}\cdot\bm{\omega}}\mathrm{d}\bm{x}}\notag\\
  &=\int_{\mathbb{R}^d}{\frac{\partial}{\partial x^p}\left[\mathrm{e}^{-i\bm{x}\cdot\bm{\omega}}\frac{\partial\psi}{\partial x^q}(\bm{x})\right]\mathrm{d}\bm{x}}-\int_{\mathbb{R}^d}{\frac{\partial\mathrm{e}^{-i\bm{x}\cdot\bm{\omega}}}{\partial x^p}\frac{\partial\psi(\bm{x})}{\partial x^q}\,\mathrm{d}\bm{x}}\notag\\
  &=-\int_{\mathbb{R}^d}{\frac{\partial\mathrm{e}^{-i\bm{x}\cdot\bm{\omega}}}{\partial x^p}\frac{\partial\psi(\bm{x})}{\partial x^q}\,\mathrm{d}\bm{x}}\notag\\
  &=-\int_{\mathbb{R}^d}{\frac{\partial}{\partial x^q}\left[\psi(\bm{x})\frac{\partial\mathrm{e}^{-i\bm{x}\cdot\bm{\omega}}}{\partial x^p}\right]\mathrm{d}\bm{x}}+\int_{\mathbb{R}^d}{\psi(\bm{x})\,\frac{\partial^2\mathrm{e}^{-i\bm{x}\cdot\bm{\omega}}}{\partial x^q\partial x^p}\,\mathrm{d}\bm{x}}\notag\\
  &=\int_{\mathbb{R}^d}{\psi(\bm{x})\,\frac{\partial^2\mathrm{e}^{-i\bm{x}\cdot\bm{\omega}}}{\partial x^q\partial x^p}\,\mathrm{d}\bm{x}}\notag\\
  &=-\omega^p\omega^q\int_{\mathbb{R}^d}{\psi(\bm{x})\,\mathrm{e}^{-i\bm{x}\cdot\bm{\omega}}\mathrm{d}\bm{x}}\notag\\
  &=-\omega^p\omega^q\widehat{\psi}(\bm{\omega})\text{,}
\end{align}
where the second and fourth equalities follow from integration by parts, and the third and fifth equalities follow from the divergence theorem.
\end{proof}

Let $\del=(\partial/\partial x^1,\partial/\partial x^2,\ldots,\partial/\partial x^d)$ denote the gradient operator in column vector form. We can now establish the following theorem regarding the construction of matrix-valued positive definite functions.
\begin{theorem}
  \label{thm:matrix_kernel}
  Suppose $\psi\in\mathcal{C}^2(\mathbb{R}^d)\cap\mathcal{L}^1(\mathbb{R}^d)$ is a scalar-valued positive definite function. Then, the matrix-valued functions
  \begin{align}
    \bm{\psi}_1(\bm{x})&=-\Delta\psi(\bm{x})\,\bm{I}\text{,}\\
    \bm{\psi}_2(\bm{x})&=-\del\del^{\top}\psi(\bm{x})\text{,}\\
    \bm{\psi}_3(\bm{x})&=-\Delta\psi(\bm{x})\,\bm{I}+\del\del^{\top}\psi(\bm{x})\text{,}
  \end{align}
  are positive definite as long as they are integrable over $\mathbb{R}^d$.
\end{theorem}
\begin{proof}
The case for $\bm{\psi}_1$ follows directly from the fact that the Laplace operator preserves the positive definiteness of the function \cite[Lemma~9.15]{Wendland2004} and the definition of matrix-valued positive definite functions. For $\bm{\psi}_2$ and $\bm{\psi}_3$, let $S={\bm{x}_1,\bm{x}_2,\ldots,\bm{x}_N} \subset \mathbb{R}^d$ be an arbitrary set of distinct points, and let $\bm{\alpha}_1, \bm{\alpha}_2, \ldots, \bm{\alpha}_N \in \mathbb{C}^d$ be arbitrary vector coefficients.

We first consider the following summation for $\bm{\psi}_2$ using Lemma~\ref{lem:drv2}:
\begin{align}
    \sum_{j,k=1}^N{\overline{\bm{\alpha}_j}^{\top}\bm{\psi}_2(\bm{x}_j-\bm{x}_k)\,\bm{\alpha}_k}
    &=-\sum_{j,k=1}^N{\sum_{p,q=1}^d{\overline{\alpha_j^p}\alpha_k^q\,\frac{\partial^2\psi}{\partial p\partial q}(\bm{x}_j-\bm{x}_k)}}\notag\\
    &=-\frac{1}{(\sqrt{2\pi})^d}\sum_{j,k=1}^N{\sum_{p,q=1}^d{\overline{\alpha_j^p}\alpha_k^q\int_{\mathbb{R}^d}{\mathcal{F}\left[\frac{\partial^2\psi}{\partial x^p\partial x^q}\right](\bm{\omega})\,\mathrm{e}^{i(\bm{x}_j-\bm{x}_k)\cdot\bm{\omega}}\mathrm{d}\bm{\omega}}}}\notag\\
    &=\frac{1}{(\sqrt{2\pi})^d}\sum_{j,k=1}^N{\sum_{p,q=1}^d{\overline{\alpha_j^p}\alpha_k^q\int_{\mathbb{R}^d}{\omega^p\omega^q\widehat{\psi}(\bm{\omega})\,\mathrm{e}^{i(\bm{x}_j-\bm{x}_k)\cdot\bm{\omega}}\mathrm{d}\bm{\omega}}}}\notag\\
    &=\frac{1}{(\sqrt{2\pi})^d}\int_{\mathbb{R}^d}{\widehat{\psi}(\bm{\omega})\left(\sum_{j=1}^N{\sum_{p=1}^d{\overline{\alpha_j^p}\omega^p\mathrm{e}^{i\bm{x}_j\cdot\bm{\omega}}}}\sum_{k=1}^N{\sum_{q=1}^d{\alpha_k^q\omega^q\mathrm{e}^{-i\bm{x}_k\cdot\bm{\omega}}}}\right)\mathrm{d}\bm{\omega}}\notag\\
    &=\frac{1}{(\sqrt{2\pi})^d}\int_{\mathbb{R}^d}{\widehat{\psi}(\bm{\omega})\left\vert\sum_{j=1}^N{\sum_{p=1}^d{\overline{\alpha_j^p}\omega^p\mathrm{e}^{i\bm{x}_j\cdot\bm{\omega}}}}\right\vert^2\mathrm{d}\bm{\omega}}\label{eqn:fourier_psi2}\text{.}
\end{align}
As given, $\widehat{\psi}(\bm{\omega})$ is a non-negative function that is not identically zero, so Eq.~\eqref{eqn:fourier_psi2} is always greater than 0, thus $\bm{\psi}_2$ is positive definite.
Observing the proof, we can express the Fourier transform of the Hessian matrix of the second-order partial derivatives as the matrix of Fourier transforms of its components: 
\begin{equation}
  \mathcal{F}\left[\del\del^{\top}\psi\right](\bm{\omega})=-\bm{\omega}\bm{\omega}^{\top}\widehat{\psi}(\omega)\text{,}
\end{equation}
and thus Eq.~\eqref{eqn:fourier_psi2} can be written as
  \begin{equation}
    \sum_{j,k=1}^N{\overline{\bm{\alpha}_j}^{\top}\bm{\psi}_2(\bm{x}_j-\bm{x}_k)\,\bm{\alpha}_k}
    =\frac{1}{(\sqrt{2\pi})^d}\int_{\mathbb{R}^d}{\widehat{\psi}(\bm{\omega})\left\vert\sum_{j=1}^N{\overline{\bm{\alpha}_j}\cdot\bm{\omega}\,\mathrm{e}^{i\bm{x}_j\cdot\bm{\omega}}}\right\vert^2\mathrm{d}\bm{\omega}}\text{.}
  \end{equation}

Similarly, we apply the same analysis to $\bm{\psi}_3$ using Fourier analysis to prove:
\begin{align}
    \sum_{j,k=1}^N{\overline{\bm{\alpha}_j}^{\top}\bm{\psi}_2(\bm{x}_j-\bm{x}_k)\,\bm{\alpha}_k}
    &=\frac{1}{(\sqrt{2\pi})^d}\sum_{j,k=1}^N{\int_{\mathbb{R}^d}{\overline{\bm{\alpha}_j}^{\top}\left(\Vert\bm{\omega}\Vert^2\bm{I}-\bm{\omega}\bm{\omega}^{\top}\right)\bm{\alpha}_k\,\widehat{\psi}(\bm{\omega})\,\mathrm{e}^{i(\bm{x}_j-\bm{x}_k)\cdot\bm{\omega}}\mathrm{d}\bm{\omega}}}\notag\\
    &=\frac{1}{(\sqrt{2\pi})^d}\sum_{j,k=1}^N{\int_{\mathbb{R}^d}{\overline{\bm{\alpha}_j}^{\top}\bm{L}^{\top}\bm{L}\bm{\alpha}_k\,\widehat{\psi}(\bm{\omega})\,\mathrm{e}^{i(\bm{x}_j-\bm{x}_k)\cdot\bm{\omega}}\mathrm{d}\bm{\omega}}}\notag\\
    &=\frac{1}{(\sqrt{2\pi})^d}\int_{\mathbb{R}^d}{\widehat{\psi}(\bm{\omega})\left\Vert\sum_{j=1}^N{\bm{L}\overline{\bm{\alpha}_j}\mathrm{e}^{i\bm{x}_j\cdot\bm{\omega}}}\right\Vert^2\mathrm{d}\bm{\omega}}\label{eqn:fourier_psi3}\text{,}
\end{align}
where $\bm{L}^{\top}\bm{L}=\Vert\bm{\omega}\Vert^2\bm{I}-\bm{\omega}\bm{\omega}^{\top}$ is the result of the Cholesky decomposition.
Note that for any $\bm{\alpha}\in\mathbb{R}^d$, we can obtain
\begin{equation}
    \bm{\alpha}^{\top}\left(\Vert\bm{\omega}\Vert^2\bm{I}-\bm{\omega}\bm{\omega}^{\top}\right)\bm{\alpha}=\Vert\bm{\omega}\Vert^2\Vert\bm{\alpha}\Vert^2-(\bm{\omega}\cdot\bm{\alpha})^2\ge0\text{,}
\end{equation}
so $\Vert\bm{\omega}\Vert^2\bm{I}-\bm{\omega}\bm{\omega}^{\top}$ is semi-positive definite, ensuring the existence of $\bm{L}$.
\end{proof}

\subsection{Helmholtz Decomposition}

The matrix-valued positive definite functions constructed according to Theorem~\ref{thm:matrix_kernel} possess a series of important properties.

\begin{theorem}
  \label{thm:curl-free}
   For any scalar function $\psi \in \mathcal{C}^2(\mathbb{R}^d)$ and any vector coefficient $\bm{\alpha} \in \mathbb{C}^d$, $\bm{\psi}_2 \bm{\alpha} = -\del\del^{\top} \psi(\bm{x}) \bm{\alpha}$ is curl-free, i.e., $\del\times(\bm{\psi}_2 \bm{\alpha}) = \bm{0}$.
\end{theorem}
\begin{proof}
By the linearity of the $\del$ operator, we can express $\bm{\psi}_2 \bm{\alpha}$ as
\begin{equation}
    -\del\del^{\top}\psi(\bm{x})\,\bm{\alpha}=\bm{\del}\left[-\del\psi(\bm{x})\cdot\bm{\alpha}\right]\text{,}
\end{equation}
which shows that $\bm{\psi}_2 \bm{\alpha}$ is the gradient of a scalar field, and its curl must be zero.
\end{proof}

\begin{theorem}
\label{thm:div-free}
For any scalar function $\psi \in \mathcal{C}^2(\mathbb{R}^d)$ and any vector coefficient $\bm{\alpha} \in \mathbb{C}^d$,
$\bm{\psi}_3\,\bm{\alpha}=-\Delta\psi(\bm{x})\,\bm{\alpha}+\del\del^{\top}\psi(\bm{x})\,\bm{\alpha}$ is divergence-free, i.e., $\del\cdot(\bm{\psi}_3\,\bm{\alpha})=0$.
\end{theorem}
\begin{proof}
By the linearity of the $\del$ operator, we can rewrite the expression for $\bm{\psi}_3 \bm{\alpha}$ as follows:
  \begin{align}
    -\Delta\psi(\bm{x})\,\bm{\alpha}+\del\del^{\top}\psi(\bm{x})\,\bm{\alpha}&=-\del\cdot\left[\del\psi(\bm{x})\right]\bm{\alpha}+\bm{\del}\left[\del\psi(\bm{x})\cdot\bm{\alpha}\right]\notag\\
    &=\del\times(\del\psi\times\bm{\alpha})\text{,}
  \end{align}
which shows that $\bm{\psi}_3 \bm{\alpha}$ is the curl of a vector field, and its divergence must be zero.
\end{proof}

Clearly, from Theorem~\ref{thm:curl-free}, Theorem~\ref{thm:div-free}, and the Helmholtz decomposition theorem for vector fields, we can deduce the following corollary.

\begin{corollary}
  Let $\psi\in\mathcal{C}^2(\mathbb{R}^d)$ be a scalar function, $S=\{\bm{x}_1,\bm{x}_2,\ldots,\bm{x}_N\}\subset\mathbb{R}^d$ be an arbitrary set of distinct points, and $\bm{\alpha}_1,\bm{\alpha}_2,\ldots,\bm{\alpha}_N\in\mathbb{C}^d$ be arbitrary vector coefficients. The vector field 
  \begin{equation}
    \widetilde{\bm{f}}_1(\bm{x})=-\sum_{k=1}^N{\Delta\psi(\bm{x}-\bm{x}_k)\,\bm{\alpha}_k}
  \end{equation}
  can be uniquely decomposed into the curl-free field
  \begin{equation}
    \widetilde{\bm{f}}_2(\bm{x})=-\sum_{k=1}^N{\del^{\top}\del\psi(\bm{x}-\bm{x}_k)\,\bm{\alpha}_k}
  \end{equation}
  and the divergence-free field 
  \begin{equation}
    \widetilde{\bm{f}}_3(\bm{x})=-\sum_{k=1}^N{\Delta\psi(\bm{x}-\bm{x}_k)\,\bm{\alpha}_k}+\sum_{k=1}^N{\del^{\top}\del\psi(\bm{x}-\bm{x}_k)\,\bm{\alpha}_k}
  \end{equation}
  as their sum.
\end{corollary}

\section{Analytic Expressions of Functions and Derivatives}

For practical applications, we provide the analytical forms of matrix-valued radial basis functions and their derivatives, including both the general form and the specialized version for DFKs-Wen4.

\subsection{Original Functions}

For any radial function $\psi(\bm{x})=\phi(\Vert\bm{x}\Vert)$, the Laplacian $\Delta\psi$ can be expressed as
\begin{equation}
  \label{eqn:laplacian}
  \Delta\psi(\bm{x})=\Delta\phi(r)=\frac{\mathrm{d}^2\phi}{{\mathrm{d}r}^2}+\frac{d-1}{r}\frac{\mathrm{d}\phi}{\mathrm{d}r}\text{.}
\end{equation}
Similarly, the Hessian of an Radial function can be expanded as
\begin{equation}
  \del^{\top}\del\psi(\bm{x})=\del^{\top}\del\phi(r)=\del^{\top}\left[\frac{\mathrm{d}\phi}{\mathrm{d}r}(r)\,\frac{\bm{x}}{r}\right]=\frac{1}{r}\frac{\mathrm{d}\phi}{\mathrm{d}r}\bm{I}+\frac{1}{r^2}\left(\frac{\mathrm{d}^2\phi}{{\mathrm{d}r}^2}-\frac{1}{r}\frac{\mathrm{d}\phi}{\mathrm{d}r}\right)\bm{x}\bm{x}^{\top}\text{.}
\end{equation}

For the $\mathcal{C}^4$-continous Wendland polynomial used in the DFK-Wen4 method, the result of applying the negative Laplacian is given by
\begin{equation}
  -\Delta R_\mathrm{Wen4}(r)=(1-r)^4\left[d+4dr-(5d+30)\,r^2\right]\text{.}    
\end{equation}
The corresponding curl-free RBF takes the form
\begin{equation}
  -\del^{\top}\del R_\mathrm{Wen4}(r)=(1-r)^4\left[(1+4r-5r^2)\,\bm{I}-30\,\bm{x}\bm{x}^{\top}\right]\text{,}
\end{equation}
and the corresponding divergence-free RBF is given by
\begin{align}
  -\Delta R_\mathrm{Wen4}(r)\,\bm{I}+\del^{\top}\del R_\mathrm{Wen4}(r)=(1-r)^4\left\{\left[(d-1)(1+4r)-5(d+5)\,r^2\right]\bm{I}+30\,\bm{x}\bm{x}^{\top}\right\}\text{.}
\end{align}

\subsection{Gradients}

Consider a matrix-valued RBF interpolation, where each kernel contributes in the form
\begin{equation}
  \bm{u}(\bm{x})=f(r)\,\bm{\alpha}+g(r)\,(\bm{x}\cdot\bm{\alpha})\,\bm{x}\text{,}
\end{equation}
which, in index notation, can be written as
\begin{equation}
  u_i=f\alpha_i+gx_k\alpha_kx_i\text{.}
\end{equation}
Taking the derivative with respect to $x_j$ yields
\begin{equation}
  \frac{\partial u_i}{\partial x_j}=\frac{\partial f}{\partial x_j}\alpha_i+\frac{\partial g}{\partial x_j}x_k\alpha_kx_i+g\alpha_jx_i+gx_k\alpha_k\delta_{ij}\text{,}
\end{equation}
which, in matrix form, can be expressed as
\begin{align}
  \del\bm{u}&=\bm{\alpha}\,(\del f)^\top+(\bm{x}\cdot\bm{\alpha})\,\bm{x}\,(\del g)^\top+g\bm{x}\bm{\alpha}^\top+g\,(\bm{x}\cdot\bm{\alpha})\,\bm{I}\notag\\
  &=\frac{f'(r)}{r}\bm{\alpha}\bm{x}^\top+\frac{g'(r)}{r}(\bm{x}\cdot\bm{\alpha})\,\bm{x}\bm{x}^\top+g\bm{x}\bm{\alpha}^\top+g\,(\bm{x}\cdot\bm{\alpha})\,\bm{I}\text{.}
\end{align}
For the divergence-free matrix-valued RBF derived from $R_\mathrm{Wen4}$, the coefficients in two dimensions are given by
\begin{equation}
  f'(r)=30r(1-r)^3(7r-3)\text{,\qquad}g'(r)=-120(1-r)^3\text{,\qquad}g(r)=30(1-r)^4\text{.}
\end{equation}
In three dimensions, they take the form
\begin{equation}
  f'(r)=120r(1-r)^3(2r-1)\text{,\qquad}g'(r)=-120(1-r)^3\text{,\qquad}g(r)=30(1-r)^4\text{.}
\end{equation}

\paragraph{Backpropagation}
Consider a scalar loss function $\mathcal{L}$ and define $t_{ij}=\partial u_i/\partial x_j$.
Then, the gradient of the loss with respect to $\alpha_k$ is given by
\begin{align}
  \frac{\partial\mathcal{L}}{\partial\alpha_k}=\frac{\partial\mathcal{L}}{\partial t_{ij}}\frac{\partial t_{ij}}{\partial\alpha_k}
  &=\frac{\partial\mathcal{L}}{\partial t_{ij}}\left[\frac{\partial f}{\partial x_j}\delta_{ik}+\frac{\partial g}{\partial x_j}x_kx_i+gx_i\delta_{jk}+gx_k\delta_{ij}\right]\notag\\
  &=\frac{\partial\mathcal{L}}{\partial t_{kj}}\frac{\partial f}{\partial x_j}+\frac{\partial\mathcal{L}}{\partial t_{ij}}\frac{\partial g}{\partial x_j}x_kx_i+\frac{\partial\mathcal{L}}{\partial t_{ik}}gx_i+\frac{\partial\mathcal{L}}{\partial t_{ii}}gx_k\text{.}
\end{align}
In matrix form, this can be expressed as
\begin{align}
  \frac{\partial\mathcal{L}}{\partial\bm{\alpha}}&=\bm{S}(\del f)+\bm{x}^\top\bm{S}(\del g)\bm{x}+g\bm{S}^\top\bm{x}+g(\mathop{\mathrm{tr}}{\bm{S}})\bm{x}\notag\\
  &=\frac{f'(r)}{r}\bm{S}\bm{x}+\frac{g'(r)}{r}(\bm{x}^\top\bm{S}\bm{x})\bm{x}+g\bm{S}^\top\bm{x}+g(\mathop{\mathrm{tr}}{\bm{S}})\bm{x}\text{,}
\end{align}
where $\bm{S}=\{s_{ij}\}=\{\partial\mathcal{L}/\partial t_{ij}\}$.

\subsection{Curls}

Physically, the curl of a vector field corresponds to its vorticity. Consider an RBF interpolation where the contribution of each kernel takes the form
\begin{equation}
  \bm{u}(\bm{x})=h(r)\,\bm{\alpha}\text{,}
\end{equation}
which, in index notation, can be expressed as
\begin{equation}
  u_i=h\,\alpha_i\text{.}
\end{equation}
Taking the curl of $\bm{u}$, we obtain
\begin{equation}
  \omega_i=(\del\times\bm{u})_i=\frac{\partial h}{\partial x_j}\alpha_k\epsilon_{ijk}=[(\del h)\times\bm{\alpha}]_i\text{,}
\end{equation}
which simplifies to
\begin{equation}
  \bm{\omega}=\del\times\bm{u}=\frac{h'(r)}{r}\bm{x}\times\bm{\alpha}\text{.}
\end{equation}
For the divergence-free matrix-valued RBF derived from $R_\mathrm{Wen4}$ (where the irrotational component automatically vanishes), the function $h'(r)$ is given by
\begin{equation}
  h'(r)=120r(1-r)^3(2r-1)\text{,}
\end{equation}
in 2D, which in 3D takes the form
\begin{equation}
  h'(r)=30r(1-r)^3(9r-5)\text{.}
\end{equation}

\paragraph{Backpropagation}
Consider a scalar loss function $\mathcal{L}$. The gradient of the loss with respect to $\alpha_k$ is given by
\begin{equation}
  \frac{\partial\mathcal{L}}{\partial\alpha_k}=\frac{\partial\mathcal{L}}{\partial\omega_i}\frac{\partial\omega_i}{\partial\alpha_k}=\frac{\partial\mathcal{L}}{\partial\omega_i}\frac{\partial h}{\partial x_j}\epsilon_{ijk}\text{.}
\end{equation}
In matrix form, this can be expressed as
\begin{equation}
  \frac{\partial\mathcal{L}}{\partial\bm{\alpha}}=\frac{h'(r)}{r}\frac{\partial\mathcal{L}}{\partial\bm{\omega}}\times\bm{x}\text{.}
\end{equation}

\subsection{Gradients of Curls}

Continuing from the previous subsection, we differentiate the vorticity with respect to $x_j$. In index notation, this yields
\begin{equation}
  \frac{\partial \omega_i}{\partial x_j}=\frac{\partial^2h}{\partial j\partial m}\alpha_n\epsilon_{imn}=\left[\frac{h'(r)}{r}\delta_{jm}+\frac{1}{r^2}\left(h''(r)-\frac{h'(r)}{r}\right)x_jx_m\right]\alpha_n\epsilon_{imn}\text{,}
\end{equation}
which can be written in matrix form as
\begin{equation}
  \del\bm{\omega}=\frac{h'(r)}{r}\bm{A}+\frac{1}{r^2}\left(h''(r)-\frac{h'(r)}{r}\right)(\bm{x}\times\bm{\alpha})\bm{x}^\top\text{,}
\end{equation}
where
\begin{equation}
  \bm{A}=
  \begin{pmatrix}
    0 & \alpha_3 & -\alpha_2\\
    -\alpha_3 & 0 & \alpha_1\\
    \alpha_2 & -\alpha_1 & 0
  \end{pmatrix}
  \text{.}
\end{equation}
For the divergence-free matrix-valued RBF derived from $R_\mathrm{Wen4}$, in three dimensions, the coefficient of the second term is
\begin{equation}
  \frac{1}{r^2}\left(h''(r)-\frac{h'(r)}{r}\right)=360\frac{(1-r)^2(2-3r)}{r}\text{.}
\end{equation}

\paragraph{Backpropagation}
Consider the backpropagation process with a scalar loss function $\mathcal{L}$, and let $t_{ij}=\partial \omega_i/\partial x_j$ and $s_{ij}=\partial\mathcal{L}/\partial t_{ij}$. The derivative of the loss function with respect to $\alpha_k$ is given by
\begin{equation}
  \frac{\partial\mathcal{L}}{\partial\alpha_k}=s_{ij}\frac{\partial t_{ij}}{\partial\alpha_k}
  =\left[\frac{h'(r)}{r}s_{im}+\frac{1}{r^2}\left(h''(r)-\frac{h'(r)}{r}\right)s_{ij}x_jx_m\right]\epsilon_{imk}\text{.}
\end{equation}
In matrix form, this becomes
\begin{equation}
  \frac{\partial\mathcal{L}}{\partial\bm{\alpha}}=\frac{h'(r)}{r}
  \begin{pmatrix}
    s_{23}-s_{32}\\
    s_{31}-s_{13}\\
    s_{12}-s_{21}
  \end{pmatrix}
  +\frac{1}{r^2}\left(h''(r)-\frac{h'(r)}{r}\right)
  (\bm{S}\bm{x})\times\bm{x}\text{.}
\end{equation} 

\section{Additional Descriptions of Experiments}

\subsection{Fitting}

\begin{figure}[ht]
    \centering
    \includegraphics[width=0.45\textwidth]{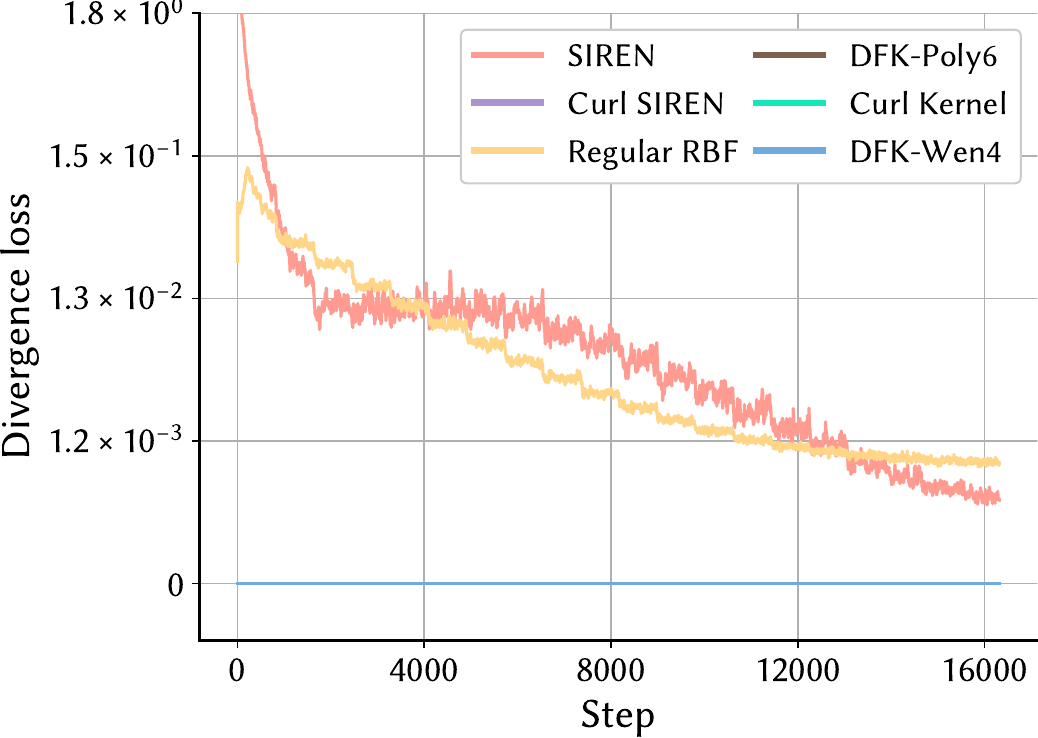}
    \qquad
    \includegraphics[width=0.45\textwidth]{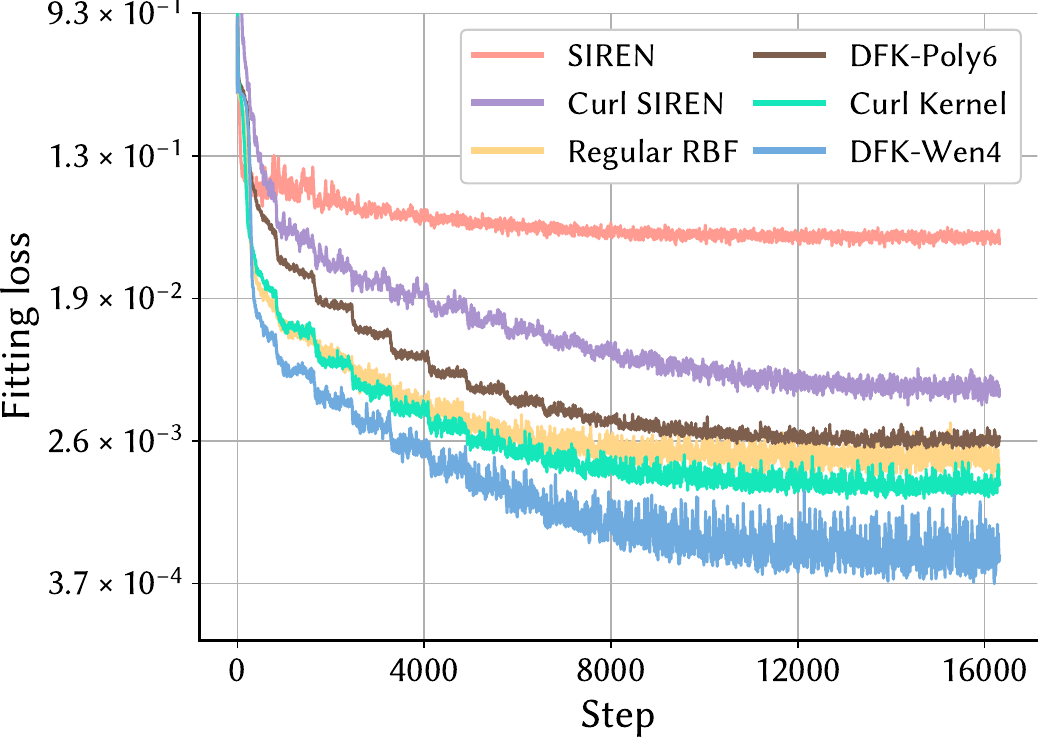}
    \\
    \vspace{-1em}
    \caption{Loss curves of fitting experiments for \emph{K\'{a}m\'{a}n vortex street (2D)}. Divergence losses of Curl SIREN, DFK-Poly6, Curl Kernel, and DFK-Wen4 are always zero due to the intrinsic properties of representations.}
    \label{fig:karman}
    \Description{loss curves}
\end{figure}

\paragraph{K\'{a}rm\'{a}n vortex street (2D)}
SIREN consists of four hidden layers with 256 neurons each, totaling \num{198658} trainable parameters. 
Curl SIREN consists of four hidden layers with 256 neurons each, totaling \num{198401} trainable parameters. 
Regular RBF comprises \num{5367} points, resulting in \num{26835} trainable parameters.
DFK-Poly6 comprises \num{5367} points, resulting in \num{26835} trainable parameters.
Curl Kernel comprises \num{6794} points, resulting in \num{27176} trainable parameters.
DFK-Wen4 comprises \num{5367} points, resulting in \num{26835} trainable parameters.

We set the batch size to \num{128} and trained each model for \num{20} epochs, with an initial learning rate of \num{1e-3}. When initializing the kernel radii, we set $\eta=9$. The loss curves are illustrated in Fig.~\ref{fig:karman}.

\begin{figure}[ht]
    \centering
    \includegraphics[width=0.45\textwidth]{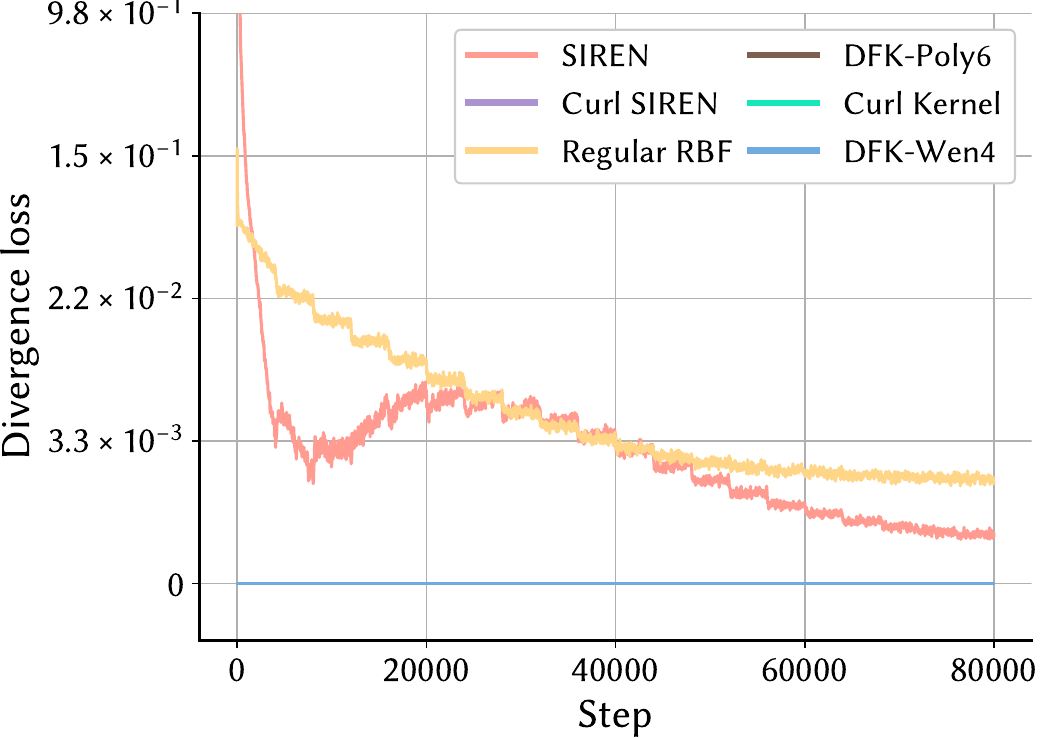}
    \qquad
    \includegraphics[width=0.45\textwidth]{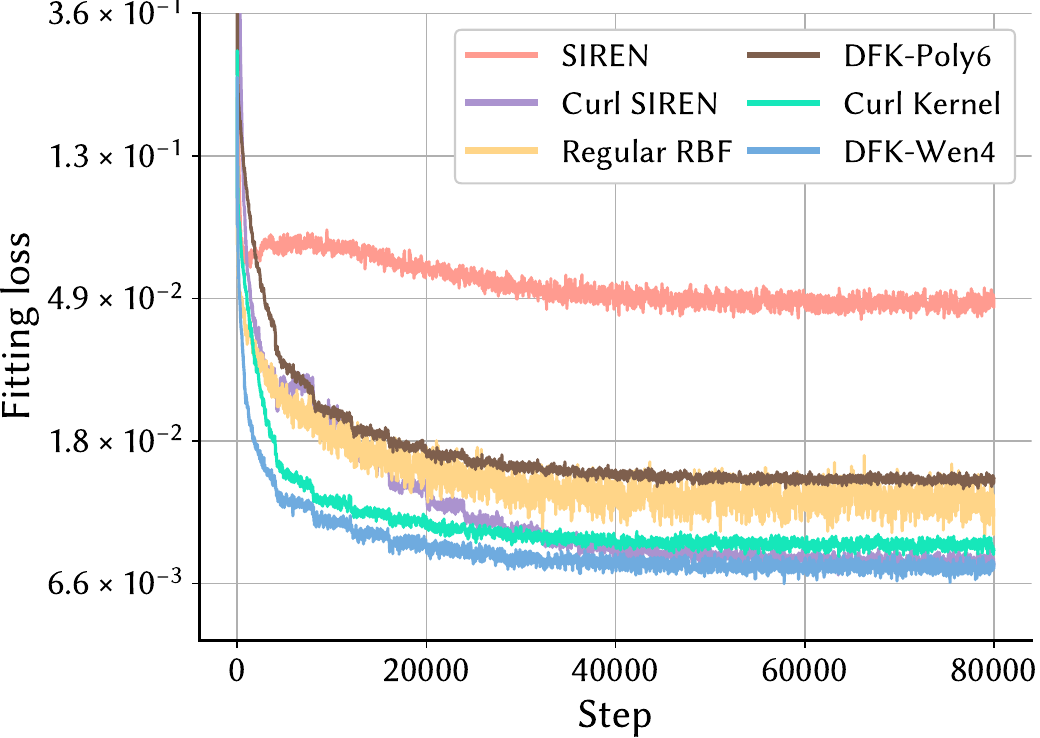}
    \\
    \vspace{-1em}
    \caption{Loss curves of fitting experiments for \emph{analytic vortices (3D)}. Divergence losses of Curl SIREN, DFK-Poly6, Curl Kernel, and DFK-Wen4 are always zero due to the intrinsic properties of representations.}
    \label{fig:analytic}
    \Description{loss curves}
\end{figure}

\paragraph{Analytic vortices (3D)}
SIREN consists of four hidden layers with 256 neurons each, totaling \num{199171} trainable parameters. 
Curl SIREN consists of four hidden layers with 256 neurons each, totaling \num{199171} trainable parameters. 
Regular RBF comprises \num{21117} points, resulting in \num{147819} trainable parameters.
DFK-Poly6 comprises \num{21117} points, resulting in \num{147819} trainable parameters.
Curl Kernel comprises \num{21117} points, resulting in \num{147819} trainable parameters.
DFK-Wen4 comprises \num{21117} points, resulting in \num{147819} trainable parameters.

We set the batch size to \num{128} and trained each model for \num{20} epochs, with an initial learning rate of \num{1e-3} for NNs and \num{5e-4} for kernels. When initializing the kernel radii, we set $\eta=6$. The loss curves are illustrated in Fig.~\ref{fig:analytic}.

\begin{figure}[ht]
    \centering
    \includegraphics[width=0.45\textwidth]{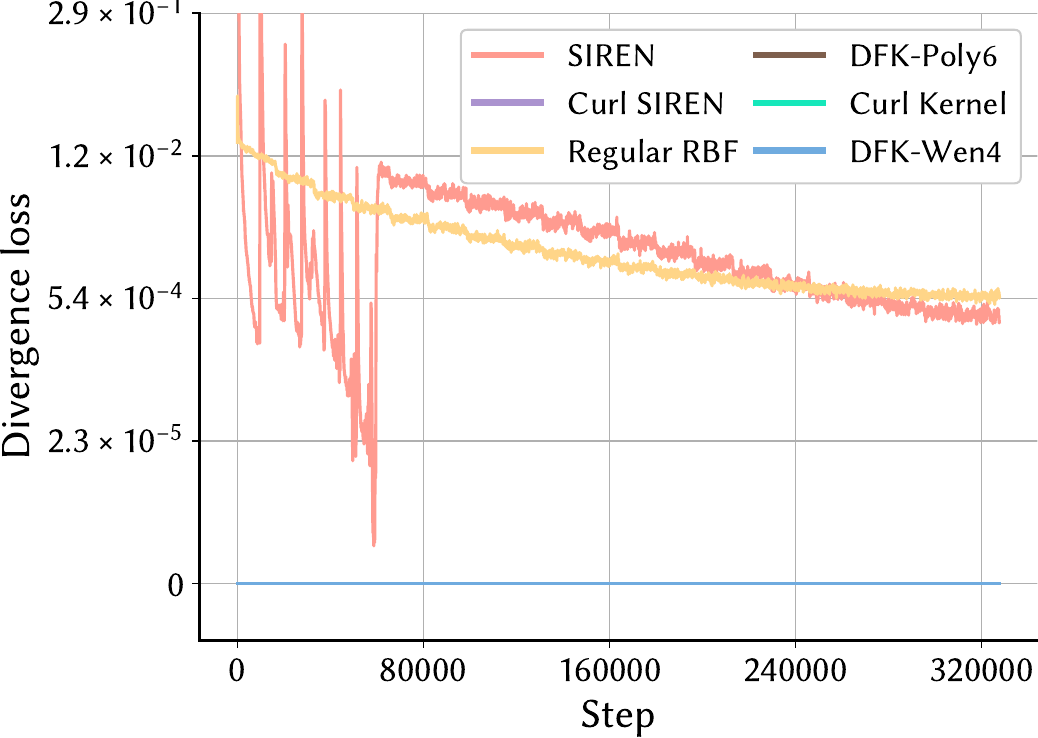}
    \qquad
    \includegraphics[width=0.45\textwidth]{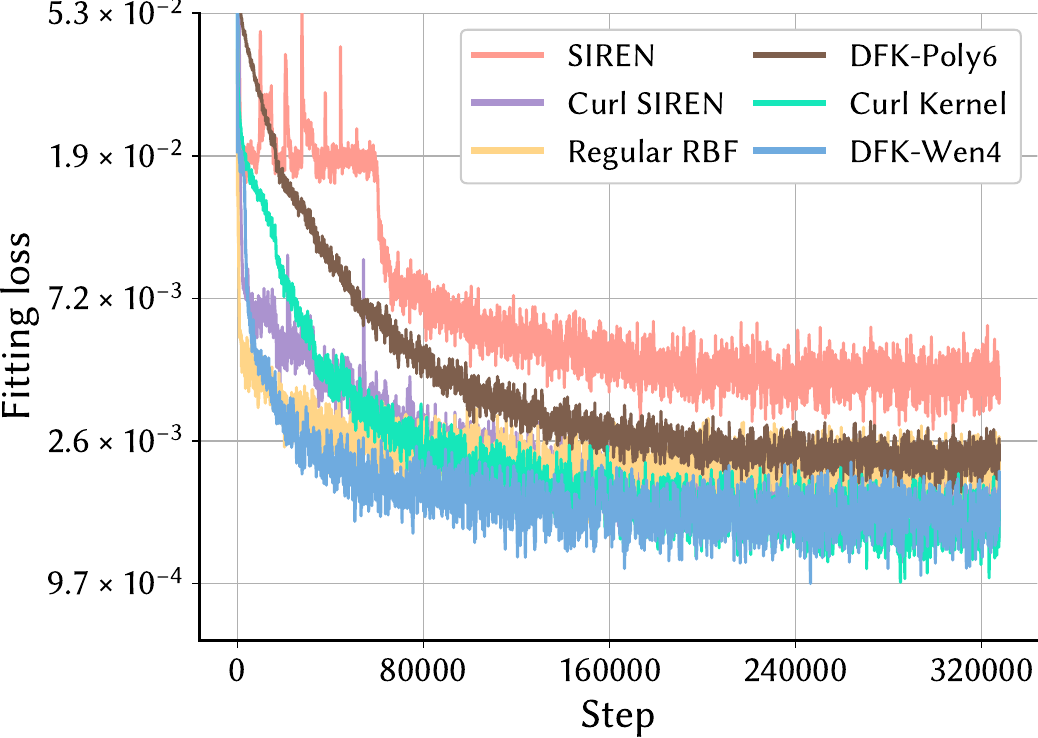}
    \\
    \vspace{-1em}
    \caption{Loss curves of fitting experiments for \emph{simple plume (3D)}. Divergence losses of Curl SIREN, DFK-Poly6, Curl Kernel, and DFK-Wen4 are always zero due to the intrinsic properties of representations.}
    \label{fig:plume}
    \Description{loss curves}
\end{figure}

\paragraph{Simple plume (3D)}
SIREN consists of six hidden layers with 256 neurons each, totaling \num{330755} trainable parameters. 
Curl SIREN consists of six hidden layers with 256 neurons each, totaling \num{330755} trainable parameters. 
Regular RBF comprises \num{42002} points, resulting in \num{294014} trainable parameters.
DFK-Poly6 comprises \num{42002} points, resulting in \num{294014} trainable parameters.
Curl Kernel comprises \num{42002} points, resulting in \num{294014} trainable parameters.
DFK-Wen4 comprises \num{42002} points, resulting in \num{294014} trainable parameters.

We set the batch size to \num{128} and trained each model for \num{20} epochs, with an initial learning rate of \num{1e-3} for NNs and \num{1e-4} for kernels. When initializing the kernel radii, we set $\eta=6$.  The loss curves are illustrated in Fig.~\ref{fig:plume}.

\subsection{Projection}

For the projection task, Curl SIREN is combined with Gradient SIREN, which shares the same architecture, to fit the data, while the two models remain independent.
Similarly, Curl Kernel and Gradient Kernel share the same point positions and radii but have independent weights.

\begin{figure}[ht]
    \centering
    \captionbox{Loss curves of projection experiments for \emph{Taylor vortex (2D)}.\label{fig:taylor}}
    {\includegraphics[width=0.45\textwidth]{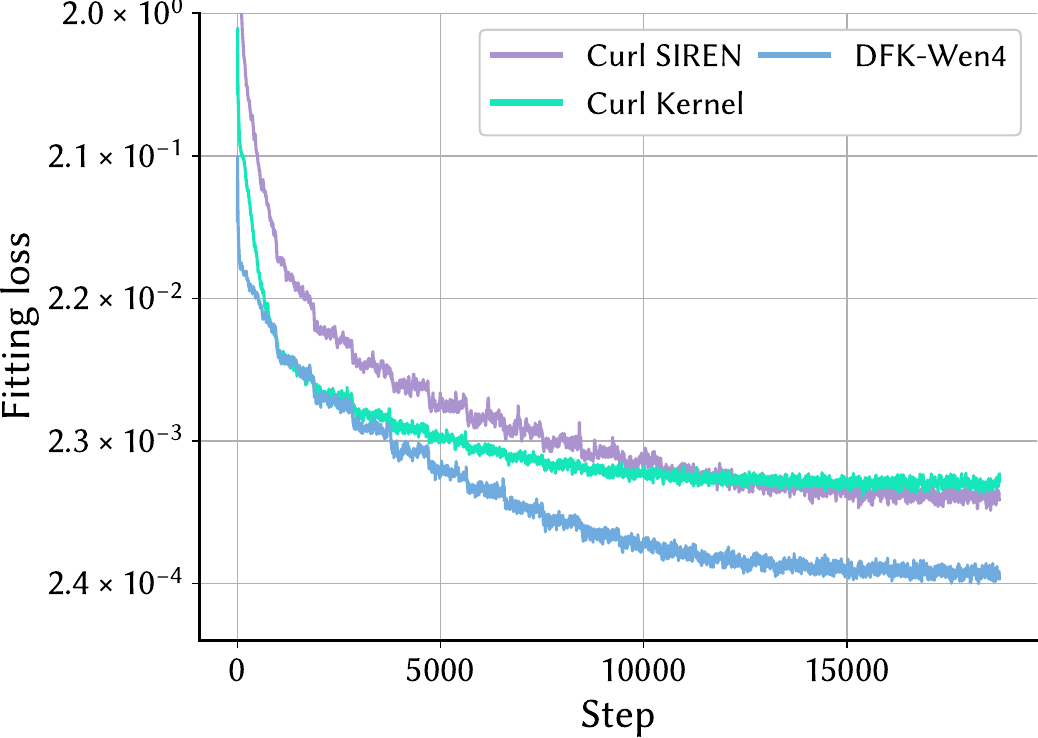}\vspace{-1em}}
    \qquad
    \captionbox{Loss curves of projection experiments for \emph{vortex ring collision (3D)}.\label{fig:rings}}
    {\includegraphics[width=0.45\textwidth]{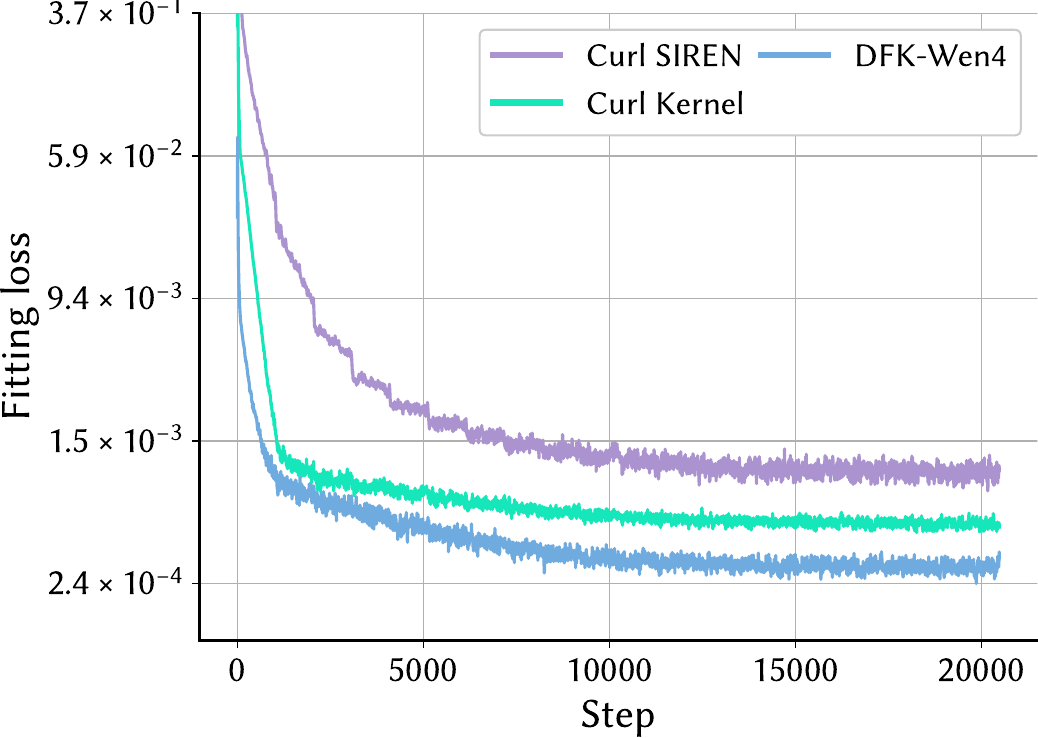}\vspace{-1em}}
    \Description{loss curves}
\end{figure}

\paragraph{Taylor vortex (2D)}
Curl SIREN (as well as Gradient SIREN) consists of four hidden layers with 256 neurons each, totaling \num{396802} trainable parameters. 
Curl Kernel (as well as Gradient Kernel) comprises \num{5416} points, resulting in \num{27080} trainable parameters.
DFK-Wen4 comprises \num{5416} points, resulting in \num{27080} trainable parameters.

We set the batch size to \num{128} and trained each model for \num{20} epochs, with an initial learning rate of \num{1e-3}. When initializing the kernel radii, we set $\eta=27$. The loss curves are illustrated in Fig.~\ref{fig:taylor}.

\paragraph{Vortex ring collision (3D)}
Curl SIREN (as well as Gradient SIREN) consists of four hidden layers with 256 neurons each, totaling \num{397828} trainable parameters. 
Curl Kernel (as well as Gradient Kernel) comprises \num{8544} points, resulting in \num{68352} trainable parameters.
DFK-Wen4 comprises \num{8544} points, resulting in \num{59808} trainable parameters.

We set the batch size to \num{2048} and trained each model for \num{20} epochs, with an initial learning rate of \num{1e-3}. When initializing the kernel radii, we set $\eta=12$. The loss curves are illustrated in Fig.~\ref{fig:rings}.

\subsection{Inpainting}

\begin{figure}[t]
    \centering
    \subcaptionbox{\SI{0}{\degree};}{\includegraphics[width=0.325\textwidth]{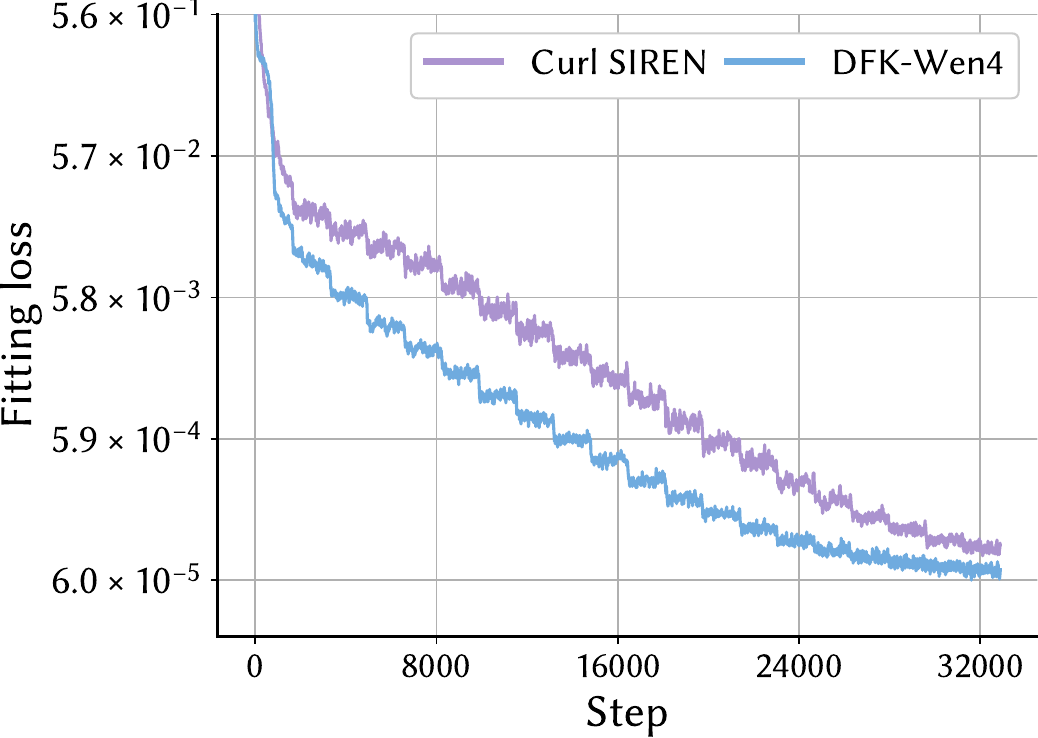}}
    \hfill
    \subcaptionbox{\SI{45}{\degree};}{\includegraphics[width=0.325\textwidth]{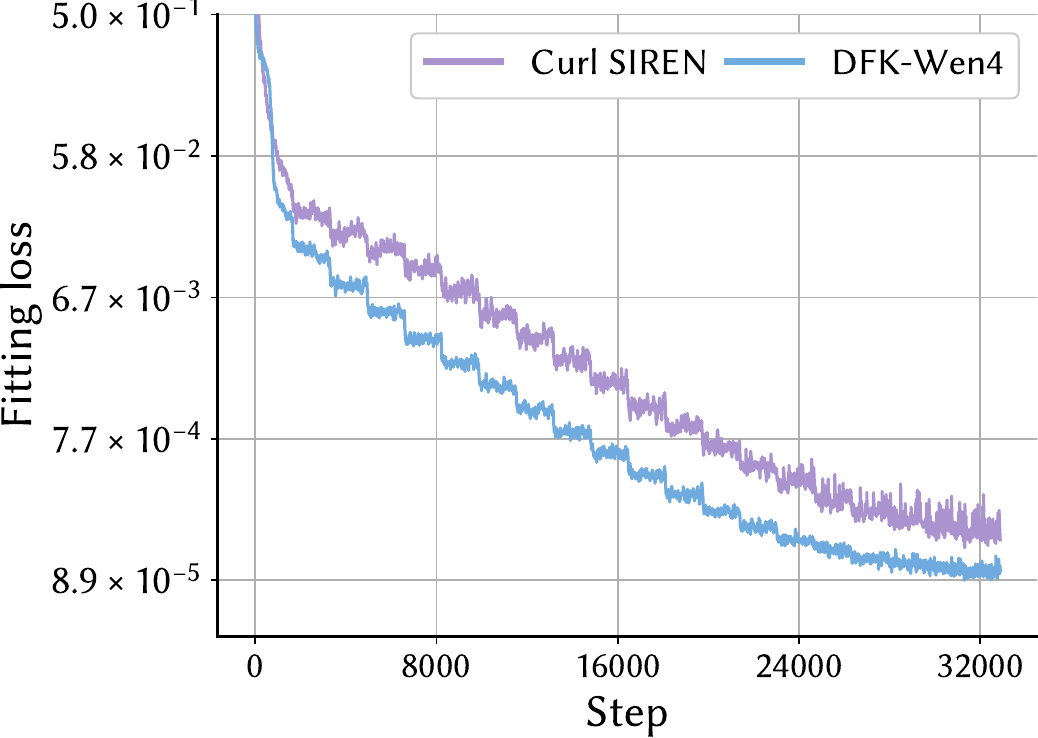}}
    \hfill
    \subcaptionbox{\SI{90}{\degree};}{\includegraphics[width=0.325\textwidth]{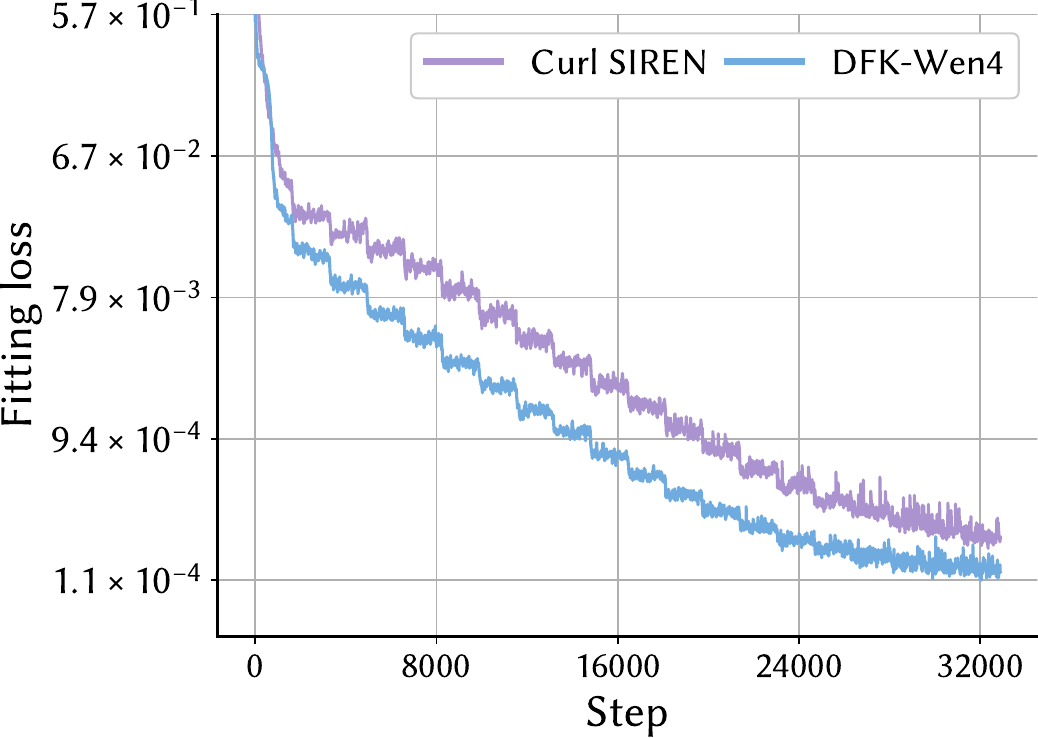}}
    \\
    \vspace{-1em}
    \caption{Loss curves of inpainting experiments for \emph{laminar flows (2D)}.}
    \label{fig:laminar}
    \Description{loss curves}
\end{figure}

\paragraph{Laminar flows (2D)}
Curl SIREN consists of four hidden layers with 256 neurons each, totaling \num{198401} trainable parameters. 
DFK-Wen4 comprises \num{5416} points, resulting in \num{27080} trainable parameters.

We set the batch size to \num{128} and trained each model for \num{20} epochs, with an initial learning rate of \num{1e-3}. When initializing the kernel radii, we set $\eta=9$.  The loss curves are illustrated in Fig.~\ref{fig:laminar}.

\begin{figure}[ht]
    \centering
    \captionbox{Loss curves of inpainting experiments for \emph{missile (2D)}.\label{fig:missile}}
    {\includegraphics[width=0.45\textwidth]{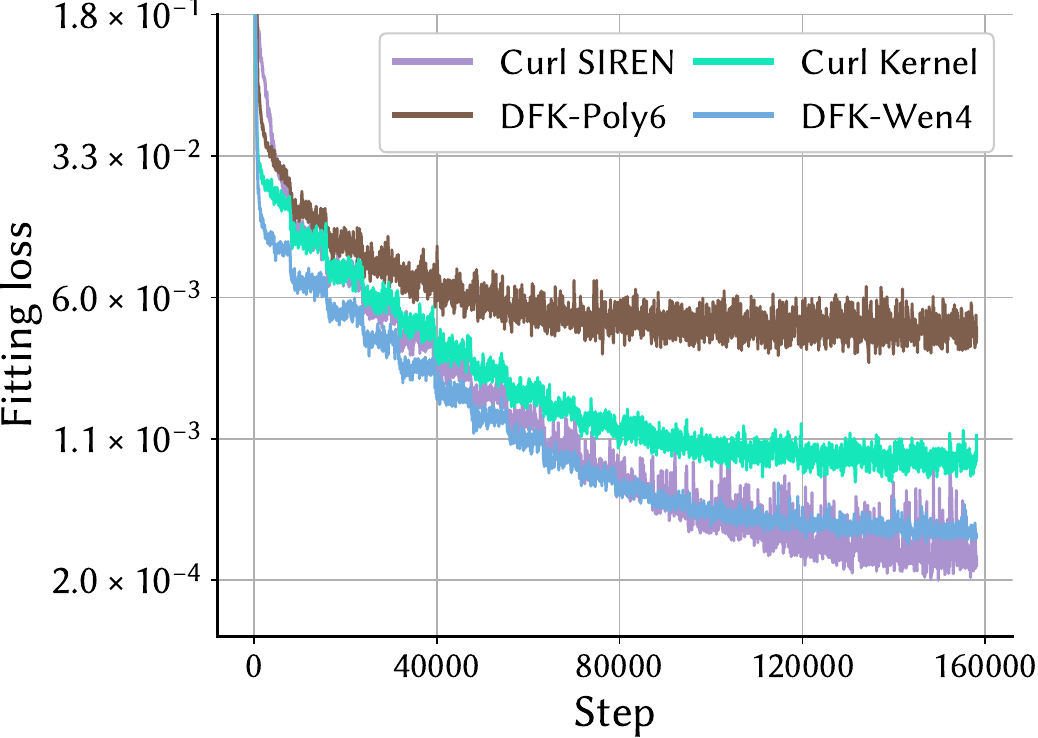}\vspace{-1em}}
    \qquad
    \captionbox{Loss curves of inpainting experiments for \emph{bullet (3D)}.\label{fig:bullet}}
    {\includegraphics[width=0.45\textwidth]{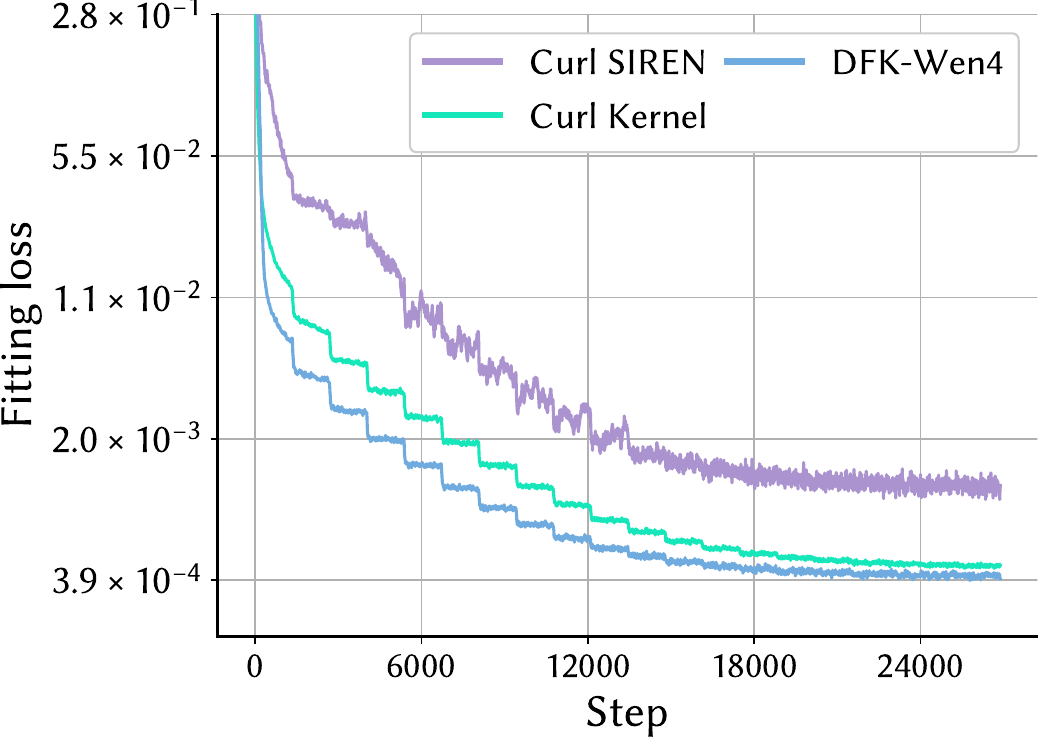}\vspace{-1em}}
    \Description{loss curves}
\end{figure}

\paragraph{Missile (2D)}
Curl SIREN consists of three hidden layers with 128 neurons each, totaling \num{33537} trainable parameters. 
DFK-Poly6 comprises \num{5390} points, resulting in \num{26950} trainable parameters.
Curl Kernel comprises \num{6797} points, resulting in \num{27188} trainable parameters.
DFK-Wen4 comprises \num{5390} points, resulting in \num{26950} trainable parameters.

We set the batch size to \num{128} and trained each model for \num{20} epochs, with an initial learning rate of \num{1e-3}. When initializing the kernel radii, we set $\eta=9$. The loss curves are illustrated in Fig.~\ref{fig:missile}.

\paragraph{Bullet (3D)}
Curl SIREN consists of four hidden layers with 256 neurons each, totaling \num{199171} trainable parameters. 
Curl Kernel comprises \num{25325} points, resulting in \num{177275} trainable parameters.
DFK-Wen4 comprises \num{25325} points, resulting in \num{177275} trainable parameters.

We set the batch size to \num{4096} and trained each model for \num{20} epochs, with an initial learning rate of \num{1e-3}. When initializing the kernel radii, we set $\eta=6$. The loss curves are illustrated in Fig.~\ref{fig:bullet}.

\subsection{Super-Resolution}

\begin{figure}[ht]
    \centering
    \subcaptionbox{\emph{Turbulence A (2D)}.\label{fig:multisrc}}
    {\includegraphics[width=0.325\textwidth]{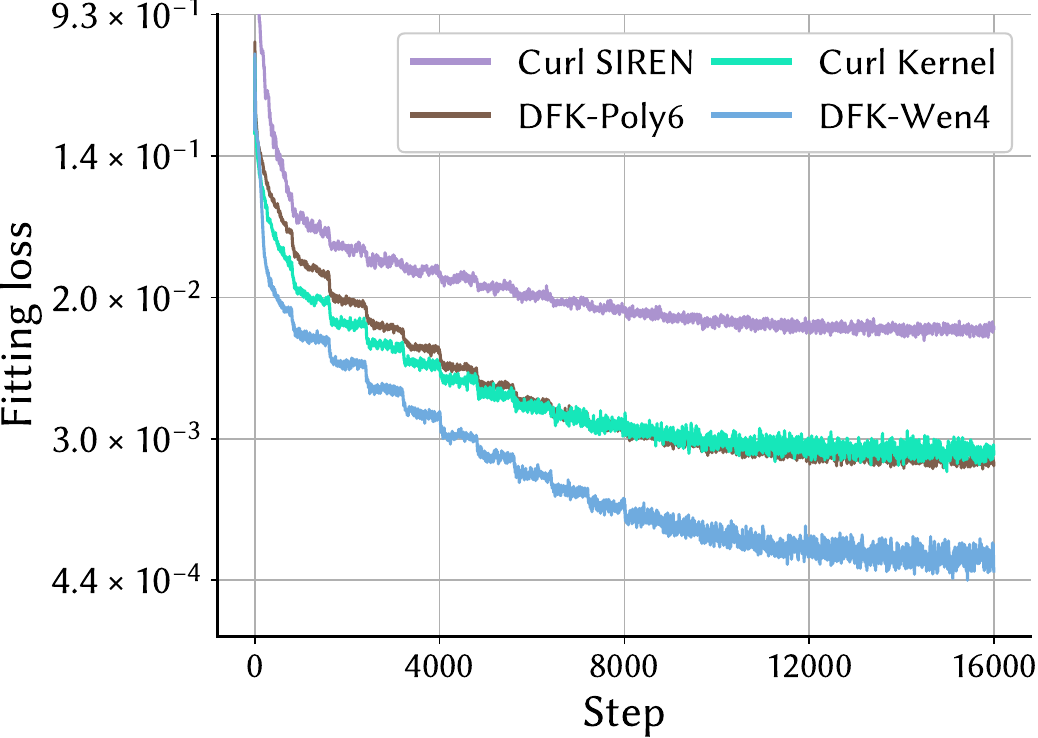}}
    \hfill
    \subcaptionbox{\emph{Turbulence B (2D)}.\label{fig:mixing}}
    {\includegraphics[width=0.325\textwidth]{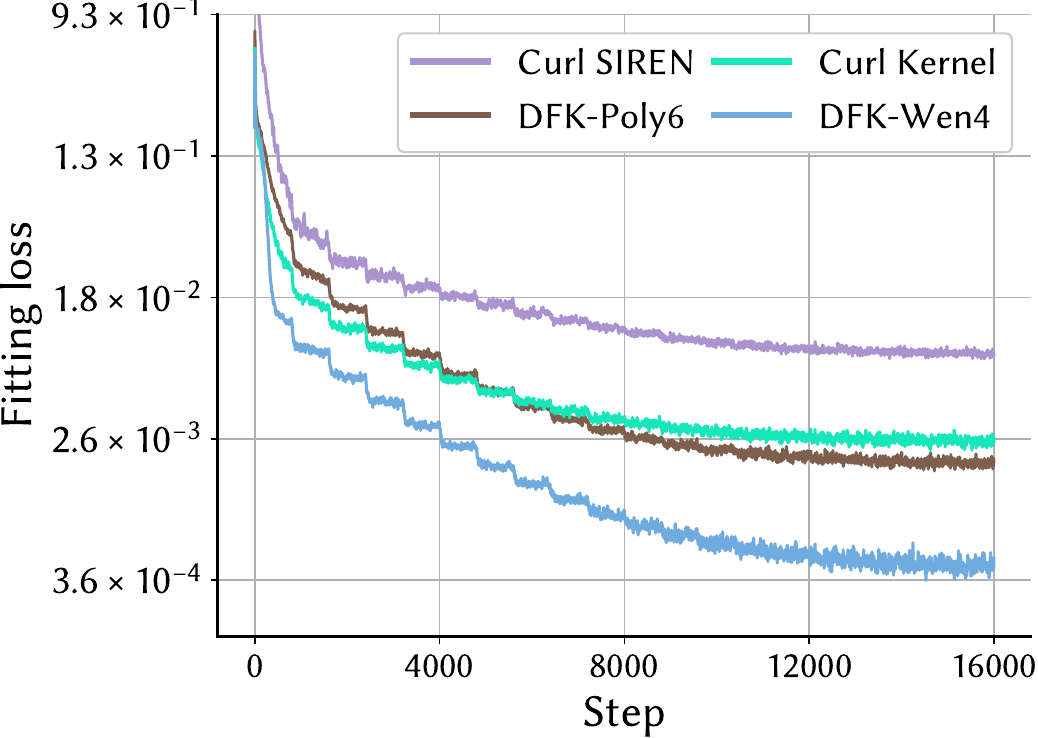}}
    \hfill
    \subcaptionbox{\emph{Spherical obstacle (3D)}.\label{fig:obsplume}}
    {\includegraphics[width=0.325\textwidth]{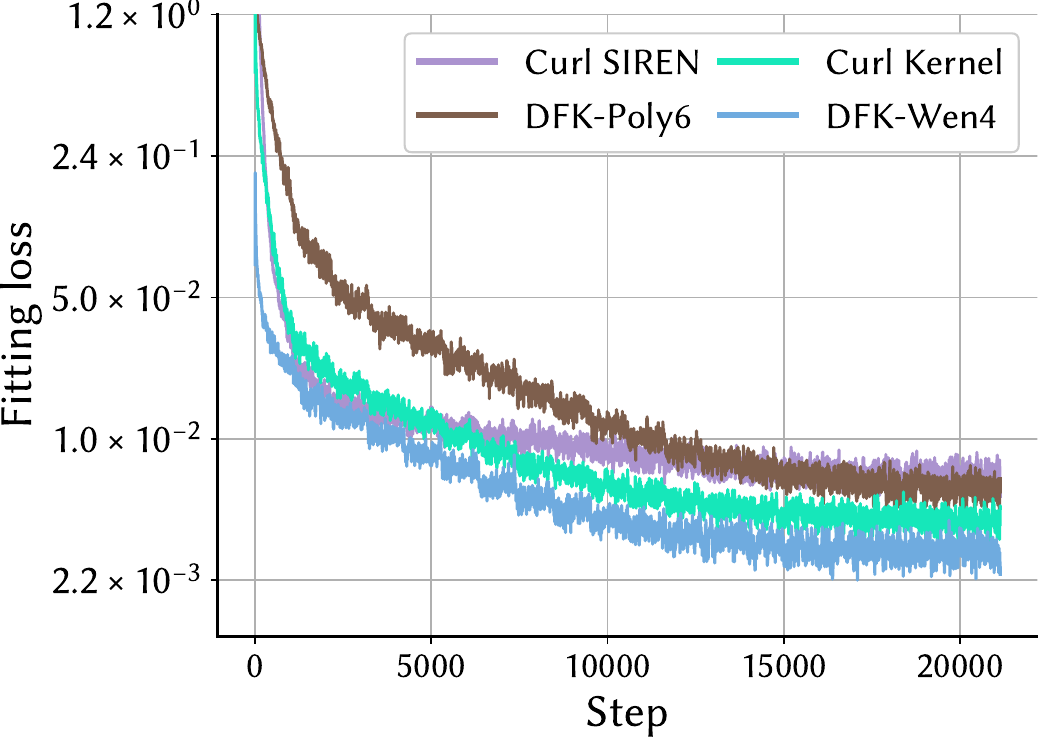}}
    \\
    \vspace{-1em}
    \caption{Loss curves of super-resolution experiments.}
    \Description{loss curves}
\end{figure}

\paragraph{Turbulence A (2D)}
Curl SIREN consists of four hidden layers with 256 neurons each, totaling \num{198401} trainable parameters. 
DFK-Poly6 comprises \num{5416} points, resulting in \num{27080} trainable parameters.
Curl Kernel comprises \num{6834} points, resulting in \num{27336} trainable parameters.
DFK-Wen4 comprises \num{5416} points, resulting in \num{27080} trainable parameters.

We set the batch size to \num{128} and trained each model for \num{20} epochs, with an initial learning rate of \num{1e-3}. When initializing the kernel radii, we set $\eta=9$.  The loss curves are illustrated in Fig.~\ref{fig:multisrc}.

\paragraph{Turbulence B (2D)}
Curl SIREN consists of four hidden layers with 256 neurons each, totaling \num{198401} trainable parameters. 
DFK-Poly6 comprises \num{5416} points, resulting in \num{27080} trainable parameters.
Curl Kernel comprises \num{6834} points, resulting in \num{27336} trainable parameters.
DFK-Wen4 comprises \num{5416} points, resulting in \num{27080} trainable parameters.

We set the batch size to \num{128} and trained each model for \num{20} epochs, with an initial learning rate of \num{1e-3}. When initializing the kernel radii, we set $\eta=9$.  The loss curves are illustrated in Fig.~\ref{fig:mixing}.

\paragraph{Spherical obstacle (3D)}
Curl SIREN consists of six hidden layers with 256 neurons each, totaling \num{330755} trainable parameters. 
DFK-Poly6 comprises \num{42002} points, resulting in \num{294014} trainable parameters.
Curl Kernel comprises \num{42002} points, resulting in \num{294014} trainable parameters.
DFK-Wen4 comprises \num{42002} points, resulting in \num{294014} trainable parameters.

We set the batch size to \num{128} and trained each model for \num{20} epochs, with an initial learning rate of \num{1e-3}. When initializing the kernel radii, we set $\eta=12$.  The loss curves are illustrated in Fig.~\ref{fig:obsplume}.

\subsection{Inference}

\begin{figure}[ht]
    \centering
    \includegraphics[width=0.245\textwidth]{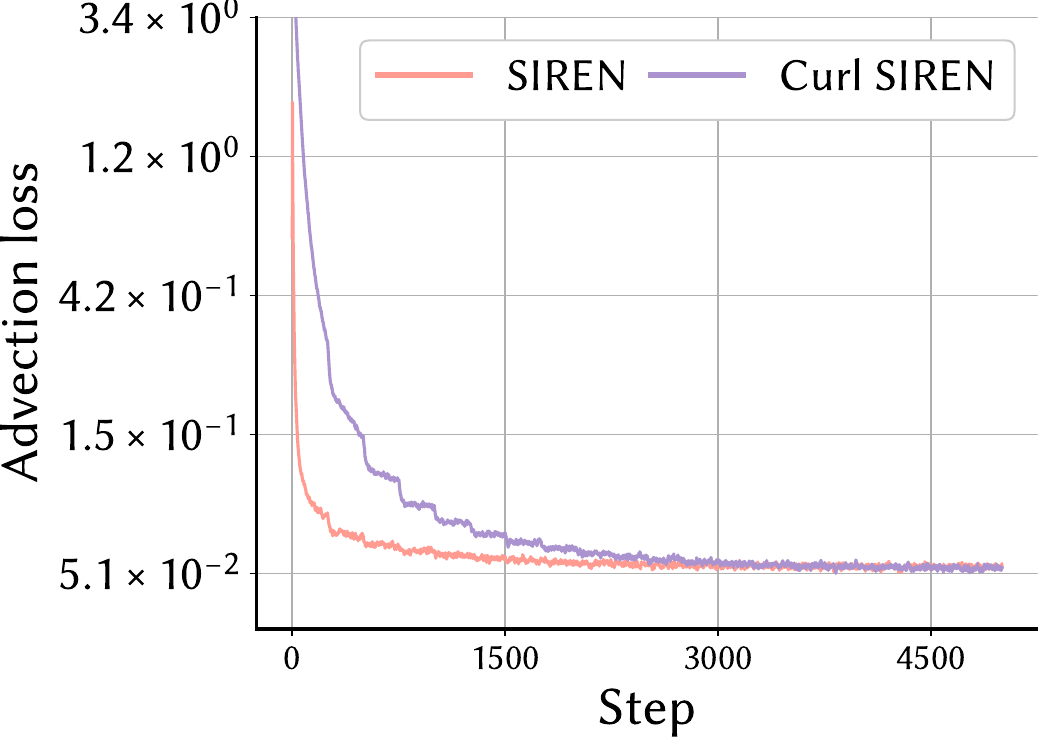}
    \hfill
    \includegraphics[width=0.245\textwidth]{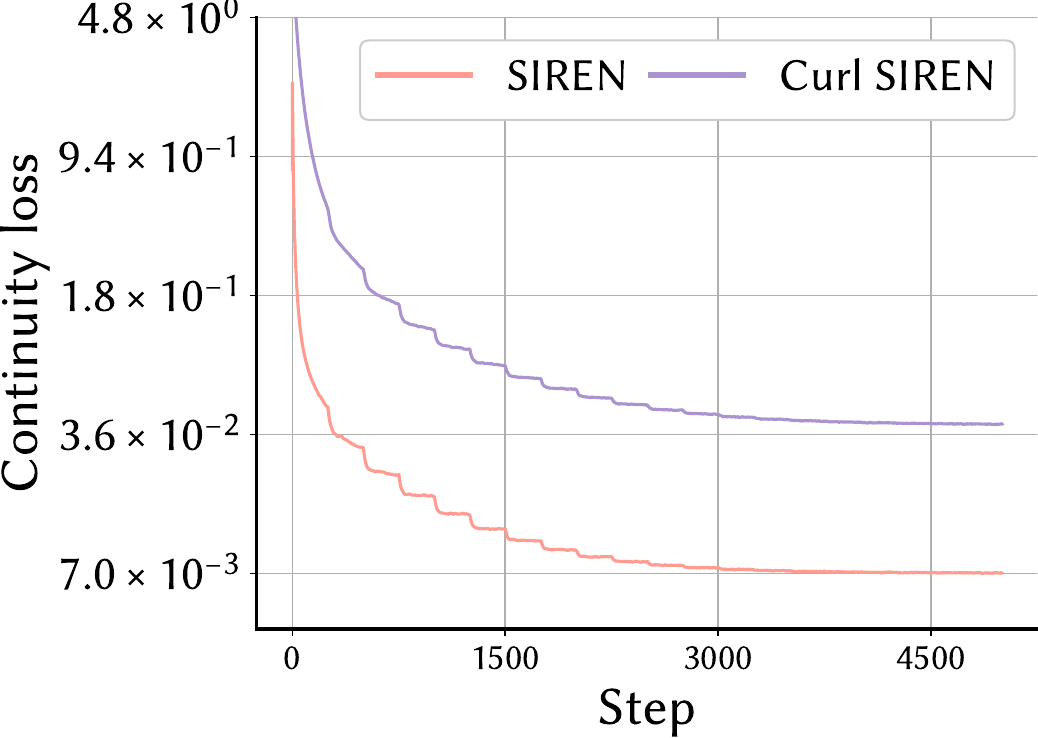}
    \hfill
    \includegraphics[width=0.245\textwidth]{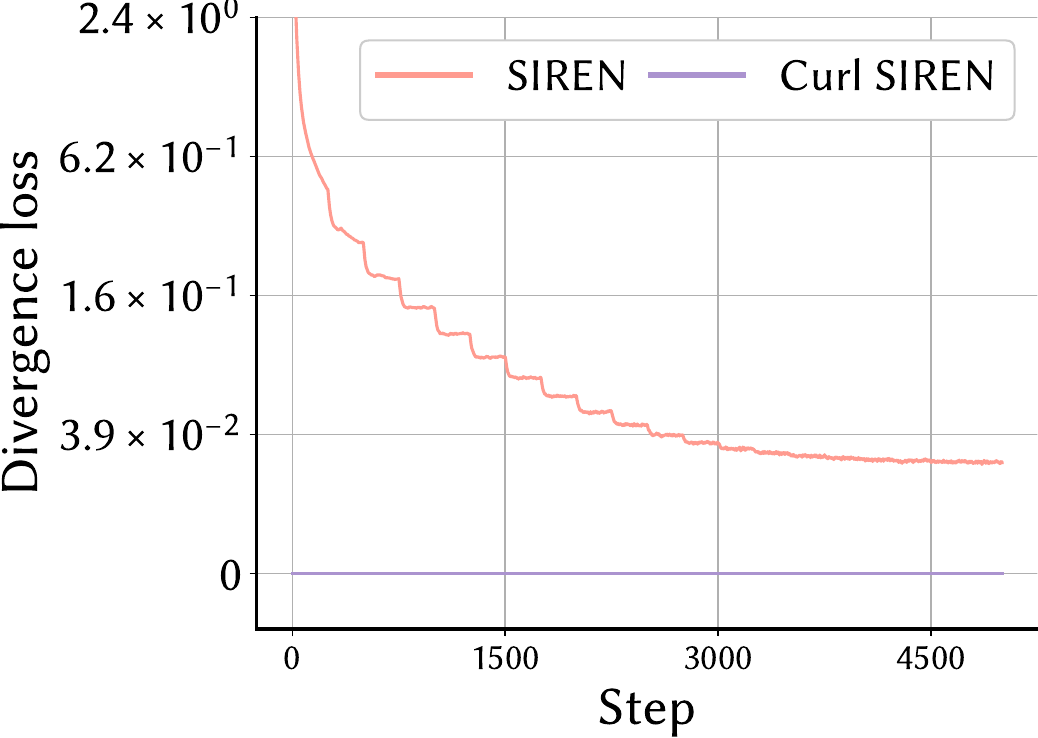}
    \hfill
    \includegraphics[width=0.245\textwidth]{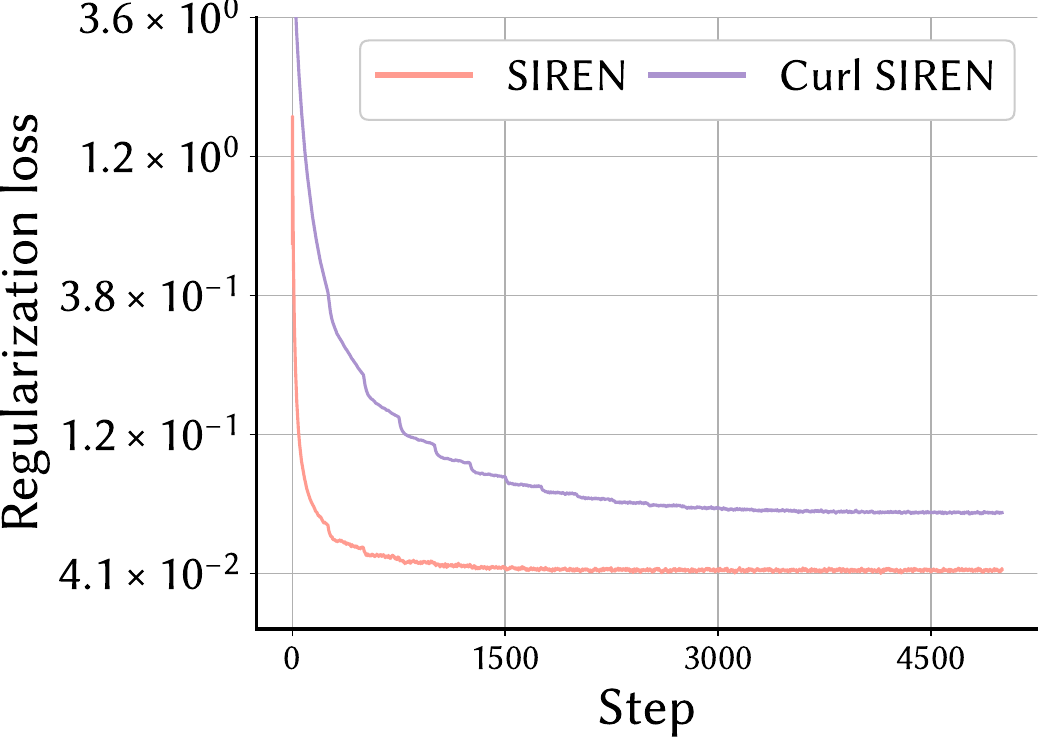}
    \\
    \includegraphics[width=0.245\textwidth]{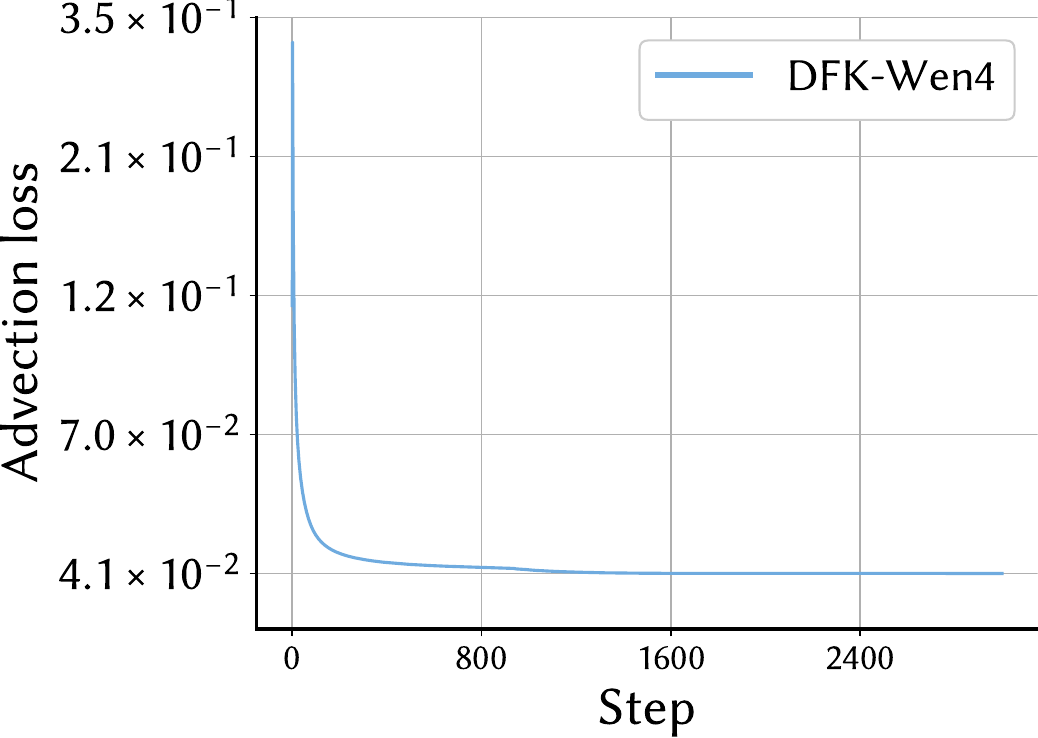}
    \hfill
    \includegraphics[width=0.245\textwidth]{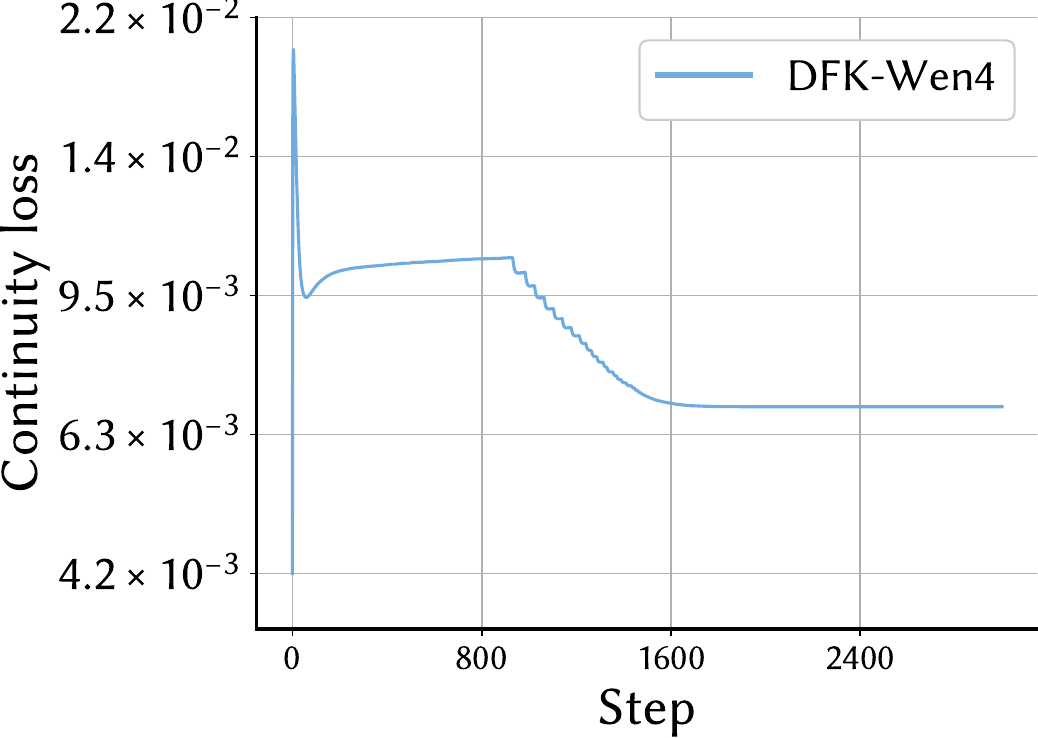}
    \hfill
    \includegraphics[width=0.245\textwidth]{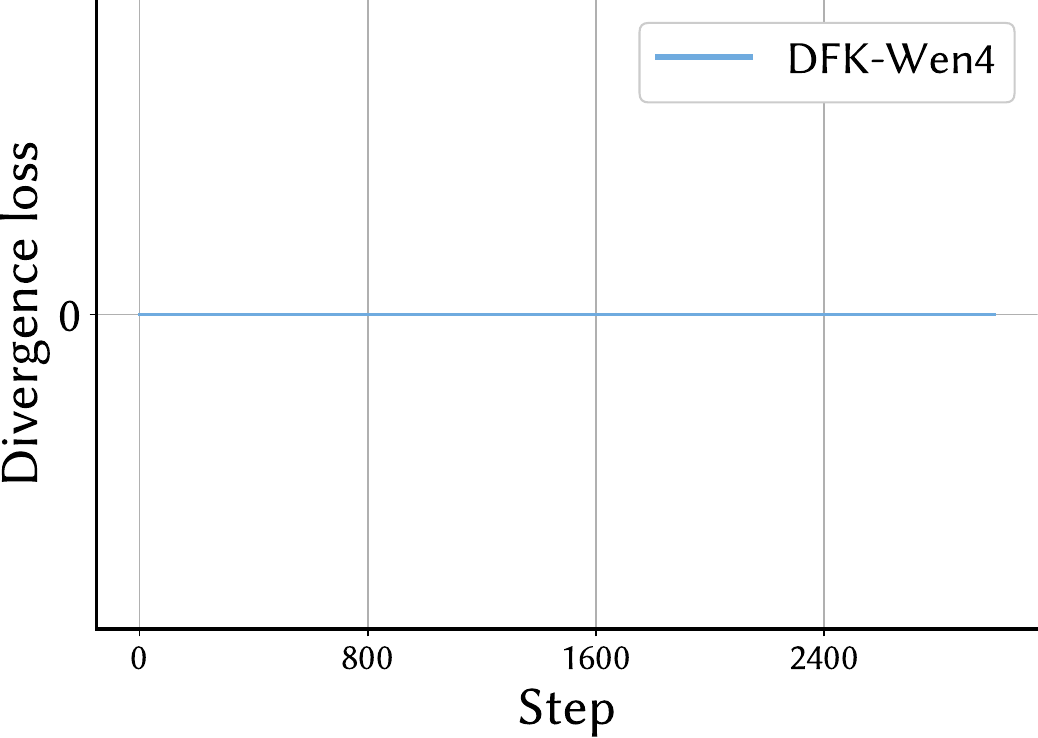}
    \hfill
    \includegraphics[width=0.245\textwidth]{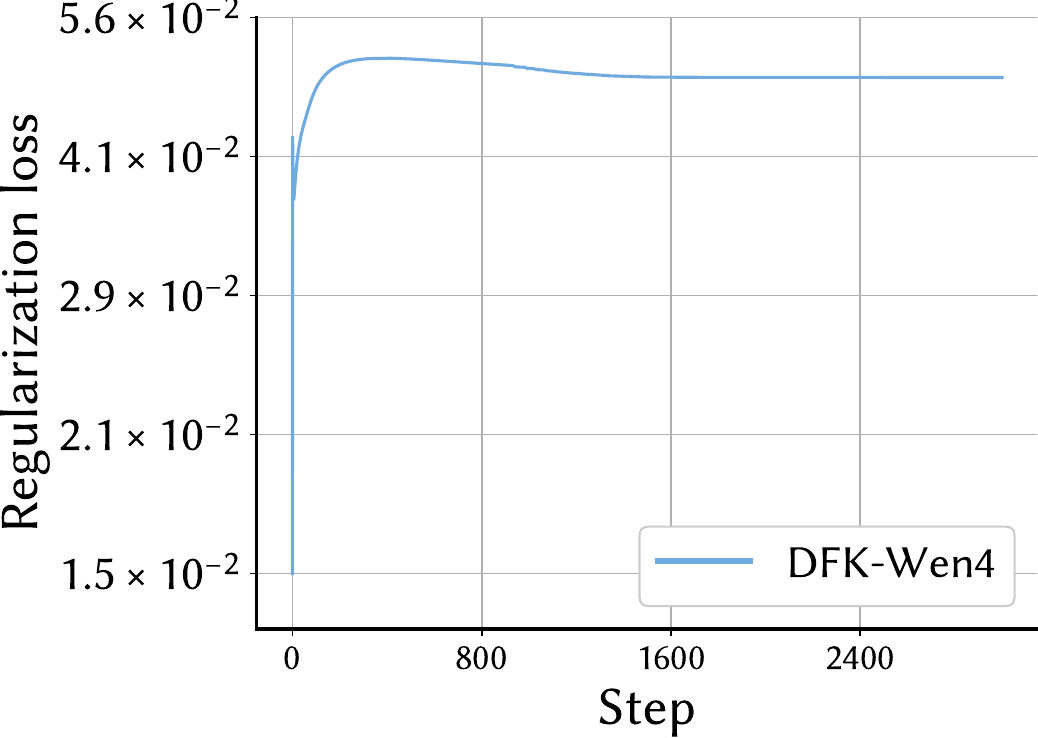}
    \\
    \vspace{-1em}
    \caption{Loss curves of inference experiments for \emph{rising (3D)}.}
    \label{fig:rising}
    \Description{loss curves}
\end{figure}

\paragraph{Rising (3D)}
SIREN consists of four hidden layers with 256 neurons each, totaling \num{29676479} trainable parameters. 
Curl SIREN consists of four hidden layers with 256 neurons each, totaling \num{29676479} trainable parameters. 
DFK-Wen4 comprises \num{26375} points, resulting in \num{11868750} trainable parameters.

We set the batch size to $\num{3072}\times\num{151}$ and trained each NN model for \num{20} epochs, with an initial learning rate of \num{1e-3}.
We trained DFKs-Wen4 using batch gradient descent for \num{1750} epochs, with an initial learning rate of \num{1e-2}. When initializing the kernel radii, we set $\eta=6$. The loss curves are illustrated in Fig.~\ref{fig:rising}.

\begin{figure}[ht]
    \centering
    \includegraphics[width=0.245\textwidth]{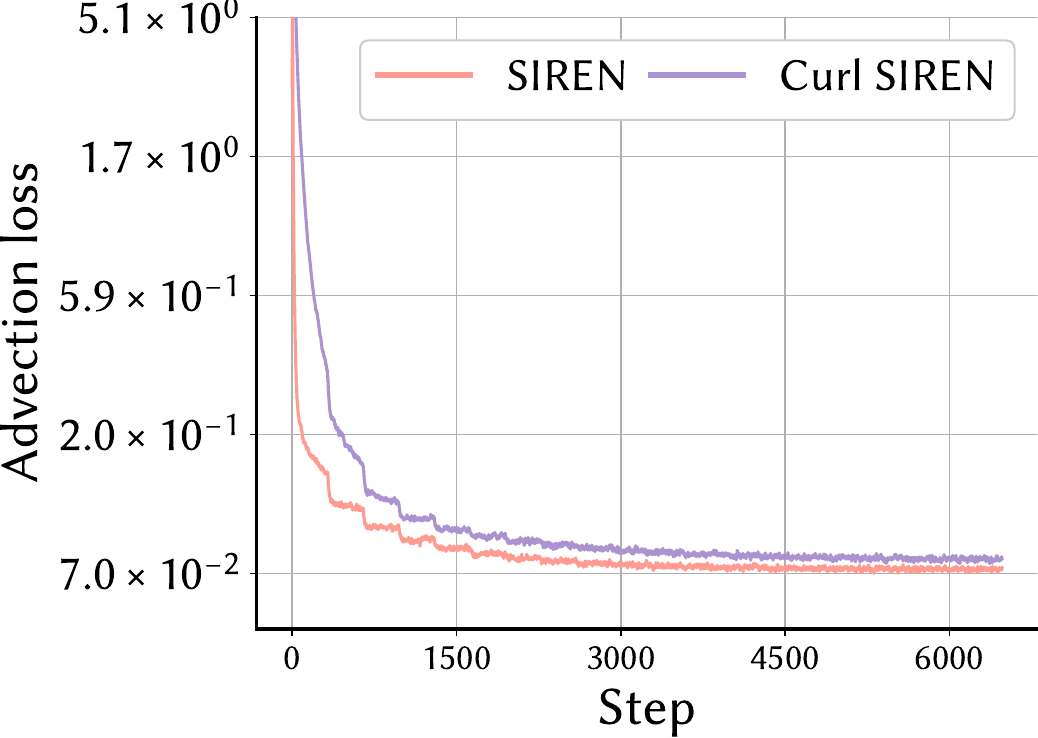}
    \hfill
    \includegraphics[width=0.245\textwidth]{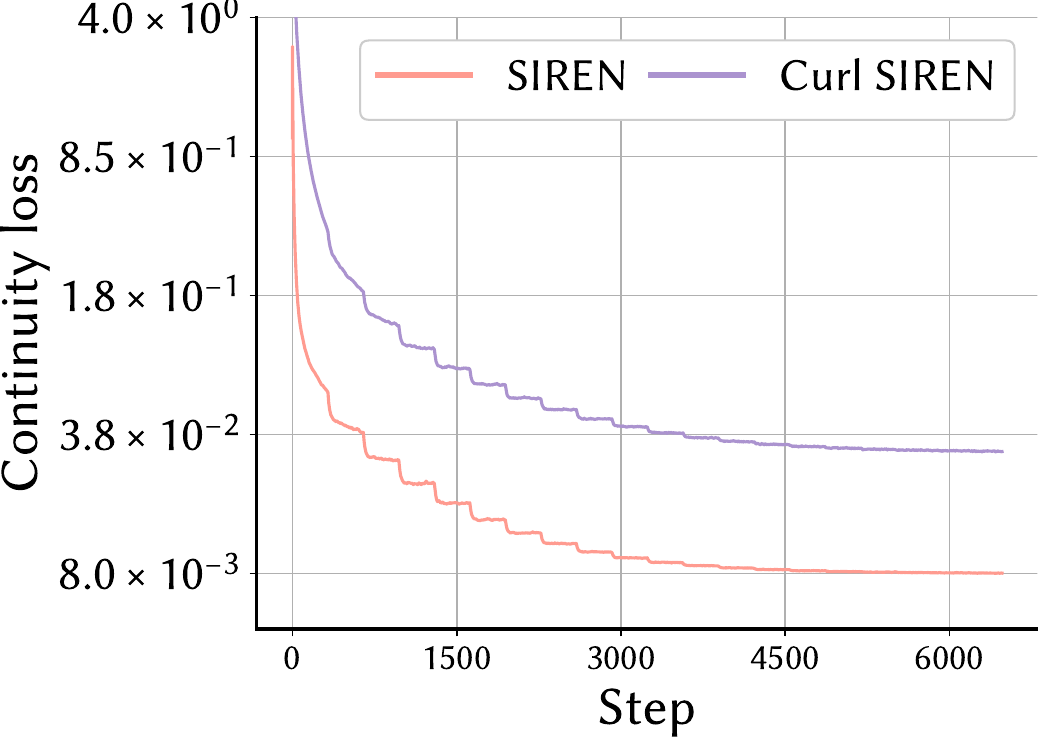}
    \hfill
    \includegraphics[width=0.245\textwidth]{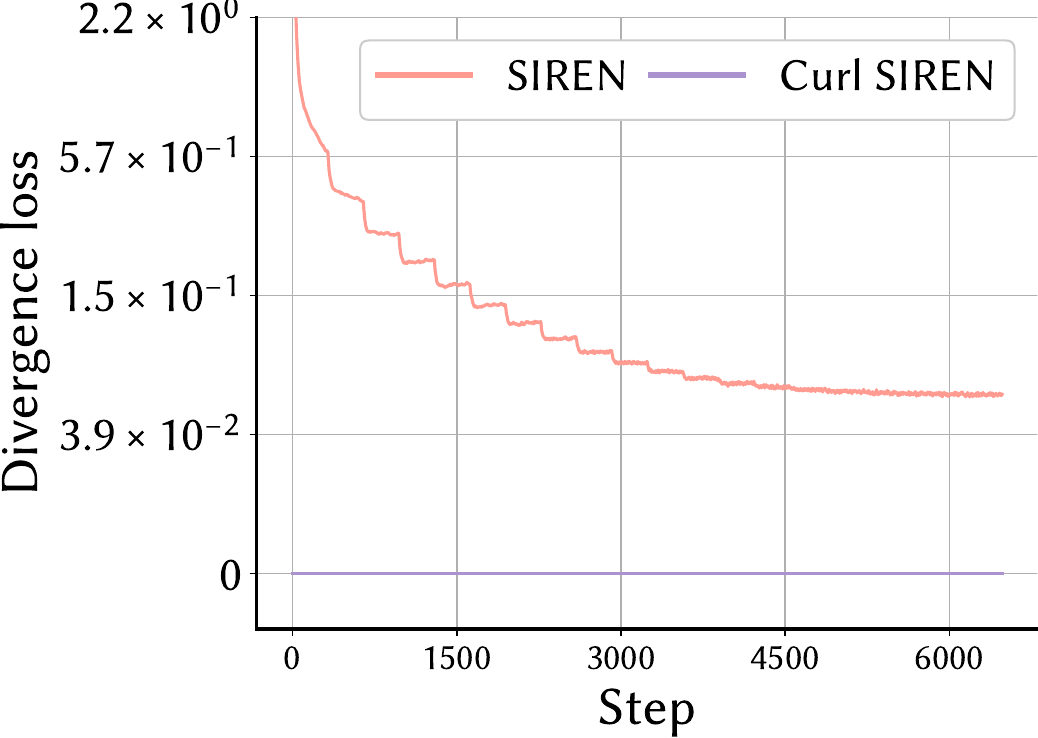}
    \hfill
    \includegraphics[width=0.245\textwidth]{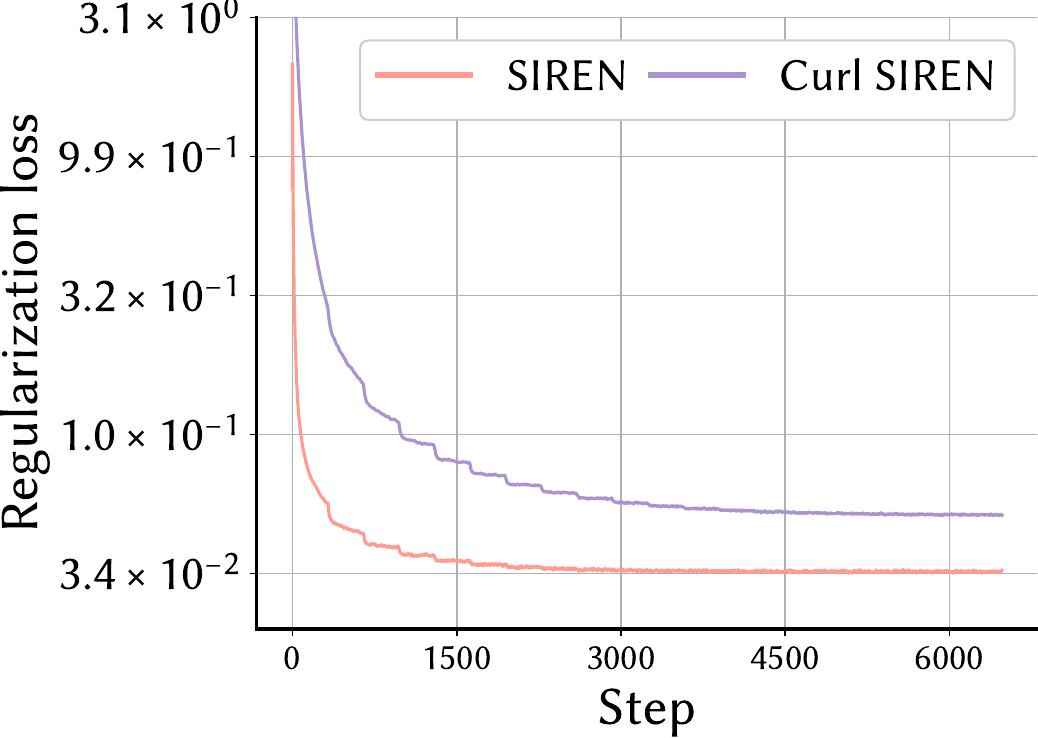}
    \\
    \includegraphics[width=0.245\textwidth]{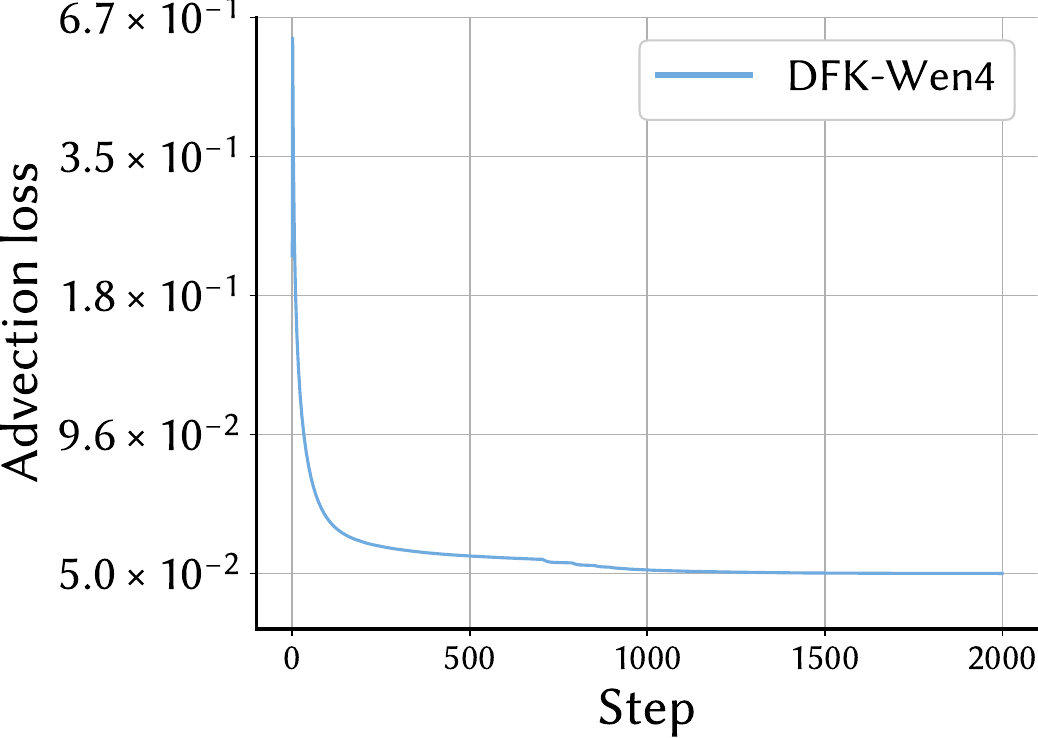}
    \hfill
    \includegraphics[width=0.245\textwidth]{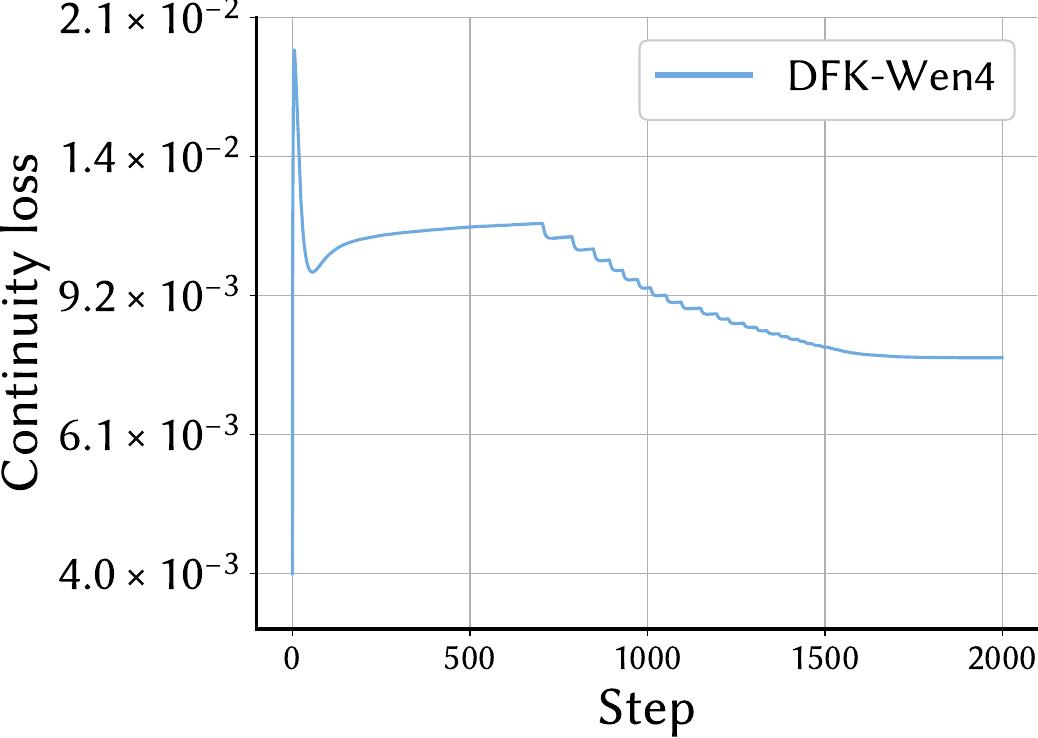}
    \hfill
    \includegraphics[width=0.245\textwidth]{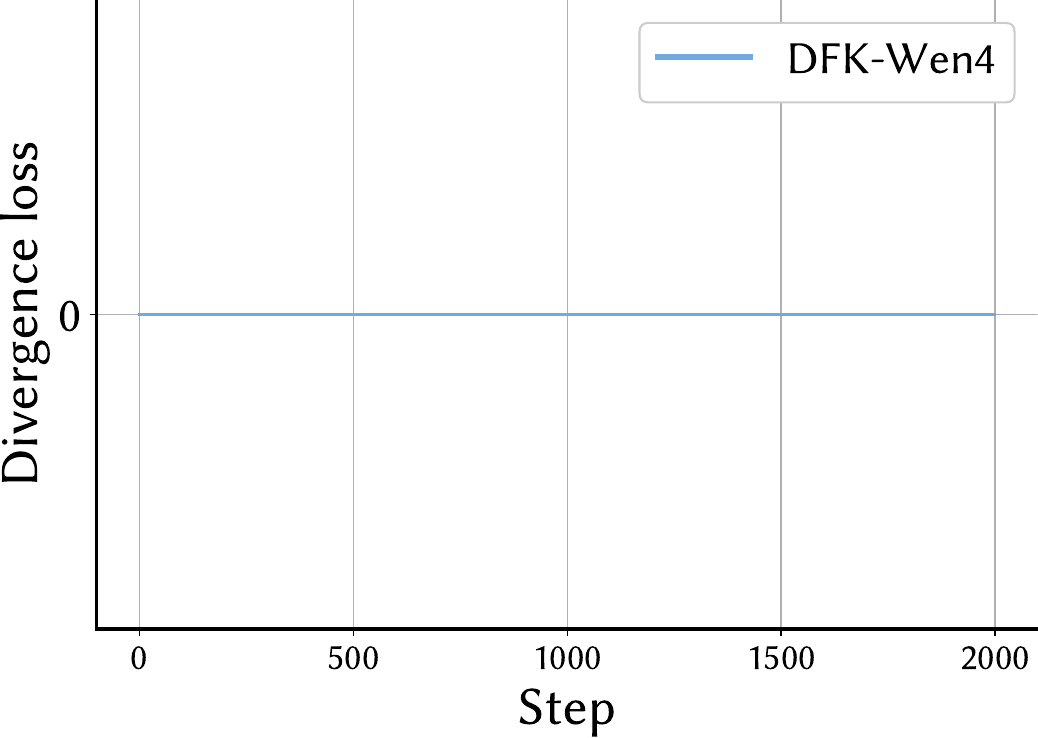}
    \hfill
    \includegraphics[width=0.245\textwidth]{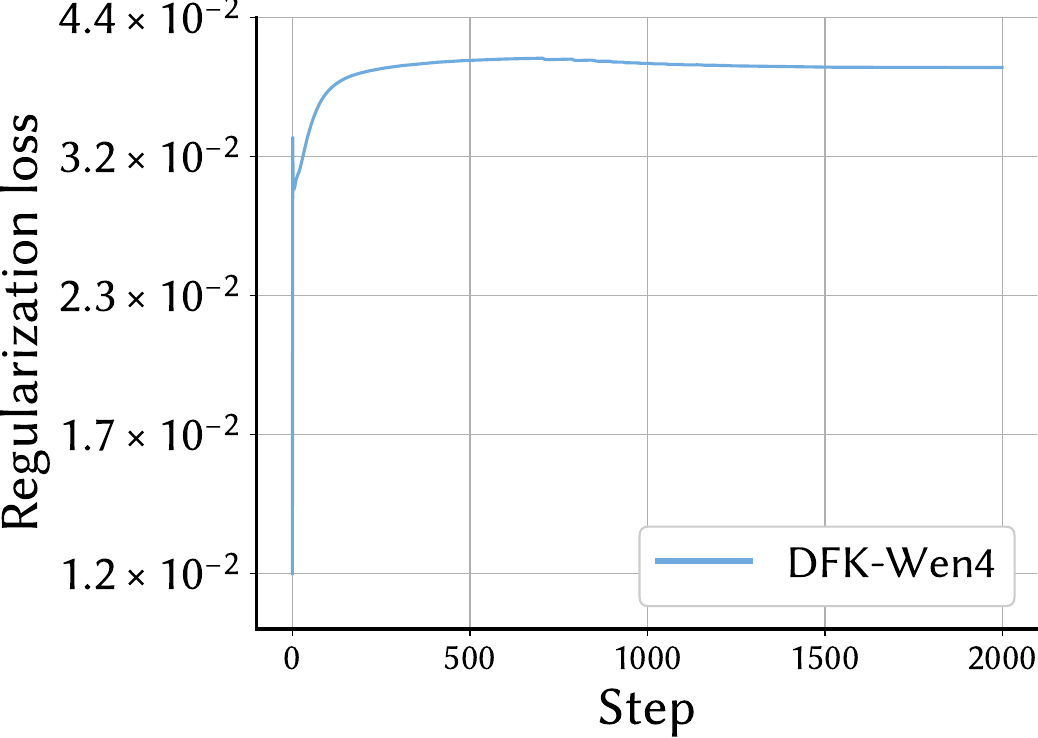}
    \\
    \vspace{-1em}
    \caption{Loss curves of inference experiments for \emph{teapot (3D)}.}
    \label{fig:teapot}
    \Description{loss curves}
\end{figure}

\paragraph{Teapot (3D)}
SIREN consists of four hidden layers with 256 neurons each, totaling \num{19917100} trainable parameters. 
Curl SIREN consists of four hidden layers with 256 neurons each, totaling \num{19917100} trainable parameters. 
DFK-Wen4 comprises \num{42053} points, resulting in \num{12615900} trainable parameters.

We set the batch size to $\num{4096}\times\num{101}$ and trained each NN model for \num{20} epochs, with an initial learning rate of \num{1e-3}.
We trained DFKs-Wen4 using batch gradient descent for \num{2000} epochs, with an initial learning rate of \num{1e-2}. When initializing the kernel radii, we set $\eta=6$. The loss curves are illustrated in Fig.~\ref{fig:teapot}.

\begin{figure}[ht]
    \centering
    \includegraphics[width=0.245\textwidth]{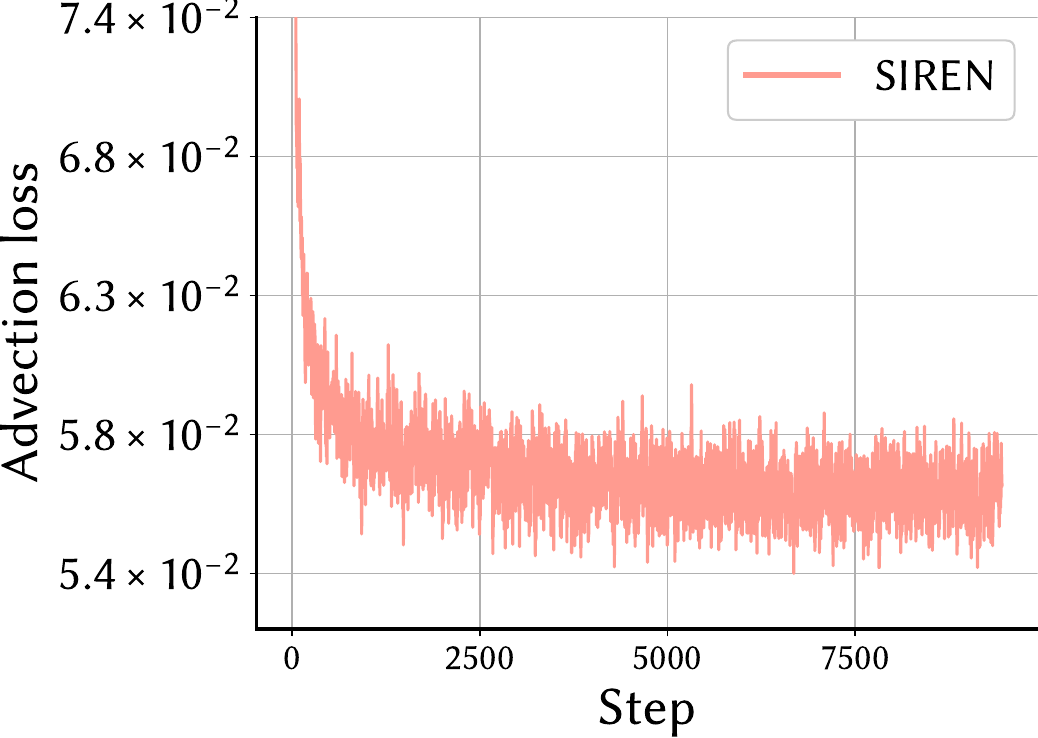}
    \hfill
    \includegraphics[width=0.245\textwidth]{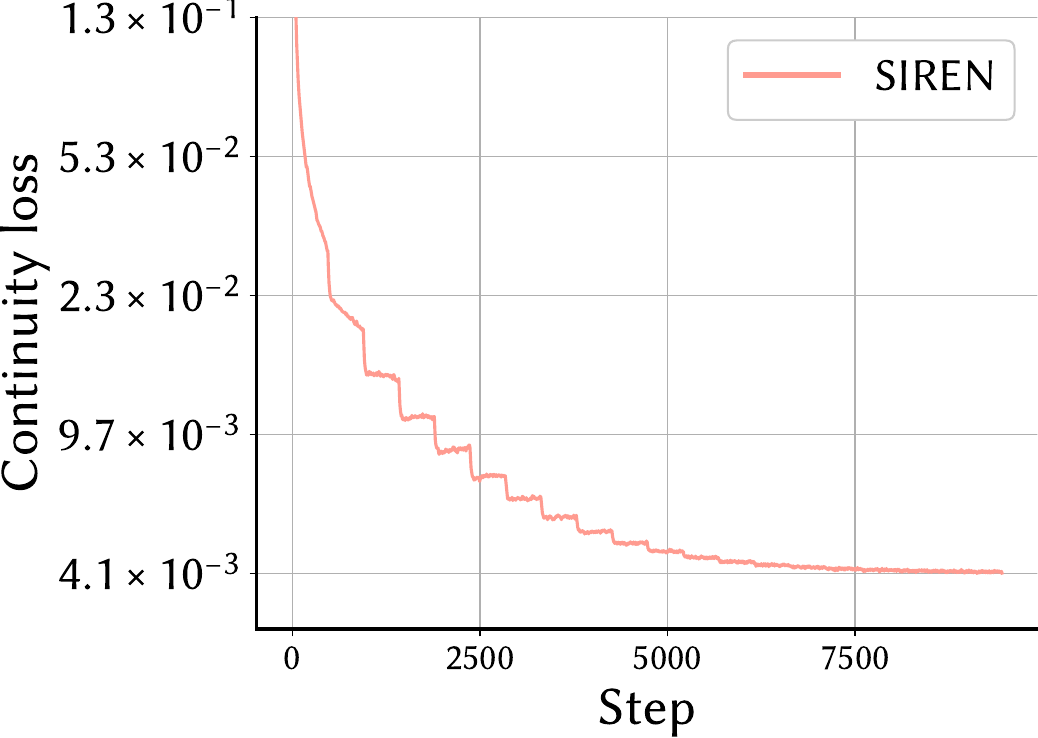}
    \hfill
    \includegraphics[width=0.245\textwidth]{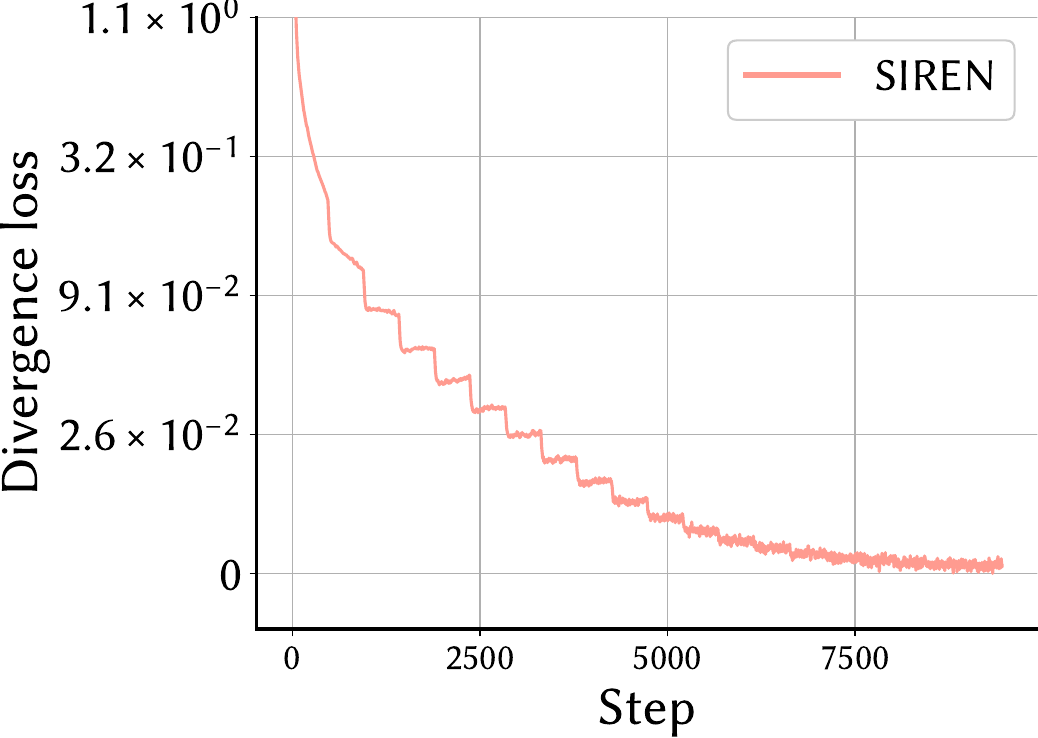}
    \hfill
    \includegraphics[width=0.245\textwidth]{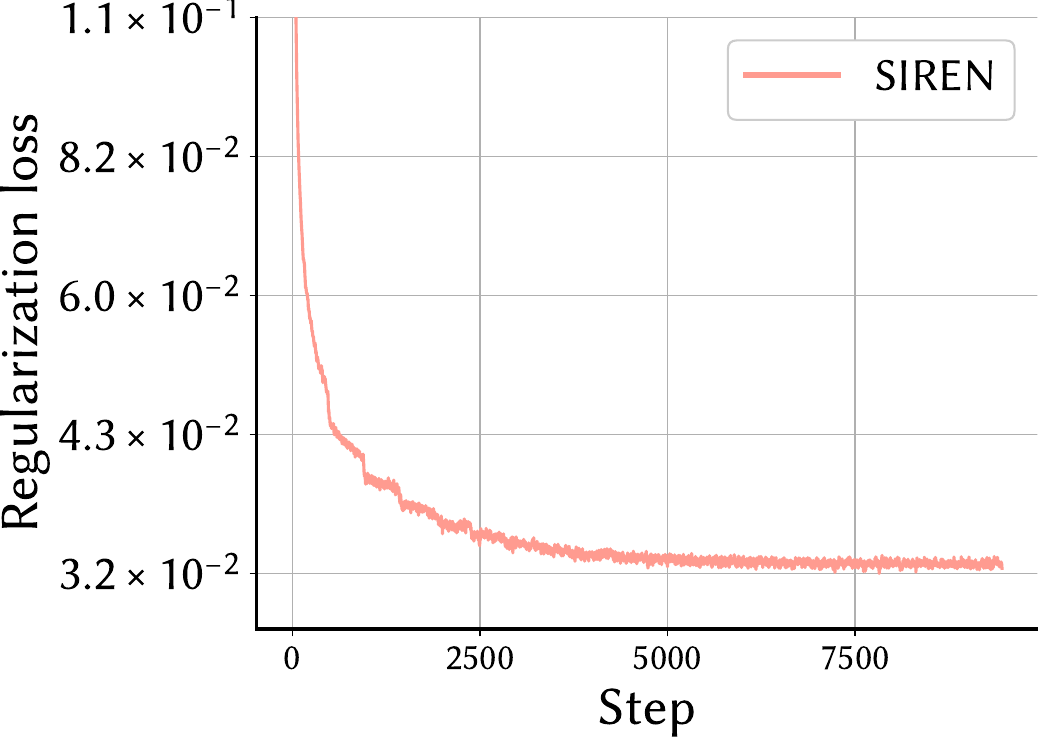}
    \\
    \includegraphics[width=0.245\textwidth]{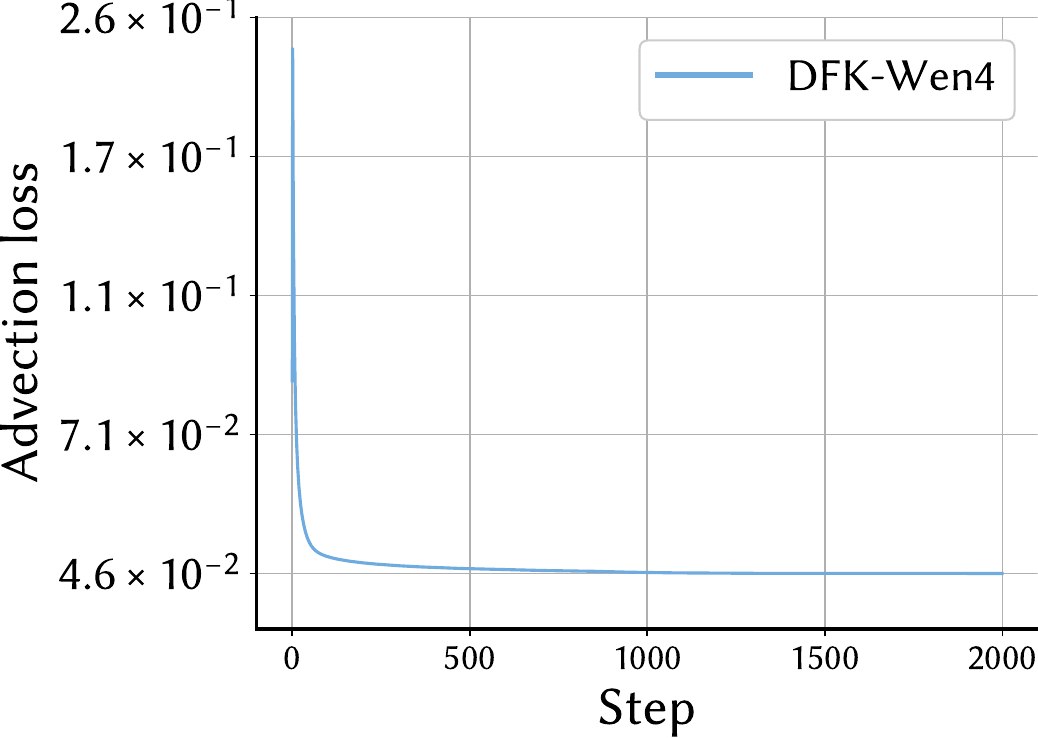}
    \hfill
    \includegraphics[width=0.245\textwidth]{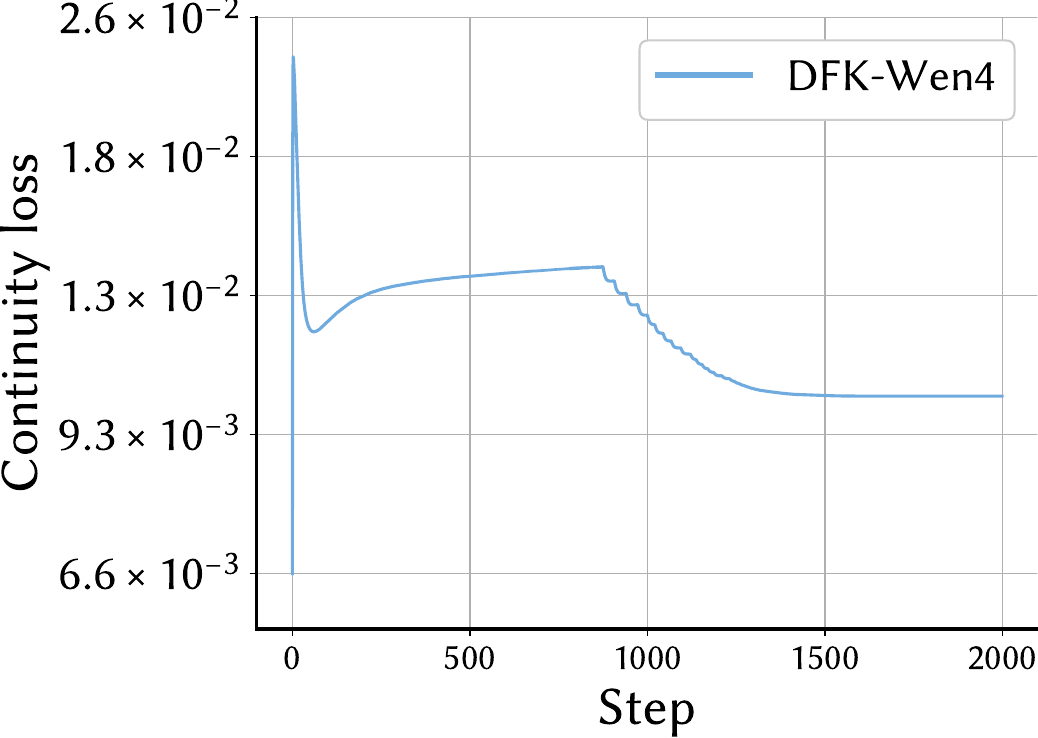}
    \hfill
    \includegraphics[width=0.245\textwidth]{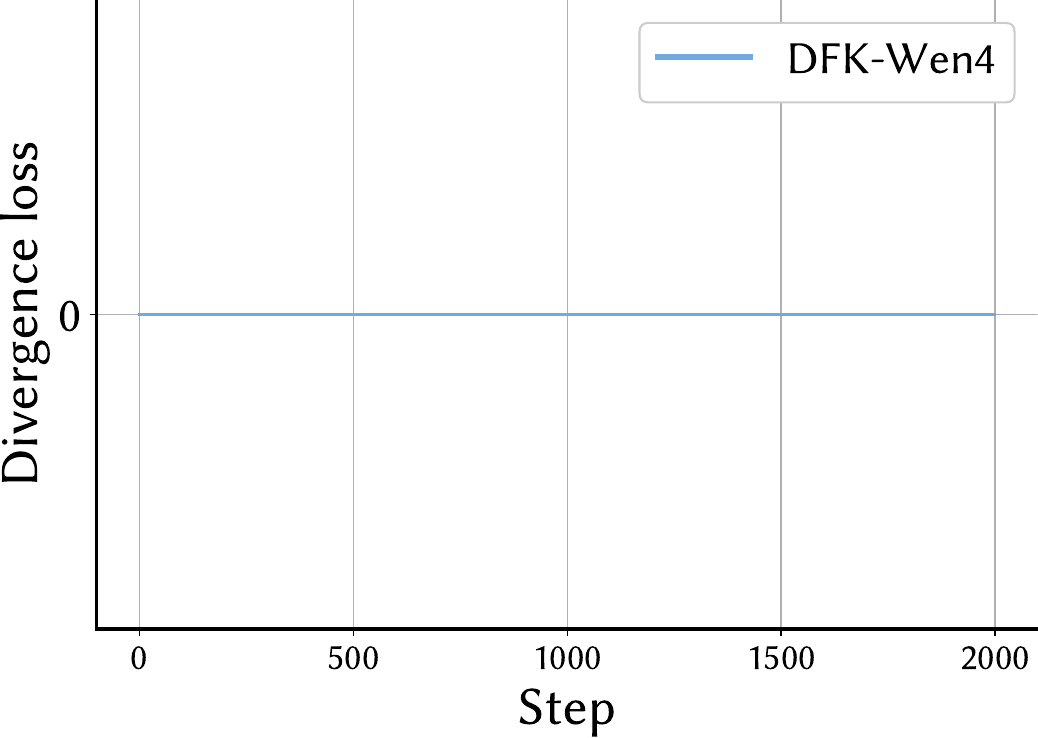}
    \hfill
    \includegraphics[width=0.245\textwidth]{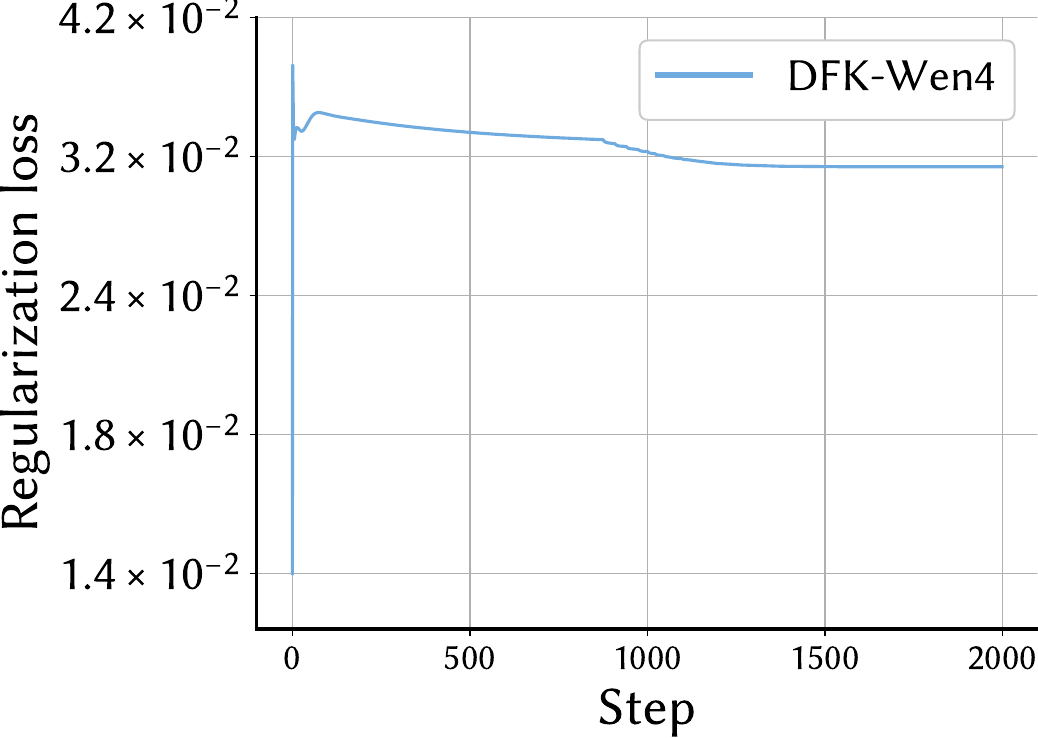}
    \\
    \vspace{-1em}
    \caption{Loss curves of inference experiments for \emph{scalar flow (3D)}.}
    \label{fig:scalar}
    \Description{loss curves}
\end{figure}

\paragraph{Scalar flow (3D)}
SIREN consists of four hidden layers with 256 neurons each, totaling \num{23701349} trainable parameters. 
DFK-Wen4 comprises \num{41959} points, resulting in \num{14979363} trainable parameters.

We set the batch size to $\num{4096}\times\num{120}$ and trained the NN model for \num{20} epochs, with an initial learning rate of \num{1e-3}.
We trained DFKs-Wen4 using batch gradient descent for \num{2000} epochs, with an initial learning rate of \num{1e-2}. When initializing the kernel radii, we set $\eta=6$. The loss curves are illustrated in Fig.~\ref{fig:scalar}.

\bibliographystyle{ACM-Reference-Format}
\bibliography{refs_suppl}